\documentclass{sadhana}
\usepackage{graphicx,amsmath,amssymb,verbatim,bm,xcolor}
\allowdisplaybreaks

\DeclareFontFamily{U}{mathx}{\hyphenchar\font45}
\DeclareFontShape{U}{mathx}{m}{n}{
      <5> <6> <7> <8> <9> <10>
      <10.95> <12> <14.4> <17.28> <20.74> <24.88>
      mathx10
      }{}
\DeclareSymbolFont{mathx}{U}{mathx}{m}{n}
\DeclareMathSymbol{\bigplus}{1}{mathx}{"90}
\DeclareMathSymbol{\bigtimes}{1}{mathx}{"91}

\newcounter{type1error}
\newcounter{type2errorx}
\newcounter{type2errory}
\newcounter{type2errorxy}

\newcounter{propqtypicalstate}
\newcounter{propqtypicalpovm}
\newcounter{propqtypicalellone}
\newcounter{propqtypicalellinfty}
\newcounter{propqtypicaldistance}
\newcounter{propqtypicalcompleteness}
\newcounter{propqtypicalsplitting}
\newcounter{propqtypicalsoundness}

\newcounter{corqtypicalcq}
\newcounter{corqtypicaldistance}
\newcounter{corqtypicalcompleteness1}
\newcounter{corqtypicalsoundness1}

\newcounter{lemqtypicalcq}
\newcounter{lemqtypicaldistance}
\newcounter{lemqtypicalcompleteness1}
\newcounter{lemqtypicalsoundness1}

\newcounter{thmqtypicalcq}
\newcounter{thmqtypicaldistance}
\newcounter{thmqtypicalcompleteness1}
\newcounter{thmqtypicalsoundness1}

\AtEndDocument{
\addtocounter{propqtypicalstate}{-1}
\addtocounter{propqtypicalpovm}{-1}
\addtocounter{propqtypicalellone}{-1}
\addtocounter{propqtypicalellinfty}{-1}
\addtocounter{propqtypicaldistance}{-1}
\addtocounter{propqtypicalcompleteness}{-1}
\addtocounter{propqtypicalsplitting}{-1}
\addtocounter{propqtypicalsoundness}{-1}
\refstepcounter{propqtypicalstate}\label{prop@qtypicalstate}
\refstepcounter{propqtypicalpovm}\label{prop@qtypicalpovm}
\refstepcounter{propqtypicalellone}\label{prop@qtypicalellone}
\refstepcounter{propqtypicalellinfty}\label{prop@qtypicalellinfty}
\refstepcounter{propqtypicaldistance}\label{prop@qtypicaldistance}
\refstepcounter{propqtypicalcompleteness}\label{prop@qtypicalcompleteness}
\refstepcounter{propqtypicalsplitting}\label{prop@qtypicalsplitting}
\refstepcounter{propqtypicalsoundness}\label{prop@qtypicalsoundness}
}

\newtheorem{definition}{Definition}
\newtheorem{lemma}{Lemma}
\newtheorem{proposition}{Proposition}
\newtheorem{fact}{Fact}
\newtheorem{theorem}{Theorem}
\newtheorem{corollary}{Corollary}

\newenvironment{proof}{\textbf{Proof:}}{\hfill$\square$}

\newcommand{\vech}{\mathbf{h}}
\newcommand{\vecl}{\mathbf{l}}
\newcommand{\vecx}{\mathbf{x}}

\newcommand{\cC}{\mathcal{C}}
\newcommand{\cH}{\mathcal{H}}
\newcommand{\cK}{\mathcal{K}}
\newcommand{\cL}{\mathcal{L}}

\newcommand{\cT}{\mathcal{T}}
\newcommand{\cU}{\mathcal{U}}
\newcommand{\cX}{\mathcal{X}}
\newcommand{\cY}{\mathcal{Y}}
\newcommand{\cZ}{\mathcal{Z}}

\newcommand{\hcX}{\hat{\cX}}
\newcommand{\hcY}{\hat{\cY}}
\newcommand{\hcZ}{\hat{\cZ}}

\newcommand{\hA}{\hat{A}}
\newcommand{\hX}{\hat{X}}
\newcommand{\hY}{\hat{Y}}
\newcommand{\hZ}{\hat{Z}}

\newcommand{\bW}{\mathbf{W}}

\newcommand{\C}{\mathbb{C}}

\newcommand{\hPi}{\hat{\Pi}}

\newcommand{\tcH}{\tilde{\cH}}

\newcommand{\chan}{\mathfrak{C}}

\newcommand{\Hmin}{H_{\mathrm{min}}}

\DeclareMathOperator*{\E}{{\rm {\bf E}}\,}
\DeclareMathOperator*{\Tr}{{\rm Tr}\;}

\DeclareMathOperator*{\spanning}{{\rm span}\;}

\newcommand{\zero}{\leavevmode\hbox{\small l\kern-3.5pt\normalsize0}}
\newcommand{\one}{\leavevmode\hbox{\small1\kern-3.8pt\normalsize1}}

\newcommand{\cupdot}{\mathbin{\mathaccent\cdot\cup}}

\newcommand{\elltwo}[1]{\left\|{ #1 }\right\|_2}
\newcommand{\ellone}[1]{\left\|{ #1 }\right\|_1}
\newcommand{\ellinfty}[1]{\left\|{ #1 }\right\|_\infty}

\newcommand{\ket}[1]{| #1 \rangle}
\newcommand{\bra}[1]{\langle #1 |}
\newcommand{\ketbra}[1]{\ket{#1}\bra{#1}}
\newcommand{\braket}[2]{\langle {#1} \ket{#2}}

\begin{document}

\title{{\bf Unions, intersections and a 
one-shot quantum joint typicality lemma
}}

\author{Pranab Sen\textsuperscript{1,*}}
\affilOne{\textsuperscript{1}
School of Technology and Computer Science, Tata Institute of Fundamental
Research, Mumbai 400005, India.
Email: {\sf pranab.sen.73@gmail.com}
}

\twocolumn[{

\maketitle

\begin{abstract}
A fundamental tool to prove inner bounds in classical network 
information theory is the so-called `conditional joint typicality
lemma'. In addition to the lemma, one often uses unions and intersections
of typical sets in the inner bound arguments without so much as giving
them a second thought. These arguments do not work in the 
quantum setting. This bottleneck shows up in the fact that so-called
`simultaneous decoders', as opposed to `successive cancellation
decoders', are known for very few channels in quantum network information
theory. Another manifestation of this bottleneck is the  lack of
so-called `simultaneous smoothing' theorems for quantum states.

In this paper, we overcome 
the bottleneck by proving for the first time
a one-shot quantum joint typicality lemma with robust union 
and intersection
properties. To do so, we develop two novel tools in quantum information
theory which may be of independent interest. The first tool is 
a simple geometric idea called
{\em tilting}, which
increases the angles between a family of subspaces in 
orthogonal directions. 
The second tool, called {\em smoothing and augmentation}, 
is a way of perturbing a multipartite quantum state such that
the partial trace over any subset of registers does not increase the
operator norm by much.

Our joint typicality lemma allows us to construct simultaneous
quantum decoders for many multiterminal quantum channels. 
It provides a powerful tool to extend many 
results in classical network information theory to the one-shot quantum
setting. 
\end{abstract}



}]



\markboth{Pranab Sen}{Unions, intersections and a 
one-shot quantum joint typicality lemma}

\section{Introduction}
A fundamental tool to prove inner bounds for communication channels
in classical network information theory is the so-called conditional
joint typicality lemma \cite{book:elgamalkim}. Very often, 
the joint typicality lemma is
used together with implicit {\em intersection} and {\em union} arguments
in the inner bound proofs. This is especially so in the construction of
so-called {\em simultaneous} decoders (as opposed to successive 
cancellation decoders) for communication problems. In this 
paper,
we investigate what happens when one tries to extend the classical inner 
bound proofs to the quantum setting. It turns out that the union and
intersection arguments present a huge stumbling block in this effort. 
The main result of this paper is a one-shot quantum joint typicality
lemma that takes care of these union and intersection bottlenecks,
and allows us to extend many classical inner bound proofs to the 
quantum setting.

\subsection{Information theory and typical sequences}
\label{subsec:typicalseq}
In a path breaking paper \cite{Shannon:channelcoding}, Shannon started the
field of information theory. The paper introduced several revolutionary
ideas viz:
\begin{itemize}

\item
Abstracting out a noisy communication channel with one input and one 
output, henceforth referred to as a {\em point-to-point channel},
in a very general fashion as a probability transition
matrix that maps a symbol from its input alphabet, e.g. $\{0,1\}$, 
into a probability distribution over symbols from its output alphabet,
e.g. $\{0,1\}$;

\item
Showing that it is possible to reduce the error of transmitting messages
over a noisy channel by introducing some redundancy by {\em encoding} 
the message first, inputting the encoded message to the channel, and then 
{\em decoding} the output
in an attempt to recover the sent message;

\item
Shifting the onus of reliable information transmission away from the
engineering approach of reducing the noise in a channel, which was
the one emphasised earlier, to the mathematical approach of constructing
better codes so that higher communication rate and lower error rate  
could simultaneously be achieved using the physical same channel as
before;

\item
Showing that with many independent and identical uses of the
channel, it is possible to simultaneously achieve larger rates 
of information transmission per channel use as well as smaller
information transmission error;

\item
Defining a precise mathematical quantity called {\em channel capacity}
which is purely a function of the channel's probability transition
matrix that has the the property that, in the limit of many independent
and identical uses of the channel, almost error free transmission
of messages at any rate below the capacity is possible, and any
strategy for transmitting at a rate larger than the capacity necessarily
incurs non-vanishing error. 
\end{itemize}

The last item above, viz. the channel capacity, was defined in terms
of a quantity called {\em mutual information}, which in turn 
was defined in terms of a novel quantity 
{\em information-theoretic entropy} inspired by the concept
of thermodynamic entropy of physics. Shannon's {\em noisy channel coding
theorem} that proved the above property regarding channel capacity 
again introduced a new concept called {\em typical sequence} that
formalised the famous law of large numbers from probability theory
in the following particular way: in any sequence of samples of length
$n$ drawn independently from a probability distribution $p(\cdot)$ on 
a fixed finite alphabet $\cX$, with high probability, the number of
occurences in the sequence of every symbol $x \in \cX$ is roughly
$n p(x)$. Such sequences are said to be typical with respect to the
probability distribution $p(x)$.

The concept of a typical sequence naturally leads to the concept
of a {\em jointly typical sequence} with respect to a probability 
distribution $p(\cdot)$ on a fixed finite set $\cX \times \cY$.
Thus, a sequence 
$((x_1, y_1), \ldots, (x_n, y_n))$ of length $n$ is said to be 
$\delta$-jointly typical with respect to $p(x,y)$ iff the number
of occurences of each symbol $(x, y) \in \cX \times \cY$ is in the
range $n p(x,y) (1 \pm \delta)$. Often proofs in Shannon theory 
investigate questions like,
e.g., 
\begin{quote}
Consider a sequence 
$((x_1, y_1), \ldots, (x_n, y_n))$ drawn independently from the
probability distribution $p(x) p(y)$, which is nothing but the product
of the marginal distributions on $\cX$ and $\cY$. 
What is the probability that it is $\delta$-jointly typical
with respect to $p(x,y)$? 
\end{quote}
The answers to these types of questions are
provided by so-called {\em joint typicality lemmas}. They form a
cornerstone of techniques to prove various inner bounds, 
i.e. achievability results, in
channel coding. Often these proofs require one to take intersections
and unions of sets of typical sequences. This is especially so when
one studies inner bounds for {\em multiterminal channels}, i.e.
channels with more than one input and/or more than one output. 

As long as one is dealing only with {\em classical channels} i.e. channels 
where information is carried by properties of physical systems behaving 
according to non-quantum physics, Shannon's mathematical abstraction
of channels as explained above works without any problem. Typical sets
and intersections and unions provide a well oiled machinery to attack
channel coding problems in the limit of many independent and identical
uses of the channel, called {\em asymptotic iid setting} henceforth. 
Classical
channels encompass the whole gamut of communication channels in current
use, ranging from wireless systems to optical fibres. 
The story however radically changes when one considers the quantum 
setting where 
information is carried by properties of physical systems behaving 
according to the laws of quantum physics. Examples of quantum channels
include using photon polarisation to transmit information over optical
fibre, using nuclear spins in organic molecules to store information etc.
It is this setting that we study in this paper.

More than a decade before Deutsch published the first 
mathematical model of a quantum computer \cite{Deutsch}, Holevo studied 
the information carrying capabilities of quantum systems in his 
seminal paper \cite{Holevo}. Interest in quantum computing and quantum 
information theory grew exponentially after Shor's publication of
an efficient quantum algorithm for integer factoring 
\cite{shor:factoring}. 
Barely a few years after Shor's result, Holevo \cite{holevo:capacity}
and Schumacher-Westmoreland \cite{schumacher:capacity} proved
a quantum analogue of Shannon's noisy channel coding theorem by showing 
that the quantum
version of mutual information is a lower bound on the capacity of a
quantum channel to carry classical information in the asymptotic iid
setting. In doing so, they generalised the notion of a typical set
to a {\em typical subspace} which, in hindsight, was the natural thing 
to do since decoding
operations in the quantum setting are mathematically modelled by
projections onto subspaces.

As long as one is working only with point-to-point channels,
many inner bounds in asymptotic iid classical information theory can 
be extended to
the quantum setting by replacing typical sets with typical subspaces
and classical joint typicality lemmas with quantum joint typicality
lemmas. However as soon as one starts to deal with multiterminal
channels, one encounters the issue of defining intersections and unions
of typical subspaces which becomes a significant bottleneck. 
The situation becomes worse when one has to deal with multiterminal
quantum channels in the {\em one shot setting} where the channel can 
be used
only once. Though analogues for typical sets and typical subspaces have
recently been developed for the one shot setting, see e.g 
\cite{wang:DepsH}, their properties are not as nice as those for the
asymptotic iid setting. This exacerbates the bottleneck of defining 
unions and intersections. Though some ad hoc solutions have been
found for some quantum channels in the asymptotic iid regime, see e.g.
\cite{winter:cqmac,sen:interference},
this bottleneck continues to be a 
hindrance in proving good inner bounds for multiterminal quantum channels
in general.

In the next few subsections we will explain the problem of
intersections and unions of typical subspaces in more detail
by taking up the example of the simplest multiterminal quantum channel
viz. the multiple access channel (MAC). We will first explicitly show why
an intersection argument arises while proving an inner bound for the
classical MAC in the one shot setting. We will then explain why
naive strategies fail to address the intersection issue for the quantum
MAC. After that we will intuitively describe the two new tools introduced
in this paper called {\em tilting} and {\em smoothing and augmentation},
and then explain how they
can be used to define a robust notion of intersection and
union of typical subspaces that finally overcomes the bottleneck in
the quantum setting. The tools developed in this paper all work in the 
most general one shot quantum setting.

\subsection{One-shot inner bound for the classical MAC}
\label{sec:cMAC}
Let us illustrate the need for an intersection argument together
with a joint typicality
lemma by considering the problem of proving inner bounds for arguably the 
simplest multiterminal communication channel viz. the multiple access
channel (MAC) (see Fig.~\ref{fig:CMAC}). 
We consider the one-shot classical setting.
\begin{figure*}[!t]
\centering{
\includegraphics[width=.9\textwidth]{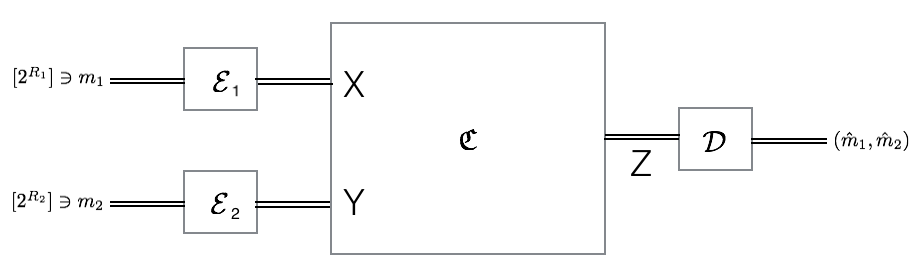}
}
\caption{A classical multiple access channel. Alice, Bob encode their
respective messages $m_1$, $m_2$ and then feed the resulting codewords
to the channel.
Charlie decodes to get his guess $(\hat{m}_1, \hat{m}_2)$ of the 
transmitted messages. $X$, $Y$
are the input alphabets and $Z$ is the output alphabet of the channel.
}
\label{fig:CMAC}
\end{figure*}
There are two senders Alice and Bob who
would like to send messages $m_1 \in [2^{R_1}]$, $m_2 \in 2^{R_2}$,
respectively to 
a receiver Charlie. There
is a communication channel $\chan$ with two inputs and one output called
the two-sender one-receiver MAC connecting Alice and Bob to Charlie. The 
two input alphabets of $\chan$ will be denoted
by $\cX$, $\cY$ and the output alphabet by $\cZ$. 
Let $0 \leq \epsilon \leq 1$.
On getting message $m_1$, Alice encodes it as a letter $x(m_1) \in \cX$
and feeds it to her channel input. 
Similarly on getting message $m_2$, Bob encodes it as a letter 
$y(m_2) \in \cY$ and feeds it to his channel input. The channel $\chan$
outputs a letter $z \in \cZ$ according to the channel probability
distribution $p(z | x(m_1), y(m_2))$. Charlie now has to try and guess 
the message pair $(m_1, m_2)$ from the channel output. 
We require that the probability of Charlie's decoding error,
where the probability arises both from Charlie's measurement as well
as classically averaging over the uniform distribution on
the set of message pairs $(m_1, m_2) \in [2^{R_1}] \times [2^{R_2}]$,
is at most $\epsilon$.

Consider the following randomised construction of a codebook $\cC$
for Alice 
and Bob. Fix probability distributions $p(x)$, $p(y)$ on sets $\cX$, $\cY$.
For $m_1 \in [2^{R_1}]$, choose $x(m_1) \in \cX$ independently according
to $p(x)$. Similarly
for $m_2 \in [2^{R_2}]$, choose $y(m_2) \in \cY$ independently according
to $p(y)$. 

Before describing the decoding strategy that Charlie follows, it
is useful to first define a concept called {\em hypothesis testing}.
Let $0 \leq \epsilon \leq 1$. Let $p$, $q$ be 
two probability distrbutions on the same sample space $\Omega$. A
`classical POVM element' or `test' on $\Omega$ is defined to be
a function $f: \Omega \rightarrow [0, 1]$. Intuitively, for a sample
point $\omega \in \Omega$, $f(\omega)$ denotes its probability 
of acceptance by the test.
For two classical POVM elements $f$, $g$ on $\Omega$, we can define the
`intersection' classical POVM element $f \cap g$ as follows: 
$(f \cap g)(\omega) := \min\{f(\omega), g(\omega)\}$. 
Similarly, we can define the `union' classical POVM element 
$f \cup g$ as follows:
$(f \cup g)(\omega) := \max\{f(\omega), g(\omega)\}$. 
Following Wang and Renner~\cite{wang:DepsH}, we define
the classical {\em hypothesis testing relative entropy} 
$D^\epsilon_H(p \| q)$ as follows:
\[
D^\epsilon_H(p \| q) := 
\max_{f: \sum_\omega f(\omega) p(\omega) \geq 1 - \epsilon} 
-\log \sum_\omega f(\omega) q(\omega),
\]
where the maximisation is over all classical POVM elements on $\Omega$
`accepting' the distribution $p$ with probability at least
$1 - \epsilon$. 
It is easy to see that the optimising POVM element $f$
attains equality in the constraint for $p(\cdot)$, as well as achieves
the maximum in objective function for $q(\cdot)$.

Define the probability distribution $p(\cdot|x)$ on $\cZ$ as 
$
p(z|x) := \sum_y p(y) p(z | x, y)
$
for all $z \in \cZ$. With a slight abuse of notation, we shall often
use $p(z|x)$ to also denote the distribution $p(\cdot|x)$ on $\cZ$.
Similarly, define probability distributions $p(\cdot|y)$, $p(\cdot)$ on 
$\cZ$ in the natural fashion.
Define the probability distributions 
\begin{eqnarray*}
p^{XYZ}(x,y,z) & := & p(x) p(y) p(z | x, y), \\
(p^{XZ} \times p^Y)(x,y,z) & := & p(x) p(y) p(z | x), \\
(p^{YZ} \times p^X)(x,y,z) & := & p(x) p(y) p(z | y), \\
(p^{X} \times p^Y \times p^Z)(x,y,z) & := & p(x) p(y) p(z) 
\end{eqnarray*}
on $\cX \times \cY \times \cZ$ in the natural manner.
Consider classical POVM elements $f^Y$, $f^X$, $f^{X,Y}$ achieving the 
respective maxima in the definitions of
$D^\epsilon_H(p^{XYZ} \| p^{XZ} \times p^Y)$,
$D^\epsilon_H(p^{XYZ} \| p^{YZ} \times p^X)$,
$D^\epsilon_H(p^{XYZ} \| p^{X} \times p^Y \times p^Z)$ respectively.
As a shorthand, we will use the hypothesis testing mutual information
quantities 
$I^\epsilon_H(Y : X Z)$,
$I^\epsilon_H(X : Y Z)$,
$I^\epsilon_H(X Y : Z)$ 
to denote the hypothesis testing relative entropy quantities
$D^\epsilon_H(p^{XYZ} \| p^{XZ} \times p^Y)$,
$D^\epsilon_H(p^{XYZ} \| p^{YZ} \times p^X)$,
$D^\epsilon_H(p^{XYZ} \| p^{X} \times p^Y \times p^Z)$
respectively. Thus,
\begin{equation}
\label{eq:cmacDeps}
\begin{array}{rcl}
\sum_{x,y,z} p(x,y,z) f^X(x,y,z) & \geq & 1 - \epsilon, \\
\sum_{x,y,z} p(x)p(y,z) f^X(x,y,z) & \leq & 2^{-I^\epsilon_H(X : Y Z)}; \\
&& \\
\sum_{x,y,z} p(x,y,z) f^Y(x,y,z) & \geq & 1 - \epsilon, \\
\sum_{x,y,z} p(y)p(y,z) f^Y(x,y,z) & \leq & 2^{-I^\epsilon_H(Y : X Z)}; \\
&& \\
\sum_{x,y,z} p(x,y,z) f^{X,Y}(x,y,z) & \geq & 1 - \epsilon, \\
\sum_{x,y,z} p(x)p(y)p(z) f^{X,Y}(x,y,z) & \leq & 2^{-I^\epsilon_H(XY:Z)}.
\end{array}
\end{equation}

We now describe the decoding strategy that Charlie follows in order to
try and guess the message pair $(m_1, m_2)$ that was actually sent. 
Consider the intersection classical POVM 
element $f := f^X \cap f^Y \cap f^{X,Y}$. 
Suppose the channel output is $z$. Then, Charlie uses the
following randomised algorithm for decoding.
\begin{quote}
For $\hat{m}_1 = 1 \mbox{ to } 2^{R_1}$ \\
\hspace*{3mm}  For $\hat{m}_2 = 1 \mbox{ to } 2^{R_2}$ \\
\hspace*{6mm}
Toss a coin with probability of HEAD \\
\hspace*{6mm}
being $f(x(\hat{m}_1), y(\hat{m}_2), z)$. \\
\ \\
\hspace*{6mm}
If the coin comes up HEAD, declare \\
\hspace*{6mm}
$(\hat{m}_1, \hat{m}_2)$ as Charlie's guess and halt. \\
\ \\
\hspace*{6mm}
If the coin comes up TAILS, go to \\
\hspace*{6mm}
next iteration.

Declare FAIL, if Charlie did not declare \\
any guess above.
\end{quote}

We now analyse the expectation, under the choice of a random codebook
$\cC$, of the error probability of Charlie's decoding algorithm.
Suppose the message pair $(m_1, m_2)$ is inputted to the channel.
Let $\prec$ denote the lexicographic order on message pairs. Let 
the channel output be denoted by $z$. Then,
a decoding error occurs only if Charlie tosses a HEAD for a pair
$(\hat{m}_1, \hat{m}_2) \prec (m_1, m_2)$ or if Charlie tosses a TAIL
for $(m_1, m_2)$. The expectation, over the choice of the random 
codebook $\cC$, of the decoding error is then upper bounded by
\begin{equation}
\label{eq:cqMAC}
\begin{array}{rcl}
\lefteqn{
\E_{\cC}[
p_e(\cC; m_1, m_2) 
]
} \\
& \leq &
2^{R_1 + R_2} 2^{-I^\epsilon_H(X Y : Z)} +
2^{R_2} 2^{-I^\epsilon_H(Y : X Z)} \\
&      &
{} +
2^{R_1} 2^{-I^\epsilon_H(X : Y Z)} +
3 \epsilon.
\end{array}
\end{equation}
A proof can be found in \ref{sec:proofcqMAC}.
Now define the average decoding error probability under a codebook
$\cC$ to be
\[
p_e(\cC) :=
2^{-(R_1 + R_2)} 
\sum_{(m_1, m_2)}
p_e(\cC; m_1, m_2).
\]
Then the expectation, under the choice of a random codebook
$\cC$, of the average decoding error probability is upper bounded by
\[
\E_{\cC}[p_e(\cC)] \leq
2^{R_1 + R_2} 2^{-I^\epsilon_H(X Y : Z)} +
2^{R_2} 2^{-I^\epsilon_H(Y : X Z)} +
2^{R_1} 2^{-I^\epsilon_H(X : Y Z)} +
3 \epsilon.
\]
Thus, for any rate pair in the region given by
\begin{eqnarray*}
R_1 & \leq & I^\epsilon_H(X : Y Z) - \log \frac{1}{\epsilon}, \\
R_2 & \leq & I^\epsilon_H(Y : X Z) - \log \frac{1}{\epsilon}, \\
R_1 + R_2 & \leq & I^\epsilon_H(X Y : Z) - \log \frac{1}{\epsilon}, \\
\end{eqnarray*}
there is a codebook $\cC$ with average decoding error probability less
than $6 \epsilon$.

\subsection{Unions and intersections in the classical setting}
We now step back and discuss the intersection classical POVM element
$f := f^X \cap f^Y \cap f^{X,Y}$ used in
the above proof. The intersection argument was crucial in constructing
a simultaneous decoder for Charlie. In the asymptotic iid setting for
the multiple access channel,
it is possible to avoid simultaneous decoding and instead use successive
cancellation decoding combined with time sharing~\cite{book:elgamalkim}. 
However, in 
the one-shot setting time sharing does not make sense and successive
cancellation gives only a finite set of achievable rate pairs. Thus,
in order to get a continuous achievable rate region, we are forced
to use simultaneous decoders only. There are also situations even in the
asymptotic iid setting, e.g.
in Marton's inner bound with common message for the broadcast channel,
where we need to use intersection arguments~\cite{book:elgamalkim}.
Similarly, union arguments crop up in some inner bound
proofs, e.g. Marton's inner bound without common message for the broadcast 
channel, even in the
asymptotic iid setting~\cite{book:elgamalkim}. In the one-shot setting,
union bounds occur more frequently e.g., in the Han-Kobayashi inner bound
for the interference channel~\cite{sen:simultaneous}. 
Thus, intersection and union arguments are indispensable
in network information theory.

Before we proceed to the quantum setting, we prove for completeness
sake a `one-shot classical joint typicality lemma'. 
In fact, it is nothing but an application of intersection and union
of classical POVM elements.
\begin{fact}[Classical joint typicality lemma]
\label{fact:ctypical}
Let $p_1, \ldots, p_t$, \\
$q_1, \ldots, q_l$ be probability 
distributions on a set $\cX$.
Let $0 \leq \{\epsilon_{ij}\}_{ij} \leq 1$, where $i \in [t]$, $j \in [l]$.
Then there is a classical POVM element $f$ on $\cX$ such that:
\begin{enumerate}

\item
For all $i \in [t]$,
$
\sum_x p_i(x) f(x) \geq 
1 - \{\sum_{j = 1}^l \epsilon_{ij}\};
$

\item
For all $j \in [l]$,
$
\sum_x q_j(x) f(x) \leq 
\sum_{i=1}^t 2^{-D^{\epsilon_{ij}}_H(p_i \| q_j)}.
$

\end{enumerate}
\end{fact}
\begin{proof}
For $i \in [t]$, $j \in [l]$, let $f_{ij}$ be the classical POVM 
element achieving the
minimum in the definition of $D^{\epsilon_{ij}}_H(p_i \| q_j)$. Define
the classical POVM element
$f := \cup_{i=1}^t \cap_{j=1}^l f_{ij}$. 
Observe that for any $x \in \cX$,
\[
1 - f(x) 
 \leq 
\min_{i: i \in [t]} \left\{\sum_{j=1}^l (1 - f_{ij}(x))\right\}, 
~~~
f(x) 
 \leq 
\sum_{i=1}^t \min_{j: j \in [l]} \{f_{ij}(x)\}.
\]
It is now easy to see that 
$f$ satisfies the properties claimed above.
\end{proof}

\subsection{Extending unions and intersections to the quantum setting}
\label{sec:quantunionintersection}
We now ponder what is required to extend the above inner bound proof for
the classical MAC to the setting of the one-shot classical-quantum 
multiple access channel (cq-MAC).
In the cq-MAC, there are two senders Alice and Bob who
would like to send messages $m_1 \in [2^{R_1}]$, $m_2 \in 2^{R_2}$,
respectively to a receiver Charlie. There
is a communication channel $\chan$ with two classical inputs and one 
quantum output connecting Alice and Bob to Charlie. The 
two input alphabets of $\chan$ will be denoted
by $\cX$, $\cY$ and the output Hilbert space by $\cZ$. 
If the pair $(x,y)$ is fed into the channel inputs, the output of
the channel is a density matrix $\rho_{x,y}$ in $\cZ$.
Let $0 \leq \epsilon \leq 1$. 
On getting message $m_1$, Alice encodes it as a letter $x(m_1) \in \cX$
and feeds it to her channel input. 
Similarly on getting message $m_2$, Bob encodes it as a letter 
$y(m_2) \in \cY$ and feeds it to his channel input. 
Charlie now has to try and guess 
the message pair $(m_1, m_2)$ from the channel output state
$\rho_{x(m_1),y(m_2)}$.
We require that the probability of Charlie's decoding error averaged
over the uniform distribution on
the set of message pairs $(m_1, m_2) \in [2^{R_1}] \times [2^{R_2}]$ 
is at most $\epsilon$.

One can do a similar randomised construction of a codebook $\cC$ for
Alice and Bob as before. One can make use of the {\em hypothesis
testing relative entropy} \cite{wang:DepsH} for a pair of 
quantum states $\rho$, $\sigma$ 
in the same Hilbert space $\cH$, which is
defined as follows:
\[
D^\epsilon_H(\rho \| \sigma) := 
\max_{\Pi: \Tr [\Pi \rho] \geq 1 - \epsilon} 
-\log \Tr [\Pi \sigma],
\]
where the maximisation is over all POVM elements $\Pi$ on $\cH$ 
(i.e. positive semidefinite operators $\Pi$ such that 
$\Pi \leq \one_{\cH}$)
`accepting' the state $\rho$ with probability at least
$1 - \epsilon$. 
Again, it is easy to see that the optimising POVM element $\Pi$
attains equality in the constraint for $\rho$, as well as achieves
the maximum in objective function for $\sigma$.
One can define the analogous quantum state
\[
\rho^{XYZ} := 
\sum_{x,y} p(x) p(y) \ketbra{x}^X \otimes \ketbra{y}^Y \otimes
           \rho_{x,y}^Z
\]
and the tensor products of the marginals
$\rho^{XZ} \otimes \rho^Y$,
$\rho^{YZ} \otimes \rho^X$,
$\rho^X \otimes \rho^Y \otimes \rho^Z$,
as well as the hypothesis testing mutual informations 
\begin{eqnarray*}
I^\epsilon_H(X:YZ)
& := &
D^\epsilon_H(\rho^{XYZ} \| \rho^X \otimes \rho^{YZ}), \\
I^\epsilon_H(Y:XZ)
& := &
D^\epsilon_H(\rho^{XYZ} \| \rho^Y \otimes \rho^{XZ}), \\
I^\epsilon_H(XY:Z)
& := &
D^\epsilon_H(\rho^{XYZ} \| \rho^X \otimes \rho^Y \otimes \rho^{Z})
\end{eqnarray*}
in the quantum setting too.
Thus, we would like to prove  that any rate pair in the region given by
\begin{equation}
\label{eq:MACregion}
\begin{array}{rcl}
R_1 & \leq & I^\epsilon_H(X : Y Z) - \log \frac{1}{\epsilon}, \\
& & \\
R_2 & \leq & I^\epsilon_H(Y : X Z) - \log \frac{1}{\epsilon}, \\
& & \\
R_1 + R_2 & \leq & I^\epsilon_H(X Y : Z) - \log \frac{1}{\epsilon} 
\end{array}
\end{equation}
is achievable with small 
average decoding error probability.

Suppose there exists a {\bf single} POVM element $\Pi$ on the Hilbert
space $\cX \otimes \cY \otimes \cZ$ that simultaneously
satisfies the following properties:
\begin{quote}
\centerline{
{\bf
Desirable properties of the single POVM element $\Pi$
}
}
\begin{enumerate}

\item
\setcounter{type1error}{\value{enumi}}
$\Tr [\Pi \rho^{XYZ}] \geq 1 - 3 \epsilon$;

\item
\setcounter{type2errorx}{\value{enumi}}
$
\Tr [\Pi (\rho^{YZ} \otimes \rho^X)] \leq 
2^{-I^\epsilon_H(X:YZ)}
$;

\item
\setcounter{type2errory}{\value{enumi}}
$
\Tr [\Pi (\rho^{XZ} \otimes \rho^Y)] \leq 
2^{-I^\epsilon_H(Y:XZ)}
$;

\item
\setcounter{type2errorxy}{\value{enumi}}
$
\Tr [\Pi (\rho^X \otimes \rho^Y \otimes \rho^Z)] \leq 
2^{-I^\epsilon_H(XY:Z)}
$.

\end{enumerate}
\end{quote}
In the classical setting, such a POVM element $f$ was constructed by
taking the intersections of the three POVM elements 
$f^X$, $f^Y$, $f^{X,Y}$ achieving the respective maxima
in the definitions of the three entropic quantities
$I^\epsilon_H(X:YZ)$, 
$I^\epsilon_H(Y:XZ)$, 
$I^\epsilon_H(XY:Z)$.
In the quantum setting, if such an `intersection' POVM element 
$\Pi$ exists it is
indeed possible to construct a decoding algorithm for Charlie with
average error probability at most $O(\epsilon)$ using
the `pretty good measurement' 
\cite{belavkin:pgm1, belavkin:pgm2, holevo:capacity, schumacher:capacity},
the Hayashi-Nagaoka operator inequality
\cite{HayashiNagaoka}
and mimicking the classical analysis given above for the decoding error.

\begin{figure*}[!t]
\begin{center}
\begin{minipage}[c]{0.35\textwidth}
\centering{
\includegraphics[width=.9\textwidth]{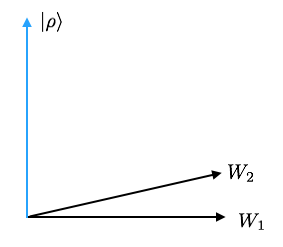} 
}

\centering{{\bf Pathology:} Span can be entire space}
\end{minipage}
~
\begin{minipage}[c]{0.55\textwidth}
\centering{
\includegraphics[width=.9\textwidth]{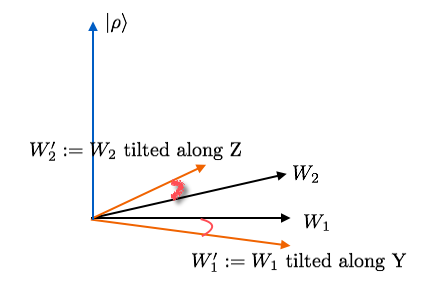}
}

\centering{{\bf Tilting:} The pathology is destroyed now!}
\end{minipage}
\end{center}
\caption{Span versus tilted span. The left figure shows how a quantum
state $\rho$ can have small projections on two subspaces $W_1$ and
$W_2$ but large projection on their span $W_1 + W_2$. The 
pathology can be severe, in that the state and the 
subspaces can all be one
dimensional and the state can lie fully in the span of the two subspaces.
On the right, we see how tilting the two subspaces along two orthogonal
directions (which we call Y and Z in the figure) destroys the pathology. 
The quantum state $\rho$ now has small projection onto the span 
$W'_1 + W'_2$ of the two tilted subspaces $W'_1$ and $W'_2$.
}
\label{fig:SpanTilting}
\end{figure*}
Let $\Pi^X$, $\Pi^Y$, $\Pi^{X,Y}$ be the three POVM elements achieving
the maxima in the definitions of the three entropic quantities
$I^\epsilon_H(X:YZ)$, 
$I^\epsilon_H(Y:XZ)$, 
$I^\epsilon_H(XY:Z)$.
Let us now look at various simple ideas to generalise intersection of
POVM elements to the quantum setting. If the POVM elements were 
projectors, which can be ensured without loss of generality by
the Gelfand-Naimark technique of embedding
the quantum states into a larger Hilbert space 
$\cX \otimes \cY \otimes (\cZ \otimes \C^2)$,
one can try taking the projector $\Pi$ onto the intersection of 
the supports
of $\Pi^X$, $\Pi^Y$, $\Pi^{X,Y}$. This indeed ensures that
Properties~\arabic{type2errorx}, \arabic{type2errory}, 
\arabic{type2errorxy}
described above hold for $\Pi$. Unfortunately, the intersection can 
easily be the zero
subspace which kills all hope of satisfying Property~\arabic{type1error}
even approximately.
The next simple idea to define the `intersection' POVM element would be
to take the product 
$\Pi := \Pi^X \Pi^Y \Pi^{X,Y} \Pi^{X,Y} \Pi^Y \Pi^X$.
With this definition, $\Pi$ can be shown to satisfy 
Property~\arabic{type1error} 
with lower bound $1 - O(\epsilon)$ using the 
non-commutative union bound~\cite{gao:union, oskouei:union}.
However, it is not at all clear that the remaining three properties
can be simultaneously satisfied, even approximately, because 
$\Pi^X$, $\Pi^Y$, $\Pi^{X,Y}$ do not commute in general. To get an
idea of the difficulty, consider the expression
\[
\Tr [\Pi (\rho^{YZ} \otimes \rho^X)] =
\Tr [\Pi^X \Pi^Y \Pi^{X,Y} \Pi^{X,Y} \Pi^Y \Pi^X 
     (\rho^{YZ} \otimes \rho^X)]
\]
that arises if one were to attempt to prove
Property~\arabic{type2errorx}.
It is true that 
$\Tr [\Pi^Y (\rho^{YZ} \otimes \rho^X)] = 2^{-I(X:YZ)}$,
but it is also possible that
\begin{eqnarray*}
\lefteqn{
\Tr [\Pi^X \Pi^Y \Pi^{X,Y} \Pi^{X,Y} \Pi^Y \Pi^X 
     (\rho^{YZ} \otimes \rho^X)] 
} \\
&   &
~~~~~
{} - \Tr [\Pi^Y (\rho^{YZ} \otimes \rho^X)] \\
& \gg &
2^{-I(X:YZ)}.
\end{eqnarray*}
This makes it impossible to show the 
achievability of the desired rate region (Equation~\ref{eq:MACregion}) 
with this definition of
$\Pi$. In fact, one of the main technical contributions of this paper is a 
robust novel notion of {\em intersection} of non-commuting POVM elements
achieving the maxima in the definitions of the appropriate
hypothesis testing mutual information quantities.

As a first step towards defining a 
robust notion of intersection of non-commuting POVM elements, we
address the complementary problem of defining a
robust notion of union of non-commuting POVM elements. The `intersection'
POVM element can then be defined to be simply the complement of the
`union' of the complements. Note that the complement of a POVM
element $\Pi$ in a Hilbert space $\cH$ is defined to be
$\one_{\cH} - \Pi$. The POVM elements are projectors
without loss of generality.
One can then naively
define the `union' of a family of projectors to be the projector onto 
the span of the supports (the union of supports is not a vector space
in general). However, this does not give us anything new
as the span can indeed be the entire space and so 
Property~\arabic{type1error} can be totally useless, that is, the 
lower bound in Property~\arabic{type1error} can be as low as zero! 
The problem with the
span idea is captured by the following example.
Consider the two-dimensional Hilbert space $\C^2$. Let 
$W_1$ be the one-dimensional space spanned by $\ket{0}$ and
$W_2$ the one-dimensional space spanned by 
$\sqrt{1-\epsilon}\ket{0} + \sqrt{\epsilon}\ket{1}$. Consider the quantum
state $\rho := \ketbra{1}$. Now the probability of $\rho$ being
accepted by $W_1$, $W_2$ is $0$ and $\epsilon$ respectively. However,
the probability of $\rho$ being accepted by the span of $W_1$ and
$W_2$ is one.

Notice however that the above pathological phenomenon with the span
occurs only because the
subspaces $W_1$ and $W_2$ have `small angles' between them. 
We overcome this problem by `tilting' $W_1$, $W_2$ in orthogonal directions
to form new spaces $W'_1$, $W'_2$ (see Fig.~\ref{fig:SpanTilting}). 
The Hilbert space has to be
enlarged sufficiently to allow this tilting to be possible. 
The process 
of tilting increases the `angles' between the subspaces in {\em orthogonal
directions} allowing one to recover an upper bound for the span very
close to the sum of acceptance probabilities, just like in the 
classical setting. This `tilting' idea is formalised in 
Proposition~\ref{prop:tiltedspan}. Thus, the `tilted span' is best 
thought of as a robust notion of {\em union of subspaces} satisfying a 
well behaved union bound.

We now describe the tilting process for the above `pathological' example
of Fig.~\ref{fig:SpanTilting} in detail.
Let $A_1 \cong \C^2$ be the original Hilbert space. Consider the 
extended Hilbert space
$\hA := A_1 \oplus A_2 \oplus A_3$, where each $A_i \cong \C^2$. 
The subspaces
$W_1$, $W_2$ are the one-dimensional subspaces spanned by the vectors
$\ket{0}^{A_1}$ and 
$\sqrt{1-\epsilon} \ket{0}^{A_1} + \sqrt{\epsilon} \ket{1}^{A_1}$ 
respectively, in $\hA$. 
Fix $0 < \alpha < 1/2$. 
The tilted subspaces are defined by
\begin{eqnarray*}
W'_1 
& := &
\spanning\{
\sqrt{1-\alpha} \ket{0}^{A_1} + \sqrt{\alpha} \ket{0}^{A_2}
\}, \\
W'_2
& := &
\spanning\{
\sqrt{(1-\alpha)(1-\epsilon)} \ket{0}^{A_1} +
\sqrt{\alpha(1-\epsilon)} \ket{0}^{A_3} \\
&    &
~~~~~~~~~~
{} + \sqrt{(1-\alpha)\epsilon} \ket{1}^{A_1} +
\sqrt{\alpha \epsilon} \ket{1}^{A_3}
\},
\end{eqnarray*}
where the span is taken in $\hA$. Consider now the span 
$W' := W'_1 + W'_2$. We refer to $W'$ as the `tilted span' of $W_1$ and
$W_2$. It has the following orthonormal basis:
\begin{eqnarray*}
W'
& = & 
\spanning\{
\sqrt{1-\alpha} \ket{0}^{A_1} + \sqrt{\alpha} \ket{0}^{A_2}, \\
&   &
~~~~~~~
\left.
\frac{
\begin{array}{l}
\alpha \sqrt{(1-\alpha)(1-\epsilon)} \ket{0}^{A_1} 
-
(1-\alpha)\sqrt{\alpha (1-\epsilon)} \ket{0}^{A_2} \\
{} +
\sqrt{\alpha(1-\epsilon)} \ket{0}^{A_3} 
+ \sqrt{(1-\alpha)\epsilon} \ket{1}^{A_1} 
+ \sqrt{\alpha \epsilon} \ket{1}^{A_3}
\end{array}
}{\sqrt{\epsilon + \alpha(1-\epsilon)(2-\alpha)}} 
\right\}.
\end{eqnarray*}
Thus tilting destroys the pathology since the projection of 
$\ket{\rho}$ onto $W'$ has length
\[
\frac{\sqrt{(1-\alpha)\epsilon}}
     {\sqrt{\epsilon + \alpha(1-\epsilon)(2-\alpha)}} 
< \frac{1}{\sqrt{\alpha}} \sqrt{\epsilon}.
\]
Suppose $\epsilon < 1/4$.
Setting $\alpha = \sqrt{\epsilon}$, we can finally conclude that
the tilted subspaces are $(2^{1/2} \epsilon^{1/4})$-close to their
respective originals in the $\ell_2$-norm, whereas $\ket{\rho}$ has
a projection of length less than $\epsilon^{1/4}$ onto $W'$,
much less than one. In terms of probabilities, the probability of
measuring $W'$ in the state $\ket{\rho}$ is less than 
$\sqrt{\epsilon}$. The classical union bound on probabilities would 
give $\epsilon$.

Let us define the `intersection' of subspaces to be the complement of the
tilted span of the complementary subspaces.
Applying this recipe to projectors $\Pi^X$, $\Pi^Y$, $\Pi^{X,Y}$ 
achieving the maxima in the definitions of the three entropic quantities
$I^\epsilon_H(X:YZ)$, 
$I^\epsilon_H(Y:XZ)$, 
$I^\epsilon_H(XY:Z)$
gives us a 
projector $\Pi$ that satisfies Property~\arabic{type1error} with
lower bound $1 - O(\sqrt{\epsilon})$ by setting $\alpha = \sqrt{\epsilon}$
in the upper bound 
of Proposition~\ref{prop:tiltedspan}. However, it is not clear whether
$\Pi$ satisfies the Properties~\arabic{type2errorx},
\arabic{type2errory},
\arabic{type2errorxy} even approximately.
This is because in order to prove, say, Property~\arabic{type2errorx},
one has to use the lower bound of Proposition~\ref{prop:tiltedspan} 
with $\alpha = \sqrt{\epsilon}$, which would give an upper bound of at 
least $\sqrt{\epsilon}$ for Property~\arabic{type2errorx}!

Hence we need another new idea in order to obtain a robust notion of
intersection of subspaces. The new idea comes from a phenomenon that
we will call {\em smoothing and augmentation}. 
Let us see a simple example of this. Consider a pure state
$\psi^{XY}$ on the bipartite Hilbert space $\cX \otimes \cY$. Expressed
in the computational basis, the state looks like
$
\ket{\psi}^{XY} = 
\sum_{xy} s_{xy} \ket{x}^X \otimes \ket{y}^Y.
$
Suppose $\ket{\psi}^{XY}$ is isometrically tilted in the following 
fashion to make a
state $\ket{\psi'}^{X'Y'}$ in a larger Hilbert space $\cX' \otimes \cY'$
as follows:
\begin{eqnarray*}
\lefteqn{\ket{\psi'}^{X'Y'}} \\
& := &
\sum_{xy} s_{xy} 
(\sqrt{1-2\epsilon} \ket{x}^X + \sqrt{\epsilon} \ket{x}^{X_x} +
 \sqrt{\epsilon} \ket{x}^{X_y}) \\
&   &
~~~~~~~~~
{} \otimes
(\sqrt{1-2\epsilon} \ket{y}^Y + \sqrt{\epsilon} \ket{y}^{Y_x} +
 \sqrt{\epsilon} \ket{y}^{Y_y}),
\end{eqnarray*}
where
\[
\cX' :=
\cX \oplus \bigoplus \cX_x \oplus \bigoplus \cX_y,
\]
\[
\cY' :=
\cY \oplus \bigoplus \cY_x \oplus \bigoplus \cY_y,
\]
$\cX_x$, $\cX_y$ are orthogonal Hilbert spaces of the same dimension as
$\cX$ indexed by the computational basis states $\ket{x}^X$, $\ket{y}^Y$
and
$\cY_x$, $\cY_y$ are orthogonal Hilbert spaces of the same dimension as
$\cY$ indexed by the computational basis states $\ket{x}^X$, $\ket{y}^Y$.
Then,
\begin{eqnarray*}
\lefteqn{(\psi')^{X'}} \\
& = &
\Tr_{Y'} [
\sum_{xx'} \sum_{yy'} s_{xy} s^*_{x'y'}
(\sqrt{1-2\epsilon} \ket{x}^X + \sqrt{\epsilon} \ket{x}^{X_x} +
 \sqrt{\epsilon} \ket{x}^{X_y}) \\
&   &
~~~~~~~~~~~~~~~~~~~~~~~~~~~~~
(\sqrt{1-2\epsilon} \bra{x'}^X + \sqrt{\epsilon} \bra{x'}^{X_{x'}} +
\sqrt{\epsilon} \bra{x'}^{X_{y'}}) \\
&   &
~~~~~~~~~~~~~~~~~~~~~
{} \otimes
(\sqrt{1-2\epsilon} \ket{y}^Y + \sqrt{\epsilon} \ket{y}^{Y_x} +
 \sqrt{\epsilon} \ket{y}^{Y_y}) \\
&   &
~~~~~~~~~~~~~~~~~~~~~~~~~~~~~
(\sqrt{1-2\epsilon} \bra{y'}^Y + \sqrt{\epsilon} \bra{y'}^{Y_{x'}} +
\sqrt{\epsilon} \bra{y'}^{Y_{y'}}) \\
& = &
\sum_{xx'} \sum_y s_{xy} s^*_{x'y}
(\sqrt{1-2\epsilon} \ket{x}^X + \sqrt{\epsilon} \ket{x}^{X_x} +
 \sqrt{\epsilon} \ket{x}^{X_y}) \\ 
&   &
~~~~~~~~~~~~~~~~~~~~~~~~~~~~~
(\sqrt{1-2\epsilon} \bra{x'}^X + \sqrt{\epsilon} \bra{x'}^{X_x} +
 \sqrt{\epsilon} \bra{x'}^{X_y}) \\
& = &
(\psi'')^{X'} + \phi^{X'},
\end{eqnarray*}
where 
\begin{eqnarray*}
\lefteqn{(\psi'')^{X'}} \\ 
& := &
\sum_{xx'} (\sum_y s_{xy} s^*_{x'y}) 
(\sqrt{1-2\epsilon} \ket{x}^X + \sqrt{\epsilon} \ket{x}^{X_x}) \\
&   &
~~~~~~~~~~~~~~~~~~~~~~~~
(\sqrt{1-2\epsilon} \bra{x'}^X + \sqrt{\epsilon} \bra{x'}^{X_x})
\end{eqnarray*}
is an isometric tilt of $\psi^X$ and
\begin{eqnarray*}
\lefteqn{\phi^{X'}} \\
& := &
\sum_{xx'} \sum_y s_{xy} s^*_{x'y}
(
(\sqrt{1-2\epsilon} \ket{x}^X + \sqrt{\epsilon} \ket{x}^{X_x}) 
(\sqrt{\epsilon} \bra{x'}^{X_y})  \\
&   &
~~~~~~~~~~~~~~~~~~~~~~~~
{} +
(\sqrt{\epsilon} \ket{x}^{X_y}) 
(\sqrt{1-2\epsilon} \bra{x'}^X + \sqrt{\epsilon} \bra{x'}^{X_x}) \\
&   &
~~~~~~~~~~~~~~~~~~~~~~~~
{} +
\epsilon \ket{x}^{X_y}\bra{x'}
)
\end{eqnarray*}
is a Hermitian matrix.
Observe that the Hilbert spaces $\cX_y$, as $\ket{y}$ runs over the
computational basis of $\cY$, are orthogonal. 
Let $\ket{\psi_y}^X := \sum_x s_{xy} \ket{x}^X$ be the unnormalised
state that results when $Y$ is measured in the computational basis and
the outcome is $y$. 
Define the following scaled isometric embeddings from $\cX$ into 
$\cX'$:
\begin{eqnarray*}
\cT_X(\ket{x}^X) 
& := &
\sqrt{1-2\epsilon} \ket{x}^X + \sqrt{\epsilon} \ket{x}^{X_x}, \\
\cT_y(\ket{x}^X) 
& := &
\ket{x}^{X_y}.
\end{eqnarray*}
It is now easy to see that 
\begin{eqnarray*}
\lefteqn{
\ellinfty{\phi^{X'}}
} \\
& \leq &
\ellinfty{
\sqrt{\epsilon} \sum_y \cT_X \ket{\psi_y}^X \bra{\psi_y} \cT_y^\dag
} 
+
\ellinfty{
\sqrt{\epsilon} \sum_y \cT_y \ket{\psi_y}^X \bra{\psi_y} \cT_X^\dag 
} \\
&   &
{} +
\epsilon \ellinfty{\sum_y \cT_y \ket{\psi_y}^X \bra{\psi_y} \cT_y^\dag} \\
& \leq &
2 \sqrt{\epsilon |\cY|} 
(\max_y \ellinfty{\cT_X \ket{\psi_y}^X \bra{\psi_y}}) 
+
\epsilon (\max_y \ellinfty{\ket{\psi_y}^X \bra{\psi_y}}) \\
&   =  &
2 \sqrt{\epsilon |\cY|} (\max_y \elltwo{\ket{\psi_y}^X}^2) +
\epsilon (\max_y \elltwo{\ket{\psi_y}^X}^2) \\
& \leq &
3 \sqrt{\epsilon |\cY|} (\max_y \elltwo{\ket{\psi_y}^X}^2).
\end{eqnarray*}
Recall that $\elltwo{\ket{\psi_y}^X}^2$ is the 
probability of obtaining the outcome $y$. Thus, if the probability
distribution of the outcomes obtained by measuring $Y$ has low
$\ell_\infty$-norm or in other words is `smooth', then 
$\ellinfty{(\psi')^{X'} - (\psi'')^{X'}}$ is
small. The import of this statement is that though $(\psi')^{X'Y'}$ 
is tilted along both `x' and `y' directions, $(\psi')^{X'}$ is primarily
tilted along `x' direction only since it is close to
$(\psi'')^{X'}$. This desired property is what we mean by 
{\em smoothing $\Tr_{Y'}[\cdot]$}. 

Note however that smoothing
$\Tr_{Y'}[\cdot]$ in the above example is possible only if 
the probability distribution obtained
by measuring $Y$ is `smooth'. This may not be the case for an arbitrary
state $\psi^{XY}$. So we {\em augment} $\psi^{XY}$ to a state
$\hat{\psi}^{\hX \hY}$ defined on the Hilbert space $\hcX \otimes \hcY$ as
follows:
\[
\hat{\psi}^{\hX \hY} :=
|\cL|^{-2} \psi^{XY} \otimes \one^{L_X} \otimes \one^{L_Y},
\]
where $\cL_X$, $\cL_Y$ are new Hilbert spaces of the same dimension
$|\cL|$, $\hcX := \cX \otimes \cL_X$, $\hcY := \cY \otimes \cL_Y$.
We can now tilt $\hat{\psi}^{\hX \hY}$ to a state 
$(\hat{\psi}')^{\hX' \hY'}$ in the same fashion as above, where tilting
a mixed state means tilting its constituent pure states. If the dimension
$|\cL|$ of the augmenting system is large enough, the
operation of $\Tr_{Y'}[\cdot]$ indeed achieves smoothing in the sense
described above since
\[
\ellinfty{(\hat{\psi}')^{\hX'} - (\hat{\psi}'')^{\hX'}} \leq
3 \sqrt{\frac{\epsilon |\cY|}{|\cL|}}.
\]
We refer to this entire process as
{\em smoothing and augmentation}.

The above paragraph can be summarised in the observation that though
partial trace can increase the $\ell_\infty$-norm
of a quantum state in general, the increase is negative or only slightly 
positive if the state is highly correlated between the traced out part and
the untraced out part. 
As we will see later on, in general after augmentation it is possible 
to tilt a multipartite quantum state, say
$\rho^{X_1 \cdots X_k}$, along several directions with carefully
chosen amplitudes so as to get a state $(\rho')^{X_1 \cdots X_k}$ with 
the property it is highly correlated across all partitions of the
systems $X_1, \ldots, X_k$. As a consequence, for example, for any
subset $\{\} \neq S \subset [k]$, 
$(\rho')^{X_S} \otimes (\rho')^{X_{\bar{S}}}$ is tilted along
$(S, \bar{S})$ for all practical purposes even though
$(\rho')^{X_{[k]}} := (\rho')^{X_1 \cdots X_k}$ is tilted along $[k]$.
The formal statement and proof of this result, for which we have to
define a more nuanced tilting operation called {\em matrix tilt}, is 
technically intricate
and can be found in Proposition~\ref{prop:cqtypical} below.

The idea of tilting along both `x' and `y' directions combined with
smoothing and augmentation can be used as a 
starting point and further refined, leading finally to 
a proof of the desired inner bound (Equation~\ref{eq:MACregion}) 
for the cq-MAC. We do so with
the following sequence of steps. A full proof is given in 
Section~\ref{sec:cqMAC} below.
\begin{enumerate}

\item
Enlarge the Hilbert space $\cX \otimes \cY \otimes \cZ$ suitably and
consider a `perturbed' version $\cC'$ of the channel $\cC$. The channel
$\cC'$ maps an input pair $(x,y)$ to a state $\rho'_{x,y}$ close to
$\rho_{x,y}$. This gives a state
$(\rho')^{XYZ}$ close to $\rho^{XYZ}$. A codebook achieving a certain
rate point for channel $\cC'$ with a certain error achieves the same 
rate point for channel $\cC$ with only slightly more error. The state
$\rho'_{x,y}$ is obtained by augmenting $\rho_{x,y}$ and then tilting
it along `x' and `y' directions;

\item
Due to smoothing and augmentation described above,
$(\rho')^{XZ} \otimes (\rho')^Y$ and
$(\rho')^{YZ} \otimes (\rho')^X$
are extremely close to the states 
$\rho^{XZ} \otimes \rho^Y$ and
$\rho^{YZ} \otimes \rho^X$
tilted only along `x' and only along `y' respectively. The state
$(\rho')^X \otimes (\rho')^Y \otimes (\rho')^Z$ is extremely close
to the untilted state
$\rho^X \otimes \rho^Y \otimes \rho^Z$;

\item
We now define the `intersection' POVM element $\Pi$
to be the projector onto the complement of the tilted span of the 
complements of the supports of $\Pi^X$, $\Pi^Y$, $\Pi^{X,Y}$. 
Using Proposition~\ref{prop:tiltedspan}, we can show that
Property~\arabic{type1error} holds
with lower bound $1 - O(\epsilon)$ under the projector $\Pi$;

\item
Finally, we notice from the construction of $\Pi$ that 
Properties~\arabic{type2errorx}, 
\arabic{type2errory}, \arabic{type2errorxy} are easily
satisfied by $\Pi$ for 
$\rho^{XZ} \otimes \rho^Y$ tilted only along `x',
$\rho^{YZ} \otimes \rho^X$ tilted only along `y' and
untilted $\rho^X \otimes \rho^Y \otimes \rho^Z$,
with upper bounds exactly the same as in the ideal case. Since the states
$(\rho')^{XZ} \otimes (\rho')^Y$,
$(\rho')^{YZ} \otimes (\rho')^X$,
$(\rho')^X \otimes (\rho')^Y \otimes (\rho')^Z$
are extremely close to the tilted versions of
$\rho^{XZ} \otimes \rho^Y$,
$\rho^{YZ} \otimes \rho^X$,
$\rho^X \otimes \rho^Y \otimes \rho^Z$ described above,
we finally conclude that they satisfy 
Properties~\arabic{type2errorx}, 
\arabic{type2errory}, \arabic{type2errorxy} under $\Pi$
with upper bounds almost as good as in the ideal case;

\item
Thus, $\Pi$ can be used by Charlie in order to
construct a decoding algorithm for the original channel $\cC$ 
that achieves any rate pair in the 
region described in Equation~\ref{eq:MACregion} with average 
error probablity 
at most $O(\epsilon)$.

\end{enumerate}

The strategy outlined above enables us to prove the following 
{\em one-shot quantum joint typicality lemma}.
\begin{quote}
{\em 
Let $A_1 \ldots A_k$ be a $k$-partite quantum system with each 
$A_i$ isomorphic to a Hilbert space $\cH$.
Let $\rho^{A_1 \ldots A_k}$ be a
quantum state in $A_1 \ldots A_k$. 
For a subset $S \subseteq [k]$, let 
$A_S$ denote the systems $\{A_s: s \in S\}$.
Let $\rho^{A_S}$ denote the marginal state on $A_S$ obtained by tracing
out the systems in $\bar{S} := [k] \setminus S$ from 
$\rho^{A_1 \ldots A_k}$.
Let  $0 < \epsilon < 1$. Let $\cK$ be a Hilbert space of dimension
\[
\frac{2^{13(k+1)} (2 |\cH|)^{6k(k+1)}}{(1 - \epsilon)^{6(k+1)}}.
\]
There exists 
a state $\tau^{\cK^{\otimes [k]}}$ independent of $\rho^{A_{[k]}}$,
a state $(\rho')^{A'_{[k]}}$, and
a POVM element $\Pi'^{A'_{[k]}}$ on $A'_1 \ldots A'_k$ where
$A'_i \cong A_i \otimes \cK$, 
with the following properties:
\begin{enumerate}
\item
$
\ellone{
(\rho')^{A'_{[k]}} - \rho^{A_{[k]}} \otimes 
\tau^{\cK^{\otimes [k]}}
} \leq 
2^{\frac{k}{2}+1} \epsilon^{\frac{1}{4k}};
$

\item
$
\Tr [(\Pi')^{A'_{[k]}} (\rho')^{A'_{[k]}}] \geq 
1 - 2^{8 (k+1)^k} \epsilon^{\frac{1}{4}} 
 - 2^{\frac{k}{2} + 1} \epsilon^{\frac{1}{4k}};
$

\item
For every set $S$, $\{\} \neq S \subset [k]$, 
\[
\Tr [(\Pi')^{A'_{[k]}} 
     ((\rho')^{A'_{S}} \otimes (\rho')^{A'_{\bar{S}}})
    ] \leq 
2^{-D_H^\epsilon(\rho^{A_{[k]}} \| 
                 \rho^{A_S} \otimes \rho^{A_{\bar{S}}}
                ).
  }
\]
\end{enumerate}
}
\end{quote}

\subsection{Related work}
The bottleneck of simultaneous decoding was first pointed out
by Fawzi {\it et al\/}~\cite{fawzi:interference} in their paper on the
classical quantum interference channel. 
The authors obtained a simultaneous decoder for the two sender
cq-MAC in the asymptotic iid setting. At around the same time,
Xu and Wilde~\cite{xu:qmac} obtained a simultaneous decoder for 
sending classical
information over a two sender entanglement assisted quantum MAC in
the asymptotic iid setting.
Subsequently, Sen~\cite{sen:interference} 
obtained simultaneous decoders for two sender and a restricted type
of three sender cq-MACs in the asymptotic iid setting.
The simultaneous decoder for the restricted three sender cq-MAC sufficed 
to obtain the
so-called Chong-Motani-Garg-El Gamal inner bound for the classical
quantum interference channel in the asymptotic iid setting, matching
the best classical result for this channel.
None of these constructions 
extended to general three sender MACs, nor to the asymptotic non-iid
or the one shot settings. 
Simultaneous decoders have been recently used in several papers 
e.g. \cite{qi:simultaneous}, but the inner bounds obtained there are
suboptimal compared to known inner bounds when restricted to the
asymptotic iid setting. In a companion paper \cite{sen:simultaneous},
we apply our quantum joint typicality lemma to obtain one shot
simultaneous decoders for several multiterminal classical quantum
channels like broadcast and interference channels, matching the best
inner bounds known for these channels in the classical asymptotic iid
setting.

On a different note Dutil~\cite{dutil:phd}
pointed out a related bottleneck of `simultaneous smoothing',
which was further discussed by Drescher and 
Fawzi~\cite{drescher:simultaneous}. 
The aim of simultaneous smoothing is to replace 
a multipartite quantum state, e.g. $\rho^{AB}$, by another 
`perturbed' state
$(\rho')^{AB}$ which is $\delta \equiv \delta(\epsilon)$-close to 
it in trace distance such that 
\begin{eqnarray*}
\Hmin(A)_{\rho'} & \leq & c(\delta) \Hmin^\epsilon(A)_{\rho'}, \\
\Hmin(B)_{\rho'} & \leq & c(\delta) \Hmin^\epsilon(B)_{\rho'}, \\
\Hmin(AB)_{\rho'} & \leq & c(\delta) \Hmin^\epsilon(AB)_{\rho'},
\end{eqnarray*}
where $\delta(\epsilon)$ is a dimension independent function of 
$\epsilon$, 
$c(\delta)$ is a dimension independent function of $\delta$,
$\Hmin(\rho) := -\log \ellinfty{\rho}$,
$
\Hmin^\epsilon(\rho) := 
\max_{\rho': \ellone{\rho' - \rho} \leq \delta}
\Hmin(\rho').
$
The main difference between simultaneous smoothing and our notion of 
joint typicality is that joint typicality requires the existence of
a single perturbed state as well as a single POVM satisfying certain
properties, whereas there is no corresponding notion of POVM for 
simultaneous smoothing. Simultaneous smoothing seems more relevant
for so-called `covering type' problems whereas joint typicality seems 
more relevant for so-called `packing type' problems in Shannon theory.

Our joint typicality lemma can be viewed through the lens of quantum
hypothesis testing as a problem where there is one positive and
a set of negative hypotheses, each negative hypothesis begin 
a convex combination of tensor 
products of marginals of the positive hypothesis.
Related work has been done by 
Berta, Brand\~{a}o and Hirche~\cite{berta:hypothesis}. The main 
difference between their
work and ours is that they work in the asymptotic iid setting and try 
to find a single POVM that minimises the largest probability
of accepting a negative hypothesis from the negative hypothesis set 
while accepting all 
positive hypotheses from the positive hypothesis set with probability 
close to one. On the other hand, we work
in the one shot setting and
find a single POVM that accepts the positive hypothesis
with probability close to one and rejects all the negative hypotheses
with probabality at least as much as the optimal POVM for 
{\em each individual negative hypothesis}. Viewed in this fashion, our
result may seem to be more general. But actually the two results are
incomparable because the negative hypotheses in our result are of 
the form of tensor products of marginals of a state close to the positive
hypothesis state, wheares the negative hypotheses in Berta {\it et al\/} 
are of the form of convex combinations of iid states from any a priori 
given set.

\subsection{Organisation of the paper}
In the next section, we state some preliminary facts which will be
useful throughout the paper. Section~\ref{sec:tiltedspan} defines
and proves the union properties of the tilted span and $A$-tilted
span of subspaces. Section~\ref{sec:cqMACBaby} gives a self-contained
proof of a simplified one shot inner bound for the quantum multiple
access channel with two senders.
In Section~\ref{sec:qtypical} we prove the so-called
one-shot quantum joint typicality lemma, which manages to construct
a robust notion of intersection of POVM elements achieving a given set 
of hypothesis testing mutual information quantitites. The essential part of
the proof of the one-shot quantum joint typicality lemma is 
encapsulated into a proposition which is proved in 
\ref{sec:proofcqtypical}
We formally prove the achievability of the rate region described
by Equation~\ref{eq:MACregion} in Section~\ref{sec:cqMAC}.
Finally, we make some concluding remarks and list some open problems
in Section~\ref{sec:conclusions}.

\section{Preliminaries}
\label{sec:preliminaries}
All Hilbert spaces in this paper are finite dimensional. 
The symbol $\oplus$ always denotes the orthogonal direct
sum of Hilbert spaces.
For a subspace $X$ of a Hilbert space $\cH$, 
let $\Pi^{\cH}_X$ denote the orthogonal projection in $\cH$ 
onto $X$. When clear from the context, we may use $\Pi_X$ instead of
$\Pi^{\cH}_X$ for brevity of notation.

By a quantum 
state or a density matrix in a Hilbert space $\cH$, we mean a Hermitian,
positive semidefinite linear operator on $\cH$ with trace equal to one.
By a POVM
element $\Pi$ in $\cH$, we mean a Hermitian
positive semidefinite linear operator on $\cH$ with eigenvalues between
$0$ and $1$. Stated in terms of inequalities on Hermitian operators,
$\zero \leq \Pi \leq \one$, where $\zero$, $\one$ denote the zero and
identity operators on $\cH$.
In what follows, we shall use several times the Gelfand-Naimark theorem
which is stated below for completeness.
\begin{fact}[Gelfand-Naimark]
\label{fact:gelfandnaimark}
Let $\Pi$ be a POVM element in a Hilbert space $\cH$. 
For any integer $d \geq 2$,
there exists an orthogonal 
projection $\Pi'$ in $\cH \otimes \C^d$ such that 
\[
\Tr [\Pi A] = 
\Tr [\Pi' (A \otimes \ketbra{0}^{\C^d})]
\]
for all linear operators $A$ acting on $\cH$. Above, $\ket{0}$ is a
fixed vector, independent of $\Pi$ and $A$, in $\C^d$.
\end{fact}

Since quantum probability is a generalisation of classical probability,
one can talk of a so-called `classical POVM element', as done in the
introduction. Suppose we have a 
probability distribution $p(x)$, $x \in \cX$. A classical POVM element
on $\cX$ is a function $f: \cX \rightarrow [0,1]$. The probability of
accepting the POVM element $f$ is then $\sum_{x: x \in \cX} p(x) f(x)$.
One can continue to use the operator formalism for classical probablity
with the understanding that density matrices and POVM elements are now
diagonal matrices.

Let $\elltwo{v}$ denote the $\ell_2$-norm of a vector $v \in \cH$.
For an operator $A$ on $\cH$, we use $\ellone{A}$ to denote the
Schatten $\ell_1$-norm, also known as trace norm, of $A$, 
which is nothing but
the sum of singular values of $A$. We use $\ellinfty{A}$ to denote
the Schatten $\ell_\infty$-norm, also known as operator norm, 
of $A$, which is 
nothing but the largest singular value of $A$. For operators $A$, $B$
on $\cH$, we have the inequality 
\[
|\Tr [A B]| \leq \ellone{A B} \leq 
\min\{\ellone{A} \ellinfty{B}, \ellinfty{A} \ellone{B}\}.
\]

Let $\cX$ be a finite set. By a {\em classical-quantum} state on
$\cX \cH$ we mean a quantum state of the form
$
\rho^{\cX \cH} =
\sum_{x \in \cX} p_x \ketbra{x}^{\cX} \otimes \rho_x^{\cH},
$
where $x$ ranges over computational basis vectors of $\cX$ viewed
as a Hilbert space, $\{p_x\}_{x \in \cX}$ is a probability distribution
on $\cX$ and the operators $\rho_x$, $x \in \cX$ are quantum states in
$\cH$. We will also use the terminology that $\rho$ is classical on
$\cX$ and quantum on $\cH$.

For a positive integer $k$, we use $[k]$ to denote the set 
$\{1, 2, \ldots, k\}$. If $k = 0$, we define $[k] := \{\}$.
Let $c$ be a non-negative integer and $k$ a positive integer.
We shall study systems that are classical on $\cX^{\otimes [c]}$ and
quantum on $\cH^{\otimes [k]}$, where $[c]$ and $[k]$ are
treated as disjoint sets. We will use the notation
$[c] \cupdot [k]$ to denote the union $[c] \cup [k]$ keeping in mind
that $[c]$, $[k]$ are disjoint.
A {\em subpartition} $(S_1, \ldots, S_l)$ of $[k]$, denoted by
$(S_1, \ldots, S_l) \vdash [k]$, is a collection of 
non-empty, pairwise disjoint subsets of $[k]$. Note that the
order of subsets does not matter in defining a subpartition 
$(S_1, \ldots, S_l)$. 
For a subset $T \subseteq [c] \cupdot [k]$ satisfying 
$T \cap [k] \neq \{\}$,
we use the notation $(S_1, \ldots, S_l) \vdash \vdash T$ to
denote a so-called {\em pseudosubpartition} of $T$ intersecting 
$[k]$ non-trivially i.e. for
$i \in [l]$, $S_i \subseteq T$, 
$S_i \cap [k] \neq \{\}$ and
for $i \neq j \in [l]$, $S_i \cap S_j \cap [k] = \{\}$. 
There is a natural
partial order called the {\em refinement partial order} on the 
pseudosubpartitions of $T$ intersecting $[k]$ non-trivially, where 
$(S_1, \ldots, S_l) \preceq (R_1, \ldots, R_m)$ if
there is a function $f: [l] \rightarrow [m]$ such that
$S_i \subseteq R_{f(i)}$ for all $i \in [l]$. Under this partial order,
the pseudosubpartitions form a lattice with minimum element being the
{\em empty pseudosubpartition} with $l = 0$, and maximum element being the
{\em full block} with $l = 1$ and $S_1 = T$. It is easy to see that
the number of pseudosubpartitions of $T$ is at most
$2^{|T \cap [k]| |T \cap [c]|} (|T \cap [k]|+1)^{|T \cap [k]|}$, and
the number of pseudosubpartitions of $T$ with $l$ blocks is at most
$\frac{2^{l |T \cap [c]|} (l+1)^{|T \cap [k]|}}{l!}$.

Suppose we have a $k$-partite quantum system $A_1 \cdots A_k$.
For a non-empty set $S \subseteq [k]$, let 
$A_S$ denote the systems $\{A_s: s \in S\}$. 
Similarly, if $\rho$ is a quantum state
in $A_{[k]}$, $\rho^{A_S}$ will denote its marginal in $A_S$. Thus,
$\rho \equiv \rho^{A_{[k]}}$. 
If $\vecx$ is a computational basis vector of $\cX^{\otimes [c]}$,
for a subset $S \subseteq [c]$,
$\vecx_S$ will denote its restriction to the system $\cX^{\otimes S}$. 
Thus, $\vecx \equiv \vecx_{[c]}$. We also use $\vecx_S$ to denote
computational basis vectors of $\cX^{\otimes S}$ without reference to
the systems in $[c] \setminus S$. For a subset 
$S \subseteq [c] \cupdot [k]$ and a computational basis vector
$\vecl$ of $\cL^{\otimes [c] \cupdot [k]}$, the notation
$\vecl_S$ has the analogous meaning.
The notation
$(\cdot)^{\otimes S}$ denotes a tensor product only for the coordinates
in $S$. 

We recall the definition of the {\em hypothesis testing relative
entropy} as given by Wang and Renner~\cite{wang:DepsH}. Very similar
quantities were defined and used in earlier 
works~\cite{buscemi:qchannel, brandao:entanglement}.
\begin{definition}
\label{def:hyptestingrelentropy}
Let $\alpha$, $\beta$ be two quantum states in the same Hilbert space.
Let $0 \leq \epsilon < 1$.
Then the {\em hypothesis testing relative entropy} of $\alpha$ with respect
to $\beta$ is defined by
\[
D^{\epsilon}_H(\alpha \| \beta) :=
\max_{\Pi: \Tr [\Pi \alpha] \geq 1 - \epsilon}
-\log \Tr [\Pi \beta],
\]
where the maximisation is over all POVM elements $\Pi$ acting on 
the Hilbert space.
\end{definition}
The definition quantifies the minimum probability of `accepting' $\beta$
by a POVM element $\Pi$ that `accepts' $\alpha$ with probability
at least $1 - \epsilon$.
From the definition, it is
easy to see that if $\epsilon < \epsilon'$,
$D^{\epsilon}_H(\alpha \| \beta) < D^{\epsilon'}_H(\alpha \| \beta)$.
We will require the following property of the so-called
{\em hypothesis testing mutual information}.
\begin{proposition}
\label{prop:maxIH}
Let $0 \leq \epsilon < 1$.
Let $\rho^{AB}$ be a quantum state in a bipartite system $AB$. Define
the {\em hypothesis testing mutual information} 
$
I^\epsilon_H(A : B)_\rho := 
D^\epsilon_H(\rho^{AB} \| \rho^A \otimes \rho^B).
$
Then,
\[
I^\epsilon_H(A : B)_\rho \leq
2 \log \min\{|A|, |B|\} + 3 \log \frac{1}{1-\epsilon} + 6 \log 3 - 4.
\]
\end{proposition}
\begin{proof}
Let $C$ be a third system and $\rho^{ABC}$ a purification of $\rho^{AB}$.
Let $\Pi^{AB}$ be the optimising POVM element for 
$
D^\epsilon_H(\rho^{AB} \| \rho^A \otimes \rho^B).
$
Define $\Pi^{ABC} := \Pi^{AB} \otimes \one^C$. Let $|A| \leq |B|$ without
loss of generality.
We will show that 
\[
\Tr [\Pi^{ABC} (\rho^A \otimes \rho^{BC})] \geq 
\frac{2^4 (1-\epsilon)^3}{3^6 |A|^2} 
\]
which would prove the proposition.

Let $\ket{\rho}^{ABC}$ denote the pure state $\rho^{ABC}$ in ket vector
form. Let
\[
\ket{\rho}^{ABC} =
\sum_{a=1}^{|A|} \sqrt{p_a} \, \ket{x_a}^A \otimes \ket{y_a}^{BC}
\]
be a Schmidt decomposition for the bipartition $(A, BC)$, where
$\sum_{a=1}^{|A|} p_a = 1$. 
Fix $0 < \delta < 1$.
Define
\[
\hat{A} := 
\left\{a \in [|A|]:
p_a \geq \frac{\delta (1-\epsilon)}{|A|}
\right\}.
\]
Using the triangle inequality, we get
\begin{eqnarray*}
\sqrt{1-\epsilon} 
& \leq &
\elltwo{\Pi^{ABC} \ket{\rho}^{ABC}}  \\
& \leq &
\elltwo{
\Pi^{ABC} 
(
\sum_{a \in \hat{A}} \sqrt{p_a} \ \,
\ket{x_a}^A \otimes \ket{y_a}^{BC}
)
} \\
&     &
{} +
\elltwo{
\Pi^{ABC} 
(
\sum_{a \not \in \hat{A}} \sqrt{p_a} \
\ket{x_a}^A \otimes \ket{y_a}^{BC}
)
} \\
& \leq &
\sum_{a \in \hat{A}} \sqrt{p_a} \
\elltwo{\Pi^{ABC} (\ket{x_a}^A \otimes \ket{y_a}^{BC})} \\
&      &
{} +
\elltwo{
\sum_{a \not \in \hat{A}} \sqrt{p_a} \
\ket{x_a}^A \otimes \ket{y_a}^{BC}
} \\
&   =  &
\sum_{a \in \hat{A}} \sqrt{p_a} \ 
\elltwo{\Pi^{ABC} (\ket{x_a}^A \otimes \ket{y_a}^{BC})} +
\sqrt{\sum_{a \not \in \hat{A}} p_a} \\
&   <  &
\sum_{a \in \hat{A}} \sqrt{p_a} \
\elltwo{\Pi^{ABC} (\ket{x_a}^A \otimes \ket{y_a}^{BC})} +
\sqrt{\delta (1 - \epsilon)},
\end{eqnarray*}
which implies that 
\[
\sum_{a \in \hat{A}} \sqrt{p_a} \
\elltwo{\Pi^{ABC} (\ket{x_a}^A \otimes \ket{y_a}^{BC})} >
\sqrt{1-\epsilon} (1 - \sqrt{\delta}).
\]
Using the Cauchy-Schwarz inequality, we get
\[
\sqrt{1-\epsilon} (1 - \sqrt{\delta}) <
\sqrt{\sum_{a \in \hat{A}} p_a} \,
\sqrt{
\sum_{a \in \hat{A}} 
\elltwo{\Pi^{ABC} (\ket{x_a}^A \otimes \ket{y_a}^{BC})}^2
} 
\]
which implies that
\[
\sqrt{
\sum_{a \in \hat{A}} 
\elltwo{\Pi^{ABC} (\ket{x_a}^A \otimes \ket{y_a}^{BC})}^2
} >
\sqrt{1-\epsilon} (1 - \sqrt{\delta}).
\]

Now,
\begin{eqnarray*}
\lefteqn{
\Tr [\Pi^{ABC} (\rho^A \otimes \rho^{BC})]
} \\
& = &
\Tr \left[
\Pi^{ABC}
\left(\sum_{a=1}^{|A|} p_a \ketbra{x_a}\right)^A \otimes
\left(\sum_{a'=1}^{|A|} p_{a'} \ketbra{y_{a'}}\right)^{BC}
\right] \\
& \geq &
\sum_{a \in \hat{A}} p_a^2
\Tr [
\Pi^{ABC}
(\ketbra{x_a}^A \otimes \ketbra{y_a}^{BC})
] \\
&  =   &
\sum_{a \in \hat{A}} p_a^2
\elltwo{
\Pi^{ABC} (\ket{x_a}^A \otimes \ket{y_a}^{BC})
}^2 \\
& \geq &
\frac{\delta^2 (1 - \sqrt{\delta})^2 (1 - \epsilon)^3}{|A|^2}.
\end{eqnarray*}
Taking $\delta = 4/9$ gives the largest lower bound above, and
completes the proof of the proposition.
\end{proof}

We will also need the so-called {\em non-commutative union bound} first
proved by Sen~\cite{sen:interference}, and later improved by
Gao~\cite{gao:union} and Oskouei {\it et al\/}~\cite{oskouei:union}.
\begin{fact}[Noncommutative union bound]
\label{fact:noncommutativeunionbound}
Let $\rho$ be a positive semidefinite operator on a Hilbert 
space $\cH$ and $\Tr \rho \leq 1$. 
Let $\Pi'_1, \ldots, \Pi'_k$ be
orthogonal projectors in $\cH$. Define $\Pi_i := \one - \Pi'_i$. Then,
\[
\Tr [\Pi'_k \cdots \Pi'_1 \rho \Pi'_1 \cdots \Pi'_k] \geq 
\Tr \rho - 4 \sum_{i=1}^k \Tr [\rho \Pi_i].
\]
\end{fact}

\section{Tilted span of subspaces}
\label{sec:tiltedspan}
Let $\cH$ be a Hilbert space. Consider Hilbert spaces $\cH_j$, $j \in [l]$,
each of dimension $\dim \cH$.
Define $\tcH := \cH \oplus \cH_1 \oplus \cdots \oplus \cH_l$.
Let $\cT_j$, $j \in [l]$ be linear maps mapping $\cH$ isometrically onto
$\cH_j$. 
For $j \in [l]$, let $0 < \alpha_j < 1$ and
define linear maps  
$\cT_{j, \alpha_j}: \cH \rightarrow \tcH$ as
\[
\cT_{j, \alpha_j} := 
\sqrt{1- \alpha_j} \one_{\cH} + 
\sqrt{\alpha_j} \cT_j,
\]
where $\one_{\cH}$ is the identity embedding of $\cH$ into $\tcH$.
Observe that each $\cT_{j, \alpha_j}$ is an isometric embedding of 
$\cH$ into $\tcH$.

We now define the tilted span of a collection of subspaces, formalising
the intuitive description in the introduction. 
\begin{definition}[Tilted span]
\label{def:tiltedspan}
Let $W_1, \ldots, W_l$ be subspaces of $\cH$.
Let $0 < \alpha_j < 1$, $j \in [l]$.
The subspace $\cT_{j, \alpha_j}(W_{j}) \leq \tcH$ will be called the
{\em $\alpha_j$-tilt} of $W_j$ along the $j$th direction. 
The subspace $\bW_{(\alpha_1, \ldots, \alpha_l)}$ of $\tcH$ defined by 
\[
\bW_{(\alpha_1, \ldots, \alpha_l)} := 
\bigplus_{j=1}^{l} 
\cT_{j, \alpha_j}(W_{j}).
\]
is called the {\em  $(\alpha_1, \ldots, \alpha_l)$-tilted span} of 
$W_1, \ldots, W_l$. If $\alpha_j = \alpha$ for all $j \in [l]$, we
write $\bW_\alpha$ and call it the {\em $\alpha$-tilted span} of
$W_1, \ldots, W_l$.
\end{definition}

The effect of {\em tilting} the subspaces  $W_1, \ldots, W_l$ along $l$
orthogonal directions is to increase the separation between them,
in a sense made precise by the following proposition.
\begin{proposition}
\label{prop:tiltedspan}
Let $0 < \alpha \leq 1$.
Let $\ket{h} \in \cH$ be a unit vector. Let $W_1, \ldots, W_l$ be
subspaces of $\cH$.
For $j \in [l]$, define
$
\epsilon_j := \elltwo{\Pi_{W_j} \ket{h}}^2
$
and 
let $0 < \alpha_j < 1$. Define $\alpha := \min_{j: j \in [l]} \alpha_j$.
Then,
\[
\max_{j: j \in [l]} (1 - \alpha_j) \epsilon_j 
\leq
\elltwo{\Pi_{\bW_{(\alpha_1, \ldots, \alpha_l)}} \ket{h}}^2 
\leq 
\frac{1-\alpha}{\alpha} \sum_{j=1}^l \epsilon_j.
\]
\end{proposition}
\begin{proof}
The lower bound follows trivially since
the projection of $\ket{h}$ onto the $j$th summand space in
the definition of $\bW_{(\alpha_1, \ldots, \alpha_l)}$ is of length 
$\sqrt{\epsilon_j (1- \alpha_j)}$. 

We now prove the upper bound.
Define $h' := \Pi_{\bW_{(\alpha_1, \ldots, \alpha_l)}} \ket{h}$. 
Since $h' \in \bW_{(\alpha_1, \ldots, \alpha_l)}$, let
$h' = \sum_{j=1}^l \lambda_j \ket{x_j}$ where 
$\lambda_j \in \C$, and $\ket{x_j}$ is a unit vector in 
$\cT_{j, \alpha_j}(W_j)$ for $j \in [l]$.
Let $\ket{x_j} =  \cT_{j,\alpha_j}(\ket{y_j})$,
where $\ket{y_j}$ is a unit vector in $W_j$ and is
uniquely determined by $\ket{x_j}$. 
Then,
\begin{eqnarray*}
\elltwo{h'}^2 
& = &
\elltwo{
\sum_{j=1}^l \sqrt{1 - \alpha_j} \lambda_j \ket{y_j} + 
\sum_{j=1}^l \sqrt{\alpha_j} \lambda_j \cT_j(\ket{y_j})
}^2 \\
& \geq &
\elltwo{
\sum_{j=1}^l \sqrt{\alpha_j} \lambda_j \cT_j(\ket{y_j})
}^2 \\
& \geq &
\alpha \sum_{j=1}^l |\lambda_j|^2.
\end{eqnarray*}

We also have
\[
\elltwo{h'}^2 =  
|\braket{h}{h'}| = 
\left|\sum_{j=1}^l \lambda_j \braket{h}{x_j}\right| \leq 
\sqrt{\sum_{j=1}^l |\lambda_j|^2} \cdot
\sqrt{\sum_{j=1}^l |\braket{h}{x_j}|^2}.
\]
This implies that
\begin{eqnarray*}
\alpha
\sqrt{
\displaystyle
\sum_{j=1}^l 
|\lambda_j|^2 
} 
& \leq &
\sqrt{\sum_{j=1}^l |\braket{h}{x_j}|^2} 
\; = \;
\sqrt{
\sum_{j=1}^l (1-\alpha_j) |\braket{h}{y_j}|^2
} \\
&   =  &
\sqrt{
\sum_{j=1}^l (1-\alpha_j) |\bra{h}\Pi_{W_j} \ket {y_j}|^2
} \\
& \leq &
\sqrt{
\sum_{j=1}^l (1-\alpha_j) \elltwo{\Pi_{W_j} \ket{h}}^2
} \\
& \leq &
\sqrt{1-\alpha}
\sqrt{\sum_{j=1}^l \epsilon_j}.
\end{eqnarray*}
Thus,
\begin{eqnarray*}
\elltwo{h'}^2 
& \leq  &
\sqrt{ \sum_{j=1}^l |\lambda_j|^2} \cdot
\sqrt{\sum_{j=1}^l |\braket{h}{x_j}|^2} 
\;\leq\;
\frac{1}{\alpha} \sum_{j=1}^l |\braket{h}{x_j}|^2 \\
& \leq &
\frac{1-\alpha}{\alpha} \sum_{j=1}^l \epsilon_j.
\end{eqnarray*}
This proves the desired upper bound on $\elltwo{h'}^2$
and completes the proof of the proposition.
\end{proof}

As a corollary, we now prove a small extension of 
Proposition~\ref{prop:tiltedspan}.
\begin{corollary}
\label{cor:tiltedspan}
Let $\ket{h} \in \cH$ be a unit vector. Let $W_0, W_1, \ldots, W_l$ be
subspaces of $\cH$.
For $0 \leq j \leq l$, define
$
\epsilon_j := \elltwo{\Pi_{W_j} \ket{h}}^2.
$
Let $0 < \alpha < 1/3$.
Define the subspace $\bW$ of $\tcH$ as
$\bW := \bW_{\alpha} + W_0$.
Define $\epsilon := \frac{1-\alpha}{\alpha} \sum_{j=1}^l \epsilon_j$.
Then,
\[
\max\{\epsilon_0, (1-\alpha)(\max_{j: 1 \leq j \leq l} \epsilon_j)\}
\leq
\elltwo{\Pi_{\bW} \ket{h}}^2 
\leq 
\frac{3l}{\alpha} (\epsilon_0 + \epsilon). 
\]
\end{corollary}
\begin{proof}
As also remarked in the proof of Proposition~\ref{prop:tiltedspan}, the
lower bound is immediate. Hence, we concentrate on proving the upper
bound.

It is easy to see that the largest inner product between a vector 
$\ket{v} \in \cH$ and a 
vector $\ket{w} \in \bigplus_{j=1}^l \cT_{j,\alpha}(\cH)$ is given by
\[
\ket{w} =
\frac{1}{\sqrt{l^2 (1 - \alpha) + l \alpha}}
(
l \sqrt{1-\alpha} \ket{v} +
\sqrt{\alpha} \sum_{j=1}^l \cT_j(\ket{v})
),
\]
with inner product equal to
\[
\frac{\sqrt{l(1-\alpha)}}{\sqrt{l(1-\alpha) + \alpha}} \leq
1 - \frac{\alpha}{3 l}.
\]

Define $h' := \Pi_{\bW} \ket{h}$. Since $h' \in \bW$, let 
$h' = h'_0 + h'_\alpha$, where $h'_0 \in W_0$ and
$h'_\alpha \in \bW_\alpha$. Then,
\begin{eqnarray*}
\elltwo{h'}^2 
&   =  &
\elltwo{h'_0}^2 + \elltwo{h'_\alpha}^2 +
2 \Re (\braket{h'_0}{h'_\alpha}) \\
& \geq &
\elltwo{h'_0}^2 + \elltwo{h'_\alpha}^2 -
2 \elltwo{h'_0} \elltwo{h'_\alpha} 
\left(
1 - \frac{\alpha}{3 l}
\right) \\
& \geq &
\left(\elltwo{h'_0}^2 + \elltwo{h'_\alpha}^2\right)
\frac{\alpha}{3 l},
\end{eqnarray*}
where, in the last inequality, we used the fact that 
$2 x y \leq x^2 + y^2$ for real numbers $x$, $y$.
On the other hand,
\begin{eqnarray*}
\elltwo{h'}^2 
& = &
|\braket{h}{h'}| \\
& \leq &
|\braket{h}{h'_0}| + |\braket{h}{h'_\alpha}| \\
& \leq &
\elltwo{h'_0} \sqrt{\epsilon_0} + \elltwo{h'_\alpha} \sqrt{\epsilon} \\
& \leq &
\sqrt{\elltwo{h'_0}^2 + \elltwo{h'_\alpha}^2} \cdot
\sqrt{\epsilon_0 + \epsilon}.
\end{eqnarray*}
In the second inequality above, we used the properties that 
$\elltwo{\Pi_{\bW_\alpha} \ket{h}}^2 \leq \epsilon$ and 
$\elltwo{\Pi_{W_0} \ket{h}}^2 = \epsilon_0$.
We thus get
\[
\sqrt{\elltwo{h'_0}^2 + \elltwo{h'_\alpha}^2} \leq
\frac{3l}{\alpha} \sqrt{\epsilon_0 + \epsilon},  
\]
which means
\[
\elltwo{h'}^2 \leq 
\frac{3l}{\alpha} (\epsilon_0 + \epsilon). 
\]
This finishes the proof of the corollary.
\end{proof}

\noindent
{\bf Remark:} 

\noindent
1.\ An easy corollary of Proposition~\ref{prop:tiltedspan} 
is a union bound for
projectors that improves upon the union bound proved in Anshu, Jain and
Warsi's paper \cite[Lemma~3]{anshu:compound} on the 
entanglement assisted compound quantum channel.

\noindent
2.\ Proposition~\ref{prop:tiltedspan} also gives an easier proof of
the main result of Harrow, Lin and Montanaro's paper 
\cite[Corollary~11]{HarrowLinMontanaro} which they and earlier authors have
referred to as a `Quantum OR bound' for projectors. The `OR projector'
constructed by Proposition~\ref{prop:tiltedspan} improves the parameters
of Harrow {\it et al\/}'s construction and is more 
efficient to implement.

\begin{corollary}
\label{cor:union}
Let $\epsilon, \alpha > 0$. For $i \in [l]$, let $\rho_i$ be a
quantum state and $\Pi_i$ be an orthogonal projector in $\cH$  such that
$\Tr [\Pi_i \rho_i] \geq 1 - \epsilon$. Let $\tcH$ be defined as in
Definition~\ref{def:tiltedspan}. Then there is a projector 
$\hPi \in \tcH$ such that $\Tr [\hPi \rho_i] \geq 1 - \epsilon - \alpha$
for all $i \in [l]$, and, for all states $\sigma \in \cH$,
\[
\Tr [\hPi \sigma] \leq 
\frac{1-\alpha}{\alpha}
\sum_{i \in [l]} \Tr [\Pi_i \sigma].
\]
\end{corollary}
Suppose $ \Tr [\Pi_i \sigma] \leq \theta$ for all $i \in [l]$. Then
the upper bound promised by the above corollary 
is $\frac{l \theta}{\alpha}$ whereas the upper bound given by
Lemma~3 of \cite{anshu:compound} is only 
$l \theta (\frac{2}{\alpha^2})^{\log (2l)}$.
Thus, the above corollary leads to corresponding improvements in the 
achievable rates for the various settings considered in
\cite{anshu:compound}.

It is interesting to consider the 
analogoue of Proposition~\ref{prop:tiltedspan} in classical 
probability. One can 
think of the subspaces $W_j$, $j \in [l]$ as a set of $l$ events, 
the vector $\ket{h}$ as the classical state of a system, and 
$\epsilon_j$ to be the probablity of the $j$th event occuring. Then
the probability of at least one of the events occuring is lower bounded by
$\max_{j: j \in [l]} \epsilon_j$ and upper bounded by
$\sum_{j=1}^l \epsilon_j$. This is nothing but the fundamental
{\em union bound}  of classical probability. Thus, the above proposition
almost recovers the classical performance of the union bound in the
context of quantum probability.

However the above proposition still falls short of defining an
intersection projector satisfying 
Properties \arabic{type1error}, \arabic{type2errorx}, \arabic{type2errory},
\arabic{type2errorxy} for the cq-MAC.
For $j \in [l]$, let $\Pi_j$ be an orthogonal projection.
We would like to define a non-trivial intersection of the projectors
$\Pi_j$, $j \in [l]$.
Define a new projector $\hat{\Pi}$ as the projection onto the
orthogonal complement of 
the tilted span $\bW_\alpha$ of $W_j$, $j \in [l]$, where $W_j$ is 
taken to be the support of $\Pi_j$. Let $\rho$ be a quantum state
and $\epsilon_j := \Tr [\Pi_j \rho]$.
Then it is easy to see that
\[
\Tr [\hat{\Pi} \rho] \geq 
1 - \frac{1-\alpha}{\alpha} \sum_{j=1}^l \epsilon_j.
\]
This is tolerable for most applications of the joint typicality lemma.
In fact, it allows us to prove a good enough version of 
Property~\arabic{type1error} for the cq-MAC taking $\rho = \rho^{XYZ}$,
$l = 3$, $\Pi_1 := \Pi^X$, $\Pi_2 := \Pi^Y$, $\Pi_3 := \Pi^{X,Y}$.

Now consider states $\sigma_i$, $i \in [l]$. Let 
$2^{-k_i} := \Tr [\Pi_i \sigma_i]$. 
For the cq-MAC, we have $l = 3$, 
$\sigma_1 := \rho^{YZ} \otimes \rho^X$,
$\sigma_2 := \rho^{XZ} \otimes \rho^Y$,
$\sigma_3 := \rho^X \otimes \rho^Y \otimes \rho^Z$.
For the upper bound, the best that one can prove is
\[
\Tr [\hat{\Pi} \sigma_i] \leq 
\alpha (1 - 2^{-k_i}) + 2^{-k_i}.
\]
The additive term of nearly $\alpha$ makes the upper bound
insufficient for applicatioins of the joint typicality lemma.
In particular, it removes any hope of proving even approximate versions of
Properties \arabic{type2errorx}, \arabic{type2errory},
\arabic{type2errorxy} for the cq-MAC.

However, for a $k$-partite state $\rho^{A_1 \cdots A_k}$, and a 
collection of $k$-partite states 
$\sigma_i := \rho^{A_{S_i}} \otimes \rho^{A_{\bar{S_i}}}$, 
$\{\} \neq S_i \subset [k]$, $i \in [l]$,
it turns out that we can do better and
indeed come up with a notion of intersection projector
that is strong enough to prove a quantum joint typicality lemma.
For this, as discussed in the introduction, we have to do a 
{\em smoothing} of $\rho$ to $\rho'$. The smoothing is achieved by
a carefully chosen tilt of $\rho$. This makes 
$\sigma'_i := (\rho')^{A_{S_i}} \otimes (\rho')^{A_{\bar{S_i}}}$ 
extremely close to a certain tilted version of $\sigma_i$. 
However, it turns out that
the tilt of $\sigma_i$ is
not only along the so-called $(S_i, \bar{S_i})$th direction but 
also along the
directions corresponding to subpartitions refining $(S_i, \bar{S_i})$. 
Thus, by taking a linear extension of the partial order of subpartitions
of $[k]$ under refinement, we are led to consider a
tilting scheme where $\sigma_i$ is tilted along the directions 
$1, \ldots, i$. This tilting scheme is described by an upper triangular
tilting matrix $A$ whose $(ij)$th entry $\alpha_{ij}$
denotes the tilt along the $i$th direction for $\sigma_j$ when
$i \leq j$. Hence, we have to formally define the 
$A$-tilted span
of $W_1, \ldots, W_l$, and prove a union bound for it.

Let $0 \leq \alpha_{ij} \leq 1$, where $i, j \in [l]$. 
Define the $l \times l$ matrix $A := (\alpha_{ij})$.
Suppose $A$ is {\em upper triangular} and {\em diagonal dominated}, that 
is, $\alpha_{ij} = 0$ for $i > j$ and
$\alpha_{ii} \geq \alpha_{ij}$ for all $i \leq j$.
Furthermore, suppose $A$ is {\em substochastic}, that is, 
for all $j$, $\sum_{i=1}^j \alpha_{ij} \leq 1$. If the last inequality
holds with equality, we say that $A$ is {\em stochastic}.
For $j \in [l]$, define linear maps  
$\cT_{j, A}: \cH \rightarrow \tcH$ as
\[
\cT_{j, A} := 
\sqrt{1- \sum_{i=1}^j \alpha_{i j}} \; \one_{\cH} + 
\sum_{i=1}^j \sqrt{\alpha_{i j}} \, \cT_{i}.
\]
Observe that each $\cT_{j, A}$ is an isometric 
embedding of $\cH$ into $\tcH$, which we shall refer to as the $A$-tilt
along the $j$th direction. We shall call $A$ as the {\em tilting
matrix}. A tilting matrix is always assumed to be upper triangular,
diagonal dominated, and substochastic.
\begin{definition}[$A$-tilted span]
\label{def:generalisedtiltedspan}
Let $W_1, \ldots, W_l$ be subspaces of $\cH$.
The subspace $\cT_{j, A}(W_{j}) \leq \tcH$ will be called the
{\em $A$-tilt} of $W_j$ along the directions $1, \ldots, j$. 
The subspace $\bW_{A}$ of $\tcH$ defined by 
\[
\bW_{A} := 
\bigplus_{j=1}^{l} 
\cT_{j, A}(W_{j}).
\]
is called the {\em  $A$-tilted span} of $W_1, \ldots, W_l$.
\end{definition}

The effect of doing $A$-tilting of the subspaces  
$W_1, \ldots, W_l$ along $l$
orthogonal directions is to increase the separation between them,
in a sense made precise by the following proposition.
\begin{proposition}
\label{prop:generalisedtiltedspan}
Let $A$ be an upper triangular diagonal domniated substochastic matrix.
Let $\ket{h} \in \cH$ be a unit vector. 
For $j \in [l]$, define
$
\epsilon_j := \elltwo{\Pi_{W_j} \ket{h}}^2.
$
Then,
\begin{eqnarray*}
\max_{j: j \in [l]} \epsilon_j 
\left(1- \sum_{i=1}^j \alpha_{ij}\right) 
& \leq &
\elltwo{\Pi_{\bW_{A}} \ket{h}}^2 \\
& \leq &
\left(
\sum_{j=1}^l \sqrt{\epsilon_j} 
\left(\sum_{k=j}^l 2^{k-j} \alpha_{kk}^{-1/2} \right)
\right)^2.
\end{eqnarray*}
In particular, if $\alpha_{jj} \geq \alpha$
for all $j \in [l]$, we can obtain an upper bound of
\[
\elltwo{\Pi_{\bW_{A}} \ket{h}}^2 \leq
\frac{2^{2l+1}}{\alpha} \sum_{j=1}^l \epsilon_j.
\]
\end{proposition}
\begin{proof}
The lower bound follows trivially since
the projection of $\ket{h}$ onto the $j$th summand space in
the definition of $\bW_{A}$ is of length 
$
\sqrt{\epsilon_j \left(1- \sum_{i=1}^j \alpha_{ij}\right)}.
$

We now prove the upper bound.
Define $h' := \Pi_{\bW_{A}} \ket{h}$. 
Since $h' \in \bW_{A}$, let
$h' = \sum_{j=1}^l \lambda_j \ket{x_j}$ where 
$\lambda_j \in \C$, and $\ket{x_j}$ is a unit vector in 
$\cT_{j, A}(W_j)$ for $j \in [l]$.
Let $\ket{x_j} =  \cT_{j,A}(\ket{y_j})$,
where $\ket{y_j}$ is a unit vector in $W_j$ and is
uniquely determined by $\ket{x_j}$. 
Since the spaces $\cH_j$, $j \in [l]$ are orthogonal to $\cH$, we have
\begin{eqnarray*}
\elltwo{h'} 
& = &
\elltwo{
\sum_{j=1}^l \sqrt{1 - \sum_{i=1}^j \alpha_{ij}} \,
\lambda_j \ket{y_j} + 
\sum_{j=1}^l \sum_{i=1}^{j} \sqrt{\alpha_{ij}} 
\lambda_j \cT_i(\ket{y_j})
} \\
& \geq  &
\elltwo{
\sum_{j=1}^l \sum_{i=1}^{j} \sqrt{\alpha_{ij}} 
\lambda_j \cT_i(\ket{y_j})
}.
\end{eqnarray*}

We now prove by backward induction on $i$ that 
\begin{equation}
\label{eq:lambdai}
|\lambda_i| \leq 
\elltwo{h'} 
\left(\sum_{j=i}^l 2^{j-i} \alpha_{jj}^{-1/2}\right).
\end{equation}
The base case of $i = l+1$ is vacuously true taking $\lambda_{l+1} = 0$.
Suppose the claim is true for $i+1$. We now prove it for $i$. Observe that
\begin{eqnarray*}
\sqrt{\alpha_{ii}} |\lambda_i| -
\sum_{j=i+1}^l \sqrt{\alpha_{ij}} |\lambda_j|
& \leq &
\elltwo{
\sum_{j=i}^l \sqrt{\alpha_{ij}} \lambda_j \cT_i(\ket{y_j})
}    \\
& \leq &
\elltwo{
\sum_{i=1}^{l} \sum_{j=i}^l \sqrt{\alpha_{ij}} 
\lambda_j \cT_i(\ket{y_j})
}    \\
& \leq &
\elltwo{h'},
\end{eqnarray*}
where we used the fact that $\cH_i$ is orthogonal to $\cH_{i'}$ for all
$i' \neq i$ in the second inequality.
Using the induction hypothesis and the diagonal dominated property, we get
\begin{eqnarray*}
|\lambda_i| 
& \leq &
\frac{\elltwo{h'}}{\sqrt{\alpha_{ii}}} + 
\sum_{j=i+1}^l \frac{\sqrt{\alpha_{ij}}}{\sqrt{\alpha_{ii}}} 
|\lambda_j| \\
& \leq &
\frac{\elltwo{h'}}{\sqrt{\alpha_{ii}}} + 
\sum_{j=i+1}^l |\lambda_j| \\
& \leq &
\elltwo{h'} 
\left(
\alpha_{ii}^{-1/2} 
+ 
\sum_{j=i+1}^l \sum_{k=j}^l 2^{k-j} \alpha_{kk}^{-1/2}
\right).
\end{eqnarray*}
Rearranging, we get
\begin{eqnarray*}
|\lambda_i| 
& \leq &
\elltwo{h'} 
\left(
\alpha_{ii}^{-1/2} 
+ 
\sum_{k=i+1}^l \sum_{j=i+1}^{k} 2^{k-j} \alpha_{kk}^{-1/2} 
\right) \\
& \leq &
\elltwo{h'} 
\left(
\alpha_{ii}^{-1/2} 
+ 
\sum_{k=i+1}^l 2^{k-i} \alpha_{kk}^{-1/2} 
\right),
\end{eqnarray*}
which completes the proof of the inequality in (\ref{eq:lambdai}) 
by induction.

We also have
\begin{eqnarray*}
\elltwo{h'}^2 
&   =  &  
|\braket{h}{h'}| 
\;  = \; 
\left|\sum_{j=1}^l \lambda_j \braket{h}{x_j}\right| \\
& \leq & 
\sum_{j=1}^l |\lambda_j| |\braket{h}{x_j}|
\;  = \;
\sum_{j=1}^l |\lambda_j| \sqrt{1-\sum_{i=1}^j \alpha_{ij}} \; 
|\braket{h}{y_j}| \\
& \leq &
\sum_{j=1}^l |\lambda_j| |\bra{h}\Pi_{W_j} \ket {y_j}| 
\;\leq\;
\sum_{j=1}^l |\lambda_j| \elltwo{\Pi_{W_j} \ket {y_j}} \\
&   =  &
\sum_{j=1}^l |\lambda_j| \sqrt{\epsilon_j}. 
\end{eqnarray*}
Combining the above inequality with the inequality in 
(Equation~\ref{eq:lambdai}), we get
\[
\elltwo{h'} \leq  
\sum_{j=1}^l \left(\sum_{k=j}^l 2^{k-j} \alpha_{kk}^{-1/2}\right) 
\sqrt{\epsilon_j},
\]
which completes the proof of the proposition. The special case
where $\alpha_{jj} \geq \alpha$ for all $j \in [l]$ follows via
Cauchy-Schwarz inequality.
\end{proof}
\begin{figure*}[!t]
\centering{
\includegraphics[width=.9\textwidth]{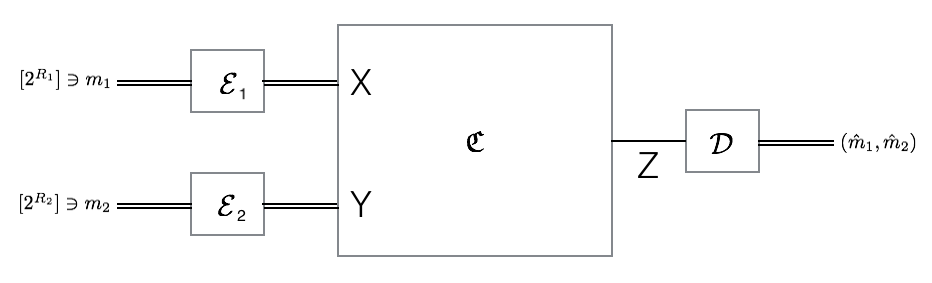}
}
\caption{A classical quantum multiple access channel. Alice, Bob 
encode their
respective messages $m_1$, $m_2$ into classical codewords and then 
feed them to the channel.
Charlie applies a decoding superoperator to the quantum output of 
the channel to get his guess $(\hat{m}_1, \hat{m}_2)$ for the transmitted
messages. $X$, $Y$
are the input alphabets and $Z$ is the output Hilbert space of the 
channel.
}
\label{fig:QMACa}
\end{figure*}

\medskip

\noindent
{\bf Remark:} 

\noindent
It is interesting to contrast the claims of
Propositions~\ref{prop:tiltedspan} and \ref{prop:generalisedtiltedspan}.
In the case where $\alpha_{jj} = \alpha$
for all $j \in [l]$ and $\alpha_{jj'} = 0$ for $j \neq j'$, 
Proposition~\ref{prop:tiltedspan} gives an upper bound of
$\frac{1-\alpha}{\alpha}\sum_{j=1}^l \epsilon_j$.
On the other hand, Proposition~\ref{prop:generalisedtiltedspan} gives a
worse upper bound of $\frac{2^{2l+1}}{\alpha}\sum_{j=1}^l \epsilon_j$. 
The power of
Proposition~\ref{prop:generalisedtiltedspan} lies in its generality.
As discussed earlier, it turns out that we need the general setting of
Proposition~\ref{prop:generalisedtiltedspan} in order to construct
an intersection projector strong enough to prove 
the one-shot quantum joint typicality lemma.
In other words, Proposition~\ref{prop:generalisedtiltedspan}
can be thought of as a way to do intersection of hypothesis tests in the
quantum setting. Proposition~\ref{prop:tiltedspan} will
be used in situations where one needs to consider union of quantum
hypothesis
tests in the quantum setting. Such situations arise in several network
information theoretic tasks on top of joint typicality requirements.
Some concrete applications where both intersection and union of
hypothesis tests are required can be found in the companion 
paper~\cite{sen:simultaneous}.

\section{Simplified one-shot inner bound for the classical quantum MAC}
\label{sec:cqMACBaby}
As a warmup, we show how tilting, and smoothing and augmentation 
can be used to prove a simplified version of our one-shot inner bound for
the multiple access channel with classical inputs and quantum output,
called cq-MAC henceforth (see Fig.~\ref{fig:QMACa}).
In the asymptotic iid setting, Winter~\cite{winter:cqmac} used a standard
{\em successive cancellation} argument to reduce the problem of finding
inner bounds for the cq-MAC to the problem of finding inner bounds
for the classical-quantum point-to-point channel. This allowed him
to prove an optimal inner bound in the asymptotic iid setting.
However in the one-shot setting, successive cancellation arguments give
only a finite set of achievable rate tuples. In order to get a continuous
rate region, we need to look for a {\em simultaneous} decoder
for the cq-MAC. Fawzi {\it et al\/}~\cite{fawzi:interference} and 
Sen~\cite{sen:interference} did construct a simultaneous decoder for
the two sender cq-MAC but their constructions, which were given in the 
asymptotic iid setting, are not known to work in the one-shot setting.
Qi, Wang and Wilde~\cite{qi:simultaneous} constructed a one-shot 
simultaneous
decoder for the cq-MAC with an arbitrary number of senders, but their
achievable rates restricted to the asymptotic iid setting are inferior
to the optimal rates obtained by Winter.
Thus, for more than two senders a 
simultaneous decoder for the cq-MAC achieving optimal rates was
hitherto unknown even in the asymptotic
iid setting. 

In this section, we construct a 
simultaneous decoder for the cq-MAC for two senders.
The simplified one shot inner bound given in this section does not
use a so-called `time sharing random variable'. The full version with
the time sharing random variable can be found in Section~\ref{sec:cqMAC}.
We obtain an inner bound which turns out to be the natural
one-shot analogue of Winter's inner bound, as well as the natural
quantum analogue of the one-shot classical inner bound described in
Section~\ref{sec:cMAC}. In other words, our one-shot inner bound
reduces to Winter's optimal inner bound in the asymptotic iid setting.
A one-shot outer bound nearly matching our 
inner bound was shown recently by Anshu, Jain and Warsi~\cite{anshu:qmac}.
That paper also gives another inner bound for the cq-MAC which, 
though nearly optimal in the one-shot setting, is not known to reduce
to the standard asymptotic iid inner bound of Winter unlike our 
one-shot inner bound.

There are two senders Alice and Bob who
would like to send classical messages 
$m_1 \in [2^{R_1}]$, $m_2 \in 2^{R_2}$ to 
a receiver Charlie. There
is a communication channel $\chan$ with two classical inputs and one 
quantum output called a cq-MAC
connecting Alice and Bob to Charlie. The 
two input alphabets of $\chan$ will be denoted
by $\cX$, $\cY$ and the output Hilbert space by $\cZ$. 
Let $0 \leq \epsilon \leq 1$.
On getting message $m_1$, Alice encodes it as a letter $x(m_1) \in \cX$
and feeds it to her channel input. 
Similarly on getting message $m_2$, Bob encodes it as a letter 
$y(m_2) \in \cY$ and feeds it to his channel input. The channel $\chan$
outputs a quantum state $\rho^Z_{x(m_1),y(m_2)}$ in $\cZ$. 
Charlie now has to try and guess 
the message pair $(m_1, m_2)$ from the channel output. 
We require that the probability of Charlie's decoding error 
averaged over the uniform distribution on
the set of message pairs $(m_1, m_2) \in [2^{R_1}] \times [2^{R_2}]$ 
is at most $\epsilon$.

Fix independent probability distributions $p(x)$, $p(y)$ 
on sets $\cX$, $\cY$.
Consider the classical-quantum state
\[
\rho^{XYZ} :=
\sum_{x,y} p(x) p(y) 
\ketbra{x,y}^{XY} \otimes
\rho^Z_{x,y}.
\]
This state `controls' the encoding and decoding performance for the
channel $\chan$. Define
\begin{eqnarray*}
\rho^Z_{x} 
& := &
\sum_y p(y) \rho^Z_{x,y}, \\
\rho^Z_{y} 
& := &
\sum_x p(x) \rho^Z_{x,y}, \\
\rho^Z
& := &
\sum_{x,y} p(x) p(y) \rho^Z_{x,y}.
\end{eqnarray*}
Then,
\begin{eqnarray*}
\rho^{XZ} 
& = &
\sum_x p(x) \ketbra{x} \otimes \rho^Z_{x}, \\
\rho^{YZ} 
& := &
\sum_y p(y) \ketbra{y} \otimes \rho^Z_{y}.
\end{eqnarray*}

Consider a new alphabet, as well as Hilbert space, $\cL$ and define
the {\em augmented} systems 
$\cX' := \cX \otimes \cL$,
$\cY' := \cY \otimes \cL$, 
and the extended system $\cZ'$ as follows: 
\[
\cZ' := 
(\cZ \otimes \C^2) \oplus
(\cZ \otimes \C^2 \otimes \cL^X) \oplus
(\cZ \otimes \C^2 \otimes \cL^Y),
\]
where the symbol $\oplus$ denotes orthogonal direct sum of spaces,
$\cL^X$, $\cL^Y$ are distinct orthogonal spaces isomorphic to $\cL$
labelled with mnemonic superscripts $X$, $Y$.
The role of the mnemonics
will become clear shortly when we described the tilted state. The
aim of {\em augmentation} is to ensure a so-called {\em smoothness} 
property of states obtained by averaging over $\cX'$ or $\cY'$ or both.

Consider a `variant' cq-MAC 
$\chan'$ with input alphabets $\cX'$, $\cY'$
and output Hilbert space $\cZ'$. On input $(x,l_x)$, 
$(y,l_y)$ to
$\chan'$, the output of $\chan'$ is the state
$
\rho^Z_{x,y} \otimes \ketbra{0}^{\C^2}.
$
The output of
$\chan'$ is taken to embed into the first summand in the definition of 
$\cZ'$ above.
The classical
quantum state `controlling' the encoding and decoding for $\chan'$ is
nothing but
$
\rho^{XYZ} \otimes 
\ketbra{0}^{\C^2} \otimes
\frac{\one^{\cL^{\otimes 2}}}{|\cL|^2}.
$
The channel $\chan'$ can be trivally obtained from channel $\chan$.
The expected average decoding error for $\chan'$ is the same as 
the expected average decoding error for $\chan$ for the same
rate pair $(R_1, R_2)$. In fact, an encoding-decoding scheme for
$\chan'$ immediately gives an encoding-decoding scheme for $\chan$
with the same rate pair $(R_1, R_2)$ and the same decoding error.

Let $0 \leq \delta \leq 1/10$.
Define $\hcZ := \cZ \otimes \C^2$. 
Let $l_x \in \cL^X$, $l_y \in \cL^Y$.
Consider the {\em tilting} map
$\cT_{X;l_x,\delta}: \hcZ \rightarrow \hcZ \oplus (\hcZ \otimes \cL^X)$ 
defined as
\[
\cT_{X;l_x,\delta}: \ket{h} \mapsto 
\frac{1}{\sqrt{1 + \delta^2}} 
(\ket{h} + \delta \ket{h}\ket{l_x}). 
\]
Define the tilting map $\cT_{Y:l_y,\delta}$ analogously. 
Define the tilting map
$\cT_{XY;l_x,l_y,\delta}: \hcZ \rightarrow \cZ'$ as
\[
\cT_{XY;l_x,l_y,\delta}: \ket{h} \mapsto 
\frac{1}{\sqrt{1 + 2 \delta^2}} 
(\ket{h} + \delta \ket{h}\ket{l_x} + \delta \ket{h}\ket{l_y}). 
\]
Next, consider a `perturbed' cq-MAC $\chan''$ with the same input alphabets
and output Hilbert space as in $\chan'$ (see Fig.~\ref{fig:QMACb}).
\begin{figure*}[!t]
\centering{
\includegraphics[width=.9\textwidth]{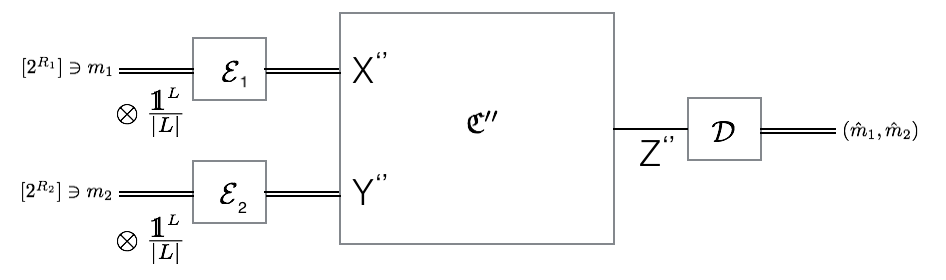}
}
\caption{The perturbation of the classical quantum multiple access 
channel $\chan$ of Fig.~\ref{fig:QMACa}. Alice, Bob 
encode their
respective messages $m_1$, $m_2$ into classical codewords, tensor
the codewords with completely mixed states over $L$, and then 
feed them to the perturbed channel $\chan''$.
Charlie applies a decoding superoperator to the quantum output of 
the channel to get his guess $(\hat{m}_1, \hat{m}_2)$ 
for the transmitted messages. $X''$, $Y''$
are the input alphabets and $Z''$ is the output Hilbert space of the 
channel. A rate pair that is achievable for the perturbed channel $\chan''$
is also achievable for the original channel $\chan$ with at most a
slight increase in error probability.
}
\label{fig:QMACb}
\end{figure*}
However, on input $(x,l_x)$, $(y,l_y)$ the output is the state
\[
(\rho')_{(x,l_x), (y, l_y), \delta}^{Z'}
 := 
\cT_{XY;l_x,l_y,\delta}(\rho^Z_{x,y} \otimes \ketbra{0}^{\C^2}),
\]
where the map $\cT_{XY;l_x,l_y,\delta}$ acts on a mixed
state by acting on each pure state in the mixture individually.
Consider the classical quantum state
\begin{eqnarray*}
\lefteqn{(\rho')^{X'Y'Z'}} \\
& := &
|\cL|^{-2}
\sum_{x, y, l_x, l_y} p(x) p(y)
\ketbra{x,l_x}^{X'} \otimes 
\ketbra{y,l_y}^{Y'} \\
&     &
~~~~~~~~~~~~~~~~~~~~~~~~~~~~~
{} \otimes 
(\rho')^{Z'}_{(x,l_x),(y,l_y),\delta}.
\end{eqnarray*}
This state `controls' the encoding and decoding performance for the
channel $\chan''$.

It is now easy to see that
\begin{equation}
\label{eq:cqmacperturbation}
\begin{array}{rcl}
\ellone{
(\rho')^{Z'}_{(x,l_x),(y,l_y),\delta} -
\rho^{Z}_{x,y} \otimes \ketbra{0}^{\C^2}
} 
& \leq &
4 \delta^2, \\
\ellone{
(\rho')^{X'Y'Z'} -
\rho^{XYZ} \otimes 
\ketbra{0}^{\C^2} \otimes
\frac{\one^{\cL^{\otimes 2}}}{|\cL|^2}
} 
& \leq &
4 \delta^2.
\end{array}
\end{equation}
Thus, the expected average decoding error for $\chan''$ is at most
the expected average decoding error for $\chan'$, which is also the same
as the expected average decoding error for $\chan$,  plus
$2 \delta^2$, for the same
rate pair $(R_1, R_2)$, and the same decoding strategy.

Consider the following randomised construction of a codebook $\cC$
for Alice and Bob for communication over the channel $\chan''$. 
Fix probability distributions $p(x)$, $p(y)$ on sets 
$\cX$, $\cY$. 
For all $m_1 \in [2^{R_1}]$, choose $(x,l_x)(m_1) \in \cX \times \cL$ 
independently according to the product of the distribution $p(x)$ on
$\cX$ and the uniform distribution on $\cL$. Similarly
for all $m_2 \in [2^{R_2}]$, choose $(y,l_y)(m_2) \in \cY \times \cL$ 
independently according to the product of the distribution $p(y)$ on
$\cY$ and the uniform distribution on $\cL$.

Let $\epsilon > 0$. 
Consider the optimising POVM element 
$(\Pi'')^{XYZ}_X$ in the Hilbert space $\cX \otimes \cY \otimes \cZ$
arising in the definition of $I^\epsilon_H(Y : XZ)$. 
Without loss of generality $\Pi^{XYZ}_X$
is of the form
\[
(\Pi'')^{XYZ}_X =
\sum_{x,y} \ketbra{x}^X \otimes \ketbra{y}^Y \otimes (\Pi'')^Z_{X;x,y}.
\]
Similarly we can define, for each $(x,y) \in \cX \times \cY$, POVM
elements $(\Pi'')^Z_{Y;x,y}$, $(\Pi'')^Z_{x,y}$ arising in the 
definitions of
$I^\epsilon_H(X : YZ)$, $I^\epsilon_H(XY : Z)$. 
Then
\begin{equation}
\label{eq:cqmacoptimising}
\begin{array}{rcl}
\Tr [(\Pi'')^{XYZ}_{X} \rho^{XYZ}] 
& \geq & 
1 - \epsilon, \\
\Tr [(\Pi'')^{XYZ}_{X} (\rho^{Y} \otimes \rho^{XZ})] 
& \leq & 
2^{-I^\epsilon_H(X:YZ)}, \\
\Tr [(\Pi'')^{XYZ}_{Y} \rho^{XYZ}] 
& \geq & 
1 - \epsilon, \\
\Tr [(\Pi'')^{XYZ}_{Y} (\rho^{X} \otimes \rho^{YZ})] 
& \leq & 
2^{-I^\epsilon_H(Y:XZ)}, \\
\Tr [(\Pi'')^{XYZ} \rho^{XYZ}] 
& \geq & 
1 - \epsilon, \\
\Tr [(\Pi'')^{XYZ} (\rho^{XY} \otimes \rho^Z)] 
& \leq & 
2^{-I^\epsilon_H(XY:Z)}.
\end{array}
\end{equation}

By Fact~\ref{fact:gelfandnaimark}, there are orthogonal projections 
$\Pi^{\hZ}_{X;x,y}$, $\Pi^{\hZ}_{Y;x,y}$, $\Pi^{\hZ}_{x,y}$ in $\cZ'$
that give the same measurement probability on states 
$\sigma^Z \otimes \ketbra{0}^{\C^2}$ that POVM elements
$(\Pi'')^{Z}_{X;x,y}$, $(\Pi'')^{Z}_{Y;x,y}$, $(\Pi'')^{Z}_{x,y}$ 
give on
states $\sigma^Z$. Let $W_{X;x,y}$ denote the orthogonal complement of
the support of $\Pi^{\hZ}_{X;x,y}$ in $\hcZ$. Define $W_{Y;x,y}$, 
$W_{x,y}$ analogously. 
Define the tilted spaces 
\[
W'_{X;(x,l_x),(y,l_y),\delta} := 
\cT_{X;l_x,\delta}(W_{X;x,y}),
~
W'_{Y;(x,l_x),(y,l_y),\delta} := 
\cT_{Y;l_y,\delta}(W_{Y;x,y}),
\]
residing in the Hilbert space $\cZ'$. Define the subspace
\[
W'_{(x,l_x),(y,l_y),\delta} :=
W'_{X;(x,l_x),(y,l_y),\delta} + 
W'_{Y;(x,l_x),(y,l_y),\delta} + 
W_{x,y},
\]
and
$
(\Pi')^{Z'}_{W'_{(x,l_x),(y,l_y),\delta}}
$
to be the orthogonal projection in $\cZ'$ onto 
$W'_{(x,l_x),(y,l_y),\delta}$. Let $\Pi^{Z'}_{\hcZ}$ 
be the orthogonal projection in $\cZ'$ onto $\hcZ$. Then define 
the POVM element $(\Pi')^{Z'}_{(x,l_x),(y,l_y),\delta}$ in $\cZ'$ to be
\[
(\Pi')^{Z'}_{(x,l_x),(y,l_y),\delta} :=
(\one^{\cZ'} - (\Pi')^{Z'}_{W'_{(x,l_x),(y,l_y),\delta}})
\Pi^{Z'}_{\hcZ}
(\one^{\cZ'} - (\Pi')^{Z'}_{W'_{(x,l_x),(y,l_y),\delta}}).
\]

Define the POVM element
\begin{eqnarray*}
\lefteqn{(\Pi')^{X' Y' Z'}} \\
& := &
\sum_{x, y, l_x, l_y}
\ketbra{x, l_x}^{X'} \otimes
\ketbra{y, l_y}^{Y'} \otimes
(\Pi')_{(x,l_x), (y,l_y), \delta}^{Z'}.
\end{eqnarray*}
By Equations~\ref{eq:cqmacperturbation}, \ref{eq:cqmacoptimising}, 
Fact~\ref{fact:noncommutativeunionbound} and 
Corollary~\ref{cor:tiltedspan}, we get
\begin{eqnarray*}
\lefteqn{
\Tr [(\one^{X'Y'Z'} - (\Pi')^{X' Y' Z'}) (\rho')^{X'Y'Z'}]
} \\
& = &
\Tr (\rho')^{X'Y'Z'} -
\Tr [(\Pi')^{X' Y' Z'}) (\rho')^{X'Y'Z'}] \\
& = &
|\cL|^{-2} \sum_{x,y,l_x,l_y} p(x) p(y) 
\left(
\Tr (\rho')^{Z'}_{(x,l_x),(y,l_y),\delta} 
\right. \\
&   &
~~~~~~~~~~~
\left.
{} -
\Tr [
(\Pi')^{Z'})_{(x,l_x),(y,l_y),\delta} 
(\rho')^{Z'}_{(x,l_x),(y,l_y),\delta}
]
\right) \\
& = &
|\cL|^{-2} \sum_{x,y,l_x,l_y} p(x) p(y) 
\left(
\Tr (\rho')^{Z'}_{(x,l_x),(y,l_y),\delta} 
\right. \\
&   &
~~~~~~~~~~~
{} -
\Tr [
\Pi^{Z'}_{\hcZ}
(\one^{\cZ'} - (\Pi')^{Z'}_{W'_{(x,l_x),(y,l_y),\delta}}) \\
&   &
~~~~~~~~~~~~~~~~~~~~~~
(\rho')^{Z'}_{(x,l_x),(y,l_y),\delta} \\
&   &
~~~~~~~~~~~~~~~~~~~~~~
\left.
(\one^{\cZ'} - (\Pi')^{Z'}_{W'_{(x,l_x),(y,l_y),\delta}})
\Pi^{Z'}_{\hcZ}
]
\right) \\
& \leq &
|\cL|^{-2} \sum_{x,y,l_x,l_y} p(x) p(y) 
\left(
\Tr [
(\one^{\cZ'} - \Pi^{Z'}_{\hcZ})
(\rho')^{Z'}_{(x,l_x),(y,l_y),\delta}
] 
\right. \\
&      &
~~~~~~~~~~~~~~~~~~~~~~~~
\left.
{} +
\Tr [
(\Pi')^{Z'}_{W'_{(x,l_x),(y,l_y),\delta}} 
(\rho')^{Z'}_{(x,l_x),(y,l_y),\delta} 
]
\right) \\
& \leq &
4 \delta^2 \\
&      &
{} +
|\cL|^{-2} \sum_{x,y,l_x,l_y} p(x) p(y) 
\left(
\Tr [
(\one^{\cZ'} - \Pi^{Z'}_{\hcZ})
(\rho^{Z}_{x,y} \otimes \ketbra{0}^{\C^2})
] 
\right. \\
&      &
~~~~~~~~~~~~~~~~~~~~~~~~
\left.
{} +
\Tr [
(\Pi')^{Z'}_{W'_{(x,l_x),(y,l_y),\delta}} 
(\rho^{Z}_{x,y} \otimes \ketbra{0}^{\C^2})
]
\right) \\
& \leq &
\frac{6}{\delta^2}
|\cL|^{-2} \sum_{x,y,l_x,l_y} p(x) p(y) 
\left( 
\right. \\
&   &
~~~~~~~~~~~~~~~~~~~~~~~~~~~~~~
(1 -
\Tr [
\Pi^{\hZ}_{X;x,y} 
(\rho^{Z}_{x,y} \otimes \ketbra{0}^{\C^2})
]
) \\
&   &
~~~~~~~~~~~~~~~~~~~~~~~~~~~~~~
{} +
(1 -
\Tr [
\Pi^{\hZ}_{Y;x,y} 
(\rho^{Z}_{x,y} \otimes \ketbra{0}^{\C^2})
]
) \\
&   &
~~~~~~~~~~~~~~~~~~~~~~~~~~~~~~
\left.
{} +
(1 -
\Tr [
\Pi^{\hZ}_{x,y} 
(\rho^{Z}_{x,y} \otimes \ketbra{0}^{\C^2})
]
)
\right) \\
&   &
{} +
4 \delta^2 \\
& \leq &
\frac{6}{\delta^2} 
\left(
(1 - \Tr [(\Pi'')^{XYZ}_{X} \rho^{XYZ}]) +
(1 - \Tr [(\Pi'')^{XYZ}_{Y} \rho^{XYZ}]) 
\right. \\
&   &
~~~~~~~
\left.
{} +
(1 - \Tr [(\Pi'')^{XYZ} \rho^{XYZ}])
\right) +
4 \delta^2 \\
& \leq &
\frac{18 \epsilon}{\delta^2} +
4 \delta^2.
\end{eqnarray*}
We choose $\delta := \epsilon^{1/4}$ to get
\begin{equation}
\label{eq:cqmactype1error}
\Tr [(\one^{X'Y'Z'} - (\Pi')^{X' Y' Z'}) (\rho')^{X'Y'Z'}] \leq
22 \epsilon^{1/2}.
\end{equation}

We now illustrate how the {\em augmentation} of $\cX$ to $\cX'$ and
$\cY$ to $\cY'$ helps in the so-called  {\em smoothing} of the 
states
\begin{eqnarray*}
(\rho')^{Z'}_{(x,l_x),\delta} 
& := &
|\cL|^{-1} \sum_{y,l_y} p(y)
(\rho')^{Z'}_{(x,l_x),(y,l_y),\delta}, \\
(\rho')^{Z'}_{(y,l_y),\delta} 
& := &
|\cL|^{-1} \sum_{x,l_x} p(x)
(\rho')^{Z'}_{(x,l_x),(y,l_y),\delta}, \\
(\rho')^{Z'}_{\delta} 
& := &
|\cL|^{-2} \sum_{x,l_x,y,l_y} p(x) p(y)
(\rho')^{Z'}_{(x,l_x),(y,l_y),\delta}.
\end{eqnarray*}
As will become clear below, the tilting map $\cT_{XY;l_x,l_y,\delta}$ is
defined in such a way that 
$(\rho')^{Z'}_{(x,l_x),\delta}$ is very
close to $\cT_{X;l_x,\delta}(\rho^Z_{x} \otimes \ketbra{0}^{\C^2})$ 
in the $\ell_\infty$-norm,
$(\rho')^{Z'}_{(y,l_y),\delta}$ is very
close to $\cT_{Y;l_y,\delta}(\rho^Z_{y} \otimes \ketbra{0}^{\C^2})$ 
in the $\ell_\infty$-norm, and
$(\rho')^{Z'}_{\delta}$ is very
close to $\rho^Z \otimes \ketbra{0}^{\C^2}$ 
in the $\ell_\infty$-norm. This closeness is what we mean by 
{\em smoothing}, and {\em augmentation} is required to ensure proper
smoothing.
Notice now that
\begin{eqnarray*}
\lefteqn{
(\rho')^{Y'} \otimes (\rho')^{X'Z'} 
} \\
&  =  &
\left(|\cL|^{-1} \sum_{y,l_y} p(y) \ketbra{y,l_y}\right) \\
&    &
{} 
\otimes
\left(
|\cL|^{-1} \sum_{x,l_x} p(x) \ketbra{x,l_x} \otimes 
 (\rho')^{Z'}_{(x,l_x),\delta}
\right), \\
\lefteqn{
(\rho')^{X'} \otimes (\rho')^{Y'Z'} 
} \\
&  =  &
\left(|\cL|^{-1} \sum_{x,l_x} p(x) \ketbra{x,l_x}\right) \\
&    &
{} 
\otimes
\left(
|\cL|^{-1} \sum_{y,l_y} p(y) \ketbra{y,l_y} \otimes 
 (\rho')^{Z'}_{(y,l_y),\delta}
\right), \\
\lefteqn{
(\rho')^{X'Y'} \otimes (\rho')^{Y'Z'} 
} \\
&  =  &
\left(|\cL|^{-2} \sum_{x,l_x,y,l_y} p(x) p(y) \ketbra{x,l_x,y,l_y}\right) 
\otimes
(\rho')^{Z'}_{\delta}.
\end{eqnarray*}

Observe that
\begin{eqnarray*}
\lefteqn{
|\cL|^{-1} \sum_{l_y} \cT_{XY;l_x,l_y}(\ketbra{h})
} \\
& = &
\frac{|\cL|^{-1}}{1 + 2\delta^2} 
\sum_{l_y} 
\left(
\right. \\
&   &
~~~~
(\ket{h} + \delta \ket{h} \ket{l_x})
(\bra{h} + \delta \bra{h} \bra{l_x}) \\
&   &
~~~~
{} +
\delta (\ket{h} + \delta \ket{h} \ket{l_x}) \bra{h}\bra{l_y} \\
&   &
~~~~
{} +
\delta \ket{h}\ket{l_y} (\bra{h} + \delta \bra{h} \bra{l_x}) \\
&   &
~~~~
\left.
{} +
\delta^2 \ketbra{h}\ketbra{l_y} 
\right) \\
& = &
\frac{1+\delta^2}{1 + 2\delta^2} \cT_{X;l_x,\delta}(\ketbra{h}) +
N_{X;l_x,\delta}(\ketbra{h}),
\end{eqnarray*}
where
\begin{eqnarray*}
\lefteqn{N_{X;l_x,\delta}(\ketbra{h})} \\
& := &
\frac{|\cL|^{-1}}{1 + 2\delta^2} 
\sum_{l_y} 
\left(
\delta (\ket{h} + \delta \ket{h} \ket{l_x}) \bra{h}\bra{l_y} 
\right. \\
&   &
~~~~~~~~~~~~~~~~~~~~
{} +
\delta \ket{h}\ket{l_y} (\bra{h} + \delta \bra{h} \bra{l_x}) \\
&   &
~~~~~~~~~~~~~~~~~~~~
\left.
{} +
\delta^2 \ketbra{h}\ketbra{l_y} 
\right).
\end{eqnarray*}
The vectors $\ket{l_y}$ are orthogonal as $l_y$ runs through the
computational basis vectors of $\cL^Y$. From this, it is easy to see that
\begin{eqnarray*}
\lefteqn{\ellinfty{N_{X;l_x,\delta}(\ketbra{h})}} \\
& \leq &
\frac{|\cL|^{-1}}{1 + 2\delta^2} 
\left(
\ellinfty{
\sum_{l_y} 
\delta (\ket{h} + \delta \ket{h} \ket{l_x}) \bra{h}\bra{l_y} 
} 
\right. \\
&    &
~~~~~~~~~~~~~~~~~~
{} +
\ellinfty{
\sum_{l_y} 
\delta \ket{h}\ket{l_y} (\bra{h} + \delta \bra{h} \bra{l_x}) 
} \\
&    &
~~~~~~~~~~~~~~~~~~
\left.
{} +
\ellinfty{
\sum_{l_y} 
\delta^2 \ketbra{h}\ketbra{l_y} 
}
\right) \\
&   =  &
\frac{|\cL|^{-1}}{1 + 2\delta^2} 
\left(
\delta \elltwo{\ket{h} + \delta \ket{h} \ket{l_x}} \cdot
\elltwo{\sum_{l_y} \ket{l_y}} 
\right. \\
&    &
~~~~~~~~~~~~~~~~~~
\left.
{} +
\delta 
\elltwo{\ket{h} + \delta \ket{h} \ket{l_x}} \cdot
\elltwo{\sum_{l_y} \ket{l_y}} +
\delta^2 
\right) \\
& \leq &
\frac{3 \delta}{\sqrt{|\cL|}}.
\end{eqnarray*}
Similarly, one can define
\begin{eqnarray*}
N_{Y;l_y,\delta}(\ketbra{h}) 
& := &
\frac{|\cL|^{-1}}{1 + 2\delta^2} 
\sum_{l_x} 
\left(
\delta (\ket{h} + \delta \ket{h} \ket{l_y}) \bra{h}\bra{l_x} 
\right. \\
&   &
~~~~~~~~~~~~~~~~~~~~
{} +
\delta \ket{h}\ket{l_x} (\bra{h} + \delta \bra{h} \bra{l_y}) \\
&   &
~~~~~~~~~~~~~~~~~~~~
\left.
{} +
\delta^2 \ketbra{h}\ketbra{l_x} 
\right),
\end{eqnarray*}
\begin{eqnarray*}
N_{\delta}(\ketbra{h}) 
& := &
\frac{|\cL|^{-2}}{1 + 2\delta^2} 
\sum_{l_x,l_y} 
\left(
\delta \ket{h} (\bra{h}\bra{l_x} + \bra{h}\bra{l_y}) 
\right. \\
&   &
~~~~~~~~~~~~~~~~~~~~
{} +
\delta (\ket{h}\ket{l_x} + \ket{h}\ket{l_y}) \bra{h}  \\
&   &
~~~~~~~~~~~~~~~~~~~~
\left.
{} +
\delta^2 \ketbra{h}(\ket{l_x} + \ket{l_y}) (\bra{l_x} + \bra{l_y})
\right),
\end{eqnarray*}
and show that their $\ell_\infty$-norms are 
upper bounded by $\frac{3 \delta}{\sqrt{|\cL|}}$ each. 
We can also extend the maps $N_{X;l_x,\delta}$,
$N_{Y;l_y,\delta}$, $N_{\delta}$ to mixed states in the natural manner.

We can thus write
\begin{eqnarray*}
(\rho')^{Z'}_{(x,l_x),\delta} 
& = &
\frac{1+\delta^2}{1 + 2\delta^2} 
\cT_{X;l_x,\delta}(\rho^Z_x \otimes \ketbra{0}^{\C^2}) \\
&   &
{} +
N_{X;l_x,\delta}(\rho^Z_x \otimes \ketbra{0}^{\C^2}), \\
(\rho')^{Z'}_{(y,l_y),\delta} 
& = &
\frac{1+\delta^2}{1 + 2\delta^2} 
\cT_{Y;l_y,\delta}(\rho^Z_y \otimes \ketbra{0}^{\C^2}) \\
&   &
{} +
N_{Y;l_y,\delta}(\rho^Z_y \otimes \ketbra{0}^{\C^2}), \\
(\rho')^{Z'}_{\delta} 
& = &
\frac{1}{1 + 2\delta^2} (\rho^Z \otimes \ketbra{0}^{\C^2}) \\
&   &
{} +
N_{\delta}(\rho^Z \otimes \ketbra{0}^{\C^2}), \\
\end{eqnarray*}
with
\begin{eqnarray*}
\ellinfty{N_{X;l_x,\delta}(\rho^Z_x \otimes \ketbra{0}^{\C^2})}
& \leq & 
\frac{3 \delta}{\sqrt{|\cL|}}, \\
\ellinfty{N_{Y;l_y,\delta}(\rho^Z_y \otimes \ketbra{0}^{\C^2})} 
& \leq & 
\frac{3 \delta}{\sqrt{|\cL|}}, \\
\ellinfty{N_{\delta}(\rho^Z \otimes \ketbra{0}^{\C^2})} 
& \leq & 
\frac{3 \delta}{\sqrt{|\cL|}}. 
\end{eqnarray*}
Observe now that
\begin{eqnarray*}
\ellone{(\Pi')^{Z'}_{(x,l_x),(y,l_y),\delta}} 
& \leq &
\ellinfty{
(\one^{\cZ'} - (\Pi')^{Z'}_{W'_{(x,l_x),(y,l_y),\delta}})
}^2
\ellone{\Pi^{Z'}_{\hcZ}}  \\
&  =   &
2 |\cZ|.
\end{eqnarray*}
Hence, we get the upper bound
\begin{eqnarray*}
\lefteqn{
\Tr [
(\Pi')^{Z'}_{(x,l_x),(y,l_y),\delta}  
(\rho')^{Z'}_{(x,l_x),\delta} 
]
} \\
& \leq &
\Tr [
(\Pi')^{Z'}_{(x,l_x),(y,l_y),\delta}  
(\cT_{X;l_x,\delta}(\rho^Z_x \otimes \ketbra{0}^{\C^2}))
] \\
&      &
{} +
\ellone{(\Pi')^{Z'}_{(x,l_x),(y,l_y),\delta}} \cdot
\ellinfty{N_{X;l_x,\delta}(\rho^Z_x \otimes \ketbra{0}^{\C^2})} \\
& \leq &
\Tr [
(\one^{\cZ'} - (\Pi')^{Z'}_{W'_{(x,l_x),(y,l_y),\delta}})
(\cT_{X;l_x,\delta}(\rho^Z_x \otimes \ketbra{0}^{\C^2}))
] \\
&      &
{} +
\frac{6 \delta |\cZ|}{\sqrt{|\cL|}} \\ 
& \leq &
\Tr [
(\one^{\cZ'} - (\Pi')^{Z'}_{W'_{X;(x,l_x),(y,l_y),\delta}})
(\cT_{X;l_x,\delta}(\rho^Z_x \otimes \ketbra{0}^{\C^2}))
] \\
&      &
{} +
\frac{6 \delta |\cZ|}{\sqrt{|\cL|}} \\ 
&   =  &
\Tr [
(\one^{\cZ'} - (\Pi')^{Z'}_{W'_{X;(x,l_x),(y,l_y),\delta}})
\Pi^{Z'}_{\cT_{X;l_x,\delta}(\hcZ)}
(\cT_{X;l_x,\delta}(\rho^Z_x \otimes \ketbra{0}^{\C^2}))
] \\
&      &
{} +
\frac{6 \delta |\cZ|}{\sqrt{|\cL|}} \\ 
&   =  &
\Tr [
(
\one^{\cT_{X;l_x,\delta}(\hcZ)} - 
(\Pi')^{\cT_{X;l_x,\delta}(\hcZ)}_{W'_{X;(x,l_x),(y,l_y),\delta}}
)
(\cT_{X;l_x,\delta}(\rho^Z_x \otimes \ketbra{0}^{\C^2}))
] \\
&      &
{} +
\frac{6 \delta |\cZ|}{\sqrt{|\cL|}} \\ 
&   =  &
\Tr [
(
\one^{\hcZ} - 
\Pi^{\hcZ}_{W_{X;x,y}}
)
(\rho^Z_x \otimes \ketbra{0}^{\C^2})
] +
\frac{6 \delta |\cZ|}{\sqrt{|\cL|}} \\ 
&   =  &
\Tr [
\Pi^{\hcZ}_{X;x,y}
(\rho^Z_x \otimes \ketbra{0}^{\C^2})
] +
\frac{6 \delta |\cZ|}{\sqrt{|\cL|}} \\ 
&   =  &
\Tr [
(\Pi'')^{Z}_{X;x,y}
\rho^Z_x 
] +
\frac{6 \delta |\cZ|}{\sqrt{|\cL|}}.
\end{eqnarray*}
Above, we used the facts that 
$W'_{X;(x,l_x),(y,l_y),\delta} \leq W'_{(x,l_x),(y,l_y),\delta}$, 
$W'_{X;(x,l_x),(y,l_y),\delta} \leq \cT_{X;l_x,\delta}(\hcZ)$,
and that $\cT_{X;l_x,\delta}$ is an isometry.
Using Equation~\ref{eq:cqmacoptimising}, we now get
\begin{eqnarray*}
\lefteqn{
\Tr [
(\Pi')^{X'Y'Z'}  
((\rho')^{Y'} \otimes (\rho')^{X'Z'})
]
} \\
&  =  &
|\cL|^{-2} \sum_{y,l_y} p(y) \sum_{x,l_x} p(x)  
\Tr [
(\Pi')^{Z'}_{(x,l_x),(y,l_y),\delta}  
(\rho')^{Z'}_{(x,l_x),\delta} 
] \\
& \leq &
|\cL|^{-2} \sum_{y,l_y} p(y) \sum_{x,l_x} p(x)  
\Tr [
(\Pi'')^{Z}_{X;x,y}
\rho^Z_x 
] +
\frac{6 \delta |\cZ|}{\sqrt{|\cL|}} \\
&   =  &
\sum_{x,y} p(x) p(y)
\Tr [
(\Pi'')^{Z}_{X;x,y}
\rho^Z_x 
] +
\frac{6 \delta |\cZ|}{\sqrt{|\cL|}} \\
&   =  &
\Tr [
(\Pi'')^{XYZ}_{X}
(\rho^Y \otimes \rho^{XZ})
] +
\frac{6 \delta |\cZ|}{\sqrt{|\cL|}} \\
& \leq &
2^{-I^\epsilon_H(Y:XZ)} +
\frac{6 \delta |\cZ|}{\sqrt{|\cL|}}.
\end{eqnarray*}
We choose $|\cL|$ large enough so that the second term in the
last inequality is less than 
$
\min\{
2^{-I^\epsilon_H(Y:XZ)}, 
2^{-I^\epsilon_H(X:YZ)}, 
2^{-I^\epsilon_H(XY:Z)} 
\}.
$
Thus, we have the inequalities
\begin{equation}
\label{eq:cqmactype2error}
\begin{array}{rcl}
\Tr [
(\Pi')^{X'Y'Z'}  
((\rho')^{Y'} \otimes (\rho')^{X'Z'})
]
& \leq &
2^{-I^\epsilon_H(Y:XZ) + 1}, \\ 
\Tr [
(\Pi')^{X'Y'Z'}  
((\rho')^{X'} \otimes (\rho')^{Y'Z'})
]
& \leq &
2^{-I^\epsilon_H(X:YZ) + 1}, \\ 
\Tr [
(\Pi')^{X'Y'Z'}  
((\rho')^{X'Y'} \otimes (\rho')^{Z'})
]
& \leq &
2^{-I^\epsilon_H(XY:Z) + 1}.
\end{array}
\end{equation}

We now describe the decoding strategy that Charlie follows in order to
try and guess the message pair $(m_1, m_2)$ that was actually sent,
given the output of $\chan''$. 
Charlie uses the {\em pretty good measurement} 
\cite{belavkin:pgm1, belavkin:pgm2, holevo:capacity, schumacher:capacity}
constructed from the POVM elements 
$
(\Pi')_{(x,l_x)(m_1), (y,l_y)(m_2), \delta}^{Z'},
$
where $(m_1,m_2) \times [2^{R_1}] \times [2^{R_2}]$.
We now analyse the expectation, under the choice of a random codebook
$\cC$, of the error probability of Charlie's decoding algorithm.
Suppose the message pair $(m_1, m_2)$ is inputted to $\chan''$.
The output of $\chan''$ is the state
$(\rho')_{(x,l_x)(m_1), (y, l_y)(m_2), \delta}^{Z'}$.
Let $\Lambda_{\hat{m}_1, \hat{m}_2}^{Z'}$ be the POVM element corresponding
to decoded output $(\hat{m}_1, \hat{m}_2)$ arising from the
pretty good measurement. By the Hayashi-Nagaoka 
inequality~\cite{HayashiNagaoka},
the decoding error for $(m_1, m_2)$ is upper bounded by
\begin{eqnarray*}
\lefteqn{
\Tr [
(\one^{Z'} - \Lambda_{m_1, m_2}^{Z'})
(\rho')_{(x,l_x)(m_1), (y, l_y)(m_2), \delta}^{Z'} 
]
} \\
& \leq &
2 \Tr [
(\one^{Z'} - (\Pi')_{(x,l_x)(m_1), (y, l_y)(m_2), \delta}^{Z'}) \\
&      &
~~~~~~~~~~~~~
(\rho')_{(x,l_x)(m_1), (y, l_y)(m_2), \delta}^{Z'} 
] \\
&      &
{} +
4 \sum_{(\hat{m}_1, \hat{m}_2) \neq (m_1, m_2)} 
\Tr [
(\Pi')_{(x,l_x)(\hat{m}_1), (y, l_y)(\hat{m}_2), \delta}^{Z'} \\
&      &
~~~~~~~~~~~~~~~~~~~~~~~~~~~~~~~~~~~~~~~
(\rho')_{(x,l_x)(m_1), (y, l_y)(m_2), \delta}^{Z'} 
] \\
&   =  &
2 \Tr [
(\one^{Z'} - (\Pi')_{(x,l_x)(m_1), (y, l_y)(m_2), \delta}^{Z'}) \\
&      &
~~~~~~~~~~~~~
(\rho')_{(x,l_x)(m_1), (y, l_y)(m_2), \delta}^{Z'} 
] \\
&      &
{} +
4 \sum_{\hat{m}_1 \neq m_1} 
\Tr [
(\Pi')_{(x,l_x)(\hat{m}_1), (y, l_y)(m_2), \delta}^{Z'} \\
&      &
~~~~~~~~~~~~~~~~~~~~~~~~~~~
(\rho')_{(x,l_x)(m_1), (y, l_y)(m_2), \delta}^{Z'} 
] \\
&      &
{} +
4 \sum_{\hat{m}_2 \neq m_2} 
\Tr [
(\Pi')_{(x,l_x)(m_1), (y, l_y)(\hat{m}_2), \delta}^{Z'} \\
&      &
~~~~~~~~~~~~~~~~~~~~~~~~~~~
(\rho')_{(x,l_x)(m_1), (y, l_y)(m_2), \delta}^{Z'} 
] \\
&      &
{} +
4 \sum_{\hat{m}_1 \neq m_1, \hat{m}_2 \neq m_2} 
\Tr [
(\Pi')_{(x,l_x)(\hat{m}_1), (y, l_y)(\hat{m}_2), \delta}^{Z'} \\
&      &
~~~~~~~~~~~~~~~~~~~~~~~~~~~~~~~~~~~~
(\rho')_{(x,l_x)(m_1), (y, l_y)(m_2), \delta}^{Z'} 
].
\end{eqnarray*}
The expectation, over the choice of the random codebook $\cC$, of the
decoding error for $(m_1, m_2)$ is upper bounded by
\begin{eqnarray*}
\lefteqn{
\E_{\cC} [
\Tr [
(\one^{Z'} - \Lambda_{m_1, m_2}^{Z'})
(\rho')_{(x,l_x)(m_1), (y, l_y)(m_2), \delta}^{Z'} 
]
]
} \\
& \leq &
2 
\E_{\cC} [
\Tr [
(\one^{Z'} - (\Pi')_{(x,l_x)(m_1), (y, l_y)(m_2), \delta}^{Z'}) \\
&     &
~~~~~~~~~~~~~~~~~~~~
(\rho')_{(x,l_x)(m_1), (y, l_y)(m_2), \delta}^{Z'} 
]
] \\
&      &
{} +
4 \sum_{\hat{m}_1 \neq m_1} 
\E_{\cC} [
\Tr [
(\Pi')_{(x,l_x)(\hat{m}_1), (y, l_y)(m_2), \delta}^{Z'} \\
&     &
~~~~~~~~~~~~~~~~~~~~~~~~~~~~~~~~
(\rho')_{(x,l_x)(m_1), (y, l_y)(m_2), \delta}^{Z'} 
]
] \\
&      &
{} +
4 \sum_{\hat{m}_2 \neq m_2} 
\E_{\cC} [
\Tr [
(\Pi')_{(x,l_x)(m_1), (y, l_y)(\hat{m}_2), \delta}^{Z'} \\
&     &
~~~~~~~~~~~~~~~~~~~~~~~~~~~~~~~~
(\rho')_{(x,l_x)(m_1), (y, l_y)(m_2), \delta}^{Z'} 
]
] \\
&      &
{} +
4 \sum_{\hat{m}_1 \neq m_1, \hat{m}_2 \neq m_2} 
\E_{\cC} [
\Tr [
(\Pi')_{(x,l_x)(\hat{m}_1), (y, l_y)(\hat{m}_2), \delta}^{Z'} \\
&     &
~~~~~~~~~~~~~~~~~~~~~~~~~~~~~~~~~~~~~~~~~~
(\rho')_{(x,l_x)(m_1), (y, l_y)(m_2), \delta}^{Z'} 
]
] \\
&   =  &
2 |\cL|^{-3} \sum_{x, y, l_x, l_y} p(x) p(y) \\
&      &
~~~~~~~~~~~~~~~~~~~~~~~~~~~~
 \Tr [
(\one^{Z'} - (\Pi')_{(x,l_x), (y, l_y), \delta}^{Z'}) \\
&       &
~~~~~~~~~~~~~~~~~~~~~~~~~~~~~~~~~~~~~~~~
(\rho')_{(x,l_x), (y, l_y), \delta}^{Z'} 
] \\
&      &
{} +
4 (2^{R_1} - 1)
|\cL|^{-4} \sum_{x, l_x, x', l'_x, y, l_y}
p(x) p(x') p(y) \\
&     &
~~~~~~~~~~
\Tr [
(\Pi')_{(x',l'_x), (y, l_y), \delta}^{Z'})
(\rho')_{(x,l_x), (y, l_y), \delta}^{Z'} 
] \\
&  &
{} +
4 (2^{R_2} - 1)
|\cL|^{-4} \sum_{x, l_x, x', l'_x, y, l_y}
p(x) p(y) p(y') \\
&     &
~~~~~~~~~~
\Tr [
(\Pi')_{(x,l_x), (y', l'_y), \delta}^{Z'})
(\rho')_{(x,l_x), (y, l_y), \delta}^{Z'} 
] \\
&      &
{} +
4 (2^{R_1} - 1) (2^{R_2} - 1)
|\cL|^{-5} \\
&     &
~~~~~~
\sum_{x, l_x, x', l'_x, y, l_y, y', l'_y}
p(x) p(x') p(y) p(y') \\
&     &
~~~~~~~~~~~~~~~~~~~~~~~
\Tr [
(\Pi')_{(x',l'_x), (y', l'_y), \delta}^{Z'})
(\rho')_{(x,l_x), (y, l_y), \delta}^{Z'} 
] \\
&   =  &
2 \Tr [
(\one^{X' Y' Z'} - (\Pi')^{X' Y' Z'})
(\rho')^{X' Y' Z'}
] \\
&      &
{} +
4 (2^{R_1} - 1)
\Tr [
(\Pi')^{X' Y' Z'}
((\rho')^{X'} \otimes (\rho')^{Y'Z'})
] \\
&      &
{} +
4 (2^{R_2} - 1)
\Tr [
(\Pi')^{X' Y' Z'}
((\rho')^{Y'} \otimes (\rho')^{X'Z'})
] \\
&      &
{} +
4 (2^{R_1} - 1) (2^{R_2} - 1)
\Tr [
(\Pi')^{X' Y' Z'}
((\rho')^{X'Y'} \otimes (\rho')^{Z'})
] \\
& \leq &
44 \epsilon^{1/2} +
2^{R_1 + 1 - I^\epsilon_H(X : Y Z)_\rho} +
2^{R_2 + 1 - I^\epsilon_H(Y : X Z)_\rho} \\
&      &
{} +
2^{R_1 + R_2 + 1 -  I^\epsilon_H(XY : Z)_\rho},
\end{eqnarray*}
where we used Equations~\ref{eq:cqmactype1error}, \ref{eq:cqmactype2error}
in the last inequality above.

Choosing a rate pair $(R_1, R_2)$ satisfying
\begin{eqnarray*}
R_1 
& \leq & 
I^\epsilon_H(X : Y Z)_\rho - 1 - \log \frac{1}{\epsilon}, \\
R_2 
& \leq & 
I^\epsilon_H(Y : X Z)_\rho - 1 - \log \frac{1}{\epsilon}, \\
R_1 + R_2 
& \leq & 
I^\epsilon_H(XY : Z)_\rho - 1 - \log \frac{1}{\epsilon}, 
\end{eqnarray*} 
ensures that the expected average decoding error for channel
$\chan''$ is at most $47 \epsilon^{1/2}$. This implies that
the expected average decoding error for the original channel $\chan$
is at most $49 \epsilon^{1/2}$. Thus there exists a codebook
$\cC$ with average decoding error for $\chan$ at most 
$49 \epsilon^{1/2}$. By a standard technique of taking maps from
classical symbols to arbitrary quantum states, we can then prove the 
following theorem.

\noindent
{\bf Theorem~\ref{thm:cqMAC}' \;}
{\em
Let $\chan: X' Y' \rightarrow Z$ be a quantum multiple access channel.
Let $\cX$, $\cY$ be two new sample spaces.
For every element $x \in \cX$, let 
$\sigma_{x}^{X'}$ be a quantum state in the input Hilbert space 
$X'$ of $\chan$.
Similarly, for every element $y \in \cY$, let 
$\sigma_{y}^{Y'}$ be a quantum state in the input Hilbert space 
$Y'$ of $\chan$.
Let $p(x) p(y)$ be a probability distribution on 
$\cX \times \cY$.
Consider the classical quantum state
\[
\rho^{X Y Z} := 
\sum_{x, y}
p(x) p(y)
\ketbra{x, y}^{X Y} \otimes
(\chan(\sigma_{x}^{X'} \otimes \sigma_{y}^{Y'}))^{Z}.
\]
Let $R_1$, $R_2$, $\epsilon$, 
be such that
\begin{eqnarray*}
R_1 
& \leq & 
I^\epsilon_H(X : Y Z)_\rho - 1 - \log \frac{1}{\epsilon}, \\
R_2 
& \leq & 
I^\epsilon_H(Y : X Z)_\rho - 1 - \log \frac{1}{\epsilon}, \\
R_1 + R_2 
& \leq & 
I^\epsilon_H(XY : Z)_\rho - 1 - \log \frac{1}{\epsilon}. 
\end{eqnarray*} 
Then there exists an $(R_1, R_2, 49 \epsilon^{1/2})$-quantum 
MAC code for sending classical information through $\chan$.
}

\section{The one-shot classical quantum joint typicality lemma}
\label{sec:qtypical}
In this section, we state various versions of our one-shot
classical quantum joint typicality lemmas. 
We first state the `intersection case' in Lemma~\ref{lem:cqtypical},
where we only have
to take a so-called `intersection' of POVM elements. This can be viewed
as the classical quantum version of Fact~\ref{fact:ctypical} when 
$t = 1$ i.e. when only intersection of classical POVM elements needs to
be taken in Fact~\ref{fact:ctypical}, not the union.
In fact many channel coding
applications in quantum Shannon theory, e.g. the `pentagonal' inner bound 
for the cq-MAC in this paper,
use classical quantum states with
only one quantum register where only `intersection' of POVM elements
needs to be taken. For these applications, we state an even simpler
joint typicality lemma in Corollary~\ref{cor:cqtypical}.
Nevertheless there are applications that require `interesection' POVM
elements acting on classical quantum states with 
more than one quantum register e.g. Marton's inner bound with common
message for sending classical information over an entanglement assisted
quantum broadcast channel \cite{sen:simultaneous}. These applications
require the extra strength of Lemma~\ref{lem:cqtypical}.
For comprehensiveness, in the special case where there 
is no classical system
i.e. $c = 0$, we get a simplification of Lemma~\ref{lem:cqtypical} which
we state in Corollary~\ref{cor:qtypical}.

After stating and proving Lemma~\ref{lem:cqtypical} and
Corollaries~\ref{cor:qtypical}, \ref{cor:cqtypical}, we state and
prove the one-shot classical quantum conditional joint typicality
lemma, general case in Theorem~\ref{thm:cqtypical}. 
Its statement takes care of the setting when one has to take a `union of 
intersection' of
POVM elements. This can be viewed as the classical quantum version of
Fact~\ref{fact:ctypical}. We call this statement as the general 
case because it is the strongest statement that can be proved with 
the techniques developed in this paper. 
Even though most channel
coding applications use tensor products of bipartitions as the 
negative hypotheses, we nevertheless state our general case with 
tensor products of arbitrary subpartitions as the negative hypotheses.
This is for generality as well as keeping an eye on potential applications
like generalised quantum Slepian-Wolf \cite{anshu:slepianwolf} where
the negative hypothesis is a tensor product of three marginals, though
the entropic quantity in Slepian-Wolf is of the covering type where
our joint typicality lemmas do not seem to apply.
Many applications, e.g.
the Chong-Motani-Garg-El Gamal inner bound for sending
classical information over a classical quantum interference channel
\cite{sen:simultaneous}, use classical quantum states with
only one quantum register but 
involve an `union of intersection' of POVM elements. For these 
applications, we state a simpler version of Theorem~\ref{thm:cqtypical}
in Corollary~\ref{cor:gencqtypical}.

The essential part of the proofs of all the  joint typicality results
obtained in this paper is
encapsulated into a technical statement in 
Proposition~\ref{prop:cqtypical}, which should be of independent interest.
Its proof is deferred to \ref{sec:proofcqtypical}

We now state our one-shot classical quantum joint typicality lemma,
intersection case.
\begin{lemma}[cq joint typicality lem., intersec.  case]
\label{lem:cqtypical}
Let $\cH$, $\cL$ be Hilbert spaces and $\cX$ be a finite set. We will
also use $\cX$ to denote the Hilbert space with computational basis
elements indexed by the set $\cX$. Let $c$ be a non-negative and
$k$ a positive integer.
Let  $A_1 \cdots A_k$ be a $k$-partite
system where each $A_i$ is isomorphic to $\cH$.
For every $\vecx \in \cX^c$, 
let $\rho_\vecx$ be a quantum state in $A_{[k]}$.
Consider the 
{\em augmented} $k$-partite system $A'_1 \cdots A'_k$ where
each $A'_i \cong A''_i \otimes \cL$, and each 
$A''_i$ is defined as
\[
A''_i := 
(\cH \otimes \C^2) \oplus 
\bigoplus_{S: i \in S \subseteq [c] \cupdot [k]}
(\cH \otimes \C^2) \otimes \cL^{\otimes |S|}.
\]
Also define $\cX' := \cX \otimes \cL$.

Below, $\vecx$, $\vecl$ denote computational basis vectors of
$\cX^{[c]}$, $\cL^{\otimes ([c] \cupdot [k])}$.
Let $p(\cdot)$ be a probability distribution on the vectors $\vecx$.
Define the classical quantum state
\[
\rho^{\cX_{[c]} A_{[k]}} :=
\sum_\vecx 
p(\vecx) \ketbra{\vecx}^{\cX_{[c]}} \otimes
\rho_\vecx^{A_{[k]}}.
\]
Let $\frac{\one^{\cL^{\otimes (c+k)}}}{|\cL|^{c+k}}$ denote the completely
mixed state on $(c+k)$ tensor copies of $\cL$. 
View 
$
\rho_\vecx^{A_{[k]}} 
\otimes (\ketbra{0})^{(\C^2)^{\otimes k}}
\otimes \frac{\one^{\cL^{\otimes (c+k)}}}{|\cL|^{c+k}} 
$
as a state in $A'_{[k]}$ under the natural embedding viz. the
embedding in the $i$th system is into the first summand of $A''_i$
defined above tensored with $\cL$. Similarly, view 
$
\rho^{\cX_{[c]} A_{[k]}} 
\otimes (\ketbra{0})^{(\C^2)^{\otimes k}}
\otimes \frac{\one^{\cL^{\otimes (c+k)}}}{|\cL|^{c+k}} 
$
as a state in $\cX'_{[c]} A'_{[k]}$ under the natural 
embedding.

Let $0 \leq \delta \leq 1$. For each pseudosubpartition
$(S_1, \ldots, S_l) \vdash \vdash [c] \cupdot [k]$, 
let $0 \leq \epsilon_{(S_1, \ldots, S_l)} \leq 1$.
Then, there is a state $\rho'$ and a POVM element $\Pi'$ in 
$\cX'_{[c]} A'_{[k]}$ such that:
\begin{enumerate}
\item
\setcounter{lemqtypicalcq}{\value{enumi}}
The state $\rho'$ and POVM element $\Pi'$ are classical 
on $\cX^{\otimes [c]} \otimes \cL^{[c] \cupdot [k]}$ and quantum on 
$A''_{[k]}$. 
More precisely, $\rho'$, $\Pi'$ can be expressed as
\begin{eqnarray*}
(\rho')^{\cX'_{[c]} A'_{[k]}} 
& = &
|\cL|^{-(c+k)} 
\sum_{\vecx, \vecl}
p(\vecx) 
\ketbra{\vecx}^{\cX_{[c]}} \otimes
\ketbra{\vecl}^{\cL_{[c] \cupdot [k]}} \\
&   &
~~~~~~~~~~~~~~~~~~~~~~~~~~
{} \otimes
(\rho')_{\vecx, \vecl, \delta}^{A''_{[k]}}, \\
(\Pi')^{\cX'_{[c]} A'_{[k]}} 
& = &
\sum_{\vecx, \vecl}
\ketbra{\vecx}^{\cX_{[c]}} \otimes
\ketbra{\vecl}^{\cL_{[c] \cupdot [k]}} \otimes
(\Pi')_{\vecx,\vecl,\delta}^{A''_{[k]}}, \\
\end{eqnarray*}
where 
$(\rho')_{\vecx,\vecl,\delta}^{A''_{[k]}}$, 
$(\Pi')_{\vecx,\vecl,\delta}^{A''_{[k]}}$ are 
quantum states and POVM elements respectively
for all computational basis vectors 
$\vecx \in \cX^{\otimes [c]}$,
$\vecl \in \cL^{\otimes ([c] \cupdot [k])}$;

\item
\setcounter{lemqtypicaldistance}{\value{enumi}}
\[
\ellone{
(\rho')^{\cX'_{[c]} A'_{[k]}} - 
\rho^{\cX_{[c]} A_{[k]}} 
\otimes (\ketbra{0})^{(\C^2)^{\otimes k}}
\otimes \frac{\one^{\cL^{\otimes (c+k)}}}{|\cL|^{c+k}} 
} \leq
2^{\frac{c+k}{2} +1} \delta;
\]

\item
\setcounter{lemqtypicalcompleteness1}{\value{enumi}}
\begin{eqnarray*}
\lefteqn{
\Tr [(\Pi')^{\cX'_{[c]} A'_{[k]}} (\rho')^{\cX'_{[c]} A'_{[k]}}]
} \\ 
& \geq &
1 - 
\delta^{-2k} 2^{2^{ck+4} (k+1)^k} 
\sum_{(S_1, \ldots, S_l) \vdash \vdash [c] \cupdot [k]} 
\epsilon_{(S_1, \ldots, S_l)} -
2^{\frac{c+k}{2}+1} \delta;
\end{eqnarray*}

\item
\setcounter{lemqtypicalsoundness1}{\value{enumi}}
Let $(S_1, \ldots, S_l) \vdash \vdash [c] \cupdot [k]$, $l > 0$. 
Define $T := [k] \setminus (S_1 \cup \cdots \cup S_l)$.
Let $\sigma_\vecx^{A_{T}}$ be a state in $A_T$.
Let $S \subseteq [c] \cupdot [k]$, $S \cap [k] \neq \{\}$,
Let $\vecx_{[c] \cap S}$, $\vecl_S$ be computational basis vectors in 
$\cX^{\otimes ([c] \cap S}$, $\cL^{\otimes S}$.
Let $p_{[c] \setminus S}(\cdot)$ be a probability distribution on
$\cX^{\otimes ([c] \setminus S)}$.
In the following definition, let $\vecx'_{[c] \setminus S}$,
$\vecl'_{\bar{S}}$ range over all
computational basis vectors of $\cX^{\otimes ([c] \setminus S)}$,
$\cL^{\otimes \bar{S}}$.
Define a state in $A''_{S \cap [k]}$,
\begin{eqnarray*}
\lefteqn{
(\rho')_{\vecx_{S \cap [c]}, \vecl_{S}, \delta}^{A''_{S \cap [k]}}
} \\
& :=  &
|\cL|^{-|\bar{S}|} 
\sum_{\vecx'_{[c] \setminus S}, \vecl'_{\bar{S}}} 
p_{[c] \setminus S}(\vecx'_{[c] \setminus S})
\Tr_{A''_{\bar{S} \cap [k]}} [
(\rho')_{\vecx_{S \cap [c]} \vecx'_{[c] \setminus S}, 
         \vecl_{S} \vecl'_{\bar{S}}, \delta
        }^{A''_{[k]}}
].
\end{eqnarray*}
Analogously define 
\[
\rho_{\vecx_{S \cap [c]}}^{A_{S \cap [k]}} :=
\sum_{\vecx'_{[c] \setminus S}}
p_{[c] \setminus S}(\vecx'_{[c] \setminus S})
\Tr_{A_{\bar{S} \cap [k]}} [
\rho_{\vecx_{S \cap [c]} \vecx'_{[c] \setminus S}}^{A_{[k]}}
].
\]
Define 
\begin{eqnarray*}
\lefteqn{
(\rho')_{\vecx,\vecl,(S_1, \ldots, S_l),\delta}^{A''_{[k]}} 
} \\
& := &
(\rho'_{\vecx_{S_1 \cap [c]}, \vecl_{S_1},\delta})^{
A''_{S_1 \cap [k]}
} \\
&   &
{} \otimes \cdots \otimes \\
&   &
(\rho'_{\vecx_{S_l \cap [c])}, \vecl_{S_l},\delta})^{
A''_{S_l \cap [k]}
} \\
&   &
{} \otimes
(\sigma_\vecx^{A_T} \otimes (\ketbra{0}^{\C^2})^{\otimes |T|}), \\
\rho_{\vecx,(S_1, \ldots, S_l)}^{A_{[k]}} 
& := &
\rho_{\vecx_{S_1 \cap [c]}}^{A_{S_1 \cap [k]}} 
\otimes \cdots \otimes 
\rho_{\vecx_{S_l \cap [c]}}^{A_{S_l \cap [k]}} \otimes
\sigma_\vecx^{A_T}.
\end{eqnarray*}
Let $q_{(S_1, \ldots, S_l)}(\cdot)$ be a probability distribution over
vectors $\vecx$. Define 
\begin{eqnarray*}
\lefteqn{
(\rho')_{(S_1, \ldots, S_l)}^{\cX'_{[c]} A'_{[k]}} 
} \\
& := &
|\cL|^{-(c+k)} 
\sum_{\vecx, \vecl}
q_{(S_1, \ldots, S_l)}(\vecx) \\
&    &
~~~~~~~~~~~~~~
\ketbra{\vecx}^{\cX_{[c]}} \otimes
\ketbra{\vecl}^{\cL_{[c] \cupdot [k]}} \\
&     &
~~~~~~~~~~~~~~
{} \otimes
(\rho')_{\vecx,\vecl,(S_1, \ldots, S_l),\delta}^{A''_{[k]}}, \\
\rho_{(S_1, \ldots, S_l)}^{\cX_{[c]} A_{[k]}} 
& := &
\sum_{\vecx}
q_{(S_1, \ldots, S_l)}(\vecx) 
\ketbra{\vecx}^{\cX_{[c]}} 
\rho_{\vecx,(S_1, \ldots, S_l)}^{A_{[k]}}.
\end{eqnarray*}

Then,
\begin{eqnarray*}
\lefteqn{
\Tr [
(\Pi')^{\cX'_{[c]} A'_{[k]}} 
(\rho')_{(S_1, \ldots, S_l)}^{\cX'_{[c]} A'_{[k]}}
]
} \\
& \leq &
\max\left\{
2^{
    -D_H^{\epsilon_{(S_1, \ldots, S_l)}}
     (
      \rho^{\cX_{[c]} A_{[k]}} \| 
      \rho^{\cX_{[c]} A_{[k]}}_{(S_1, \ldots, S_l)}
     )
  }, 
\right.\\
&    &
~~~~~~~~~~~
\left.
 \frac{3 (2 |\cH|)^k}{\sqrt{|\cL|}}
\right\}.
\end{eqnarray*}
\end{enumerate}
\end{lemma}
\begin{proof}
The lemma is a direct consequence of Proposition~\ref{prop:cqtypical}.
Consider first the augmented $k$-partite system $A''_1 \cdots A''_k$. 
For computational basis vectors $\vecx$, $\vecl$ 
let $(\rho'_{\vecx,\vecl,\delta})^{A''_{[k]}}$ be the quantum state and
$(\Pi')_{\vecx,\vecl,\delta}^{A''_{[k]}}$ be the POVM element in 
$A''_{[k]}$ 
guaranteed by Existence Statements~\ref{prop@qtypicalstate} and
\ref{prop@qtypicalpovm} of Proposition~\ref{prop:cqtypical}. 
These quantities are used to define the state $\rho'$ and
POVM element $\Pi'$ in $\cX'_{[c]} A'_{[k]}$ as in
Claim~\arabic{lemqtypicalcq} of the lemma.

From Claim~\ref{prop@qtypicaldistance} of 
Proposition~\ref{prop:cqtypical}, 
\begin{eqnarray*}
\lefteqn{
\ellone{
(\rho')^{\cX'_{[c]} A'_{[k]}} - \rho^{\cX_{[c]} A_{[k]}} 
\otimes (\ketbra{0}^{\C^2})^{\otimes k}
\otimes \frac{\one^{\cL^{\otimes k}}}{|\cL|^{k}} 
}
} \\
&   =  &
\left\|
\left(
|\cL|^{-(c+k)} 
\sum_{\vecx,\vecl} 
p(\vecx)
\ketbra{\vecx}^{\cX_{[c]}} \otimes
\ketbra{\vecl}^{\cL^{\otimes (k+c)}} 
\right.
\right. \\
&    &
~~~~~~
{} \otimes 
\left.
\left.
\left(
(\rho')_{\vecx,\vecl,\delta}^{A''_{[k]}} - 
\rho^{A_{[k]}} 
\otimes (\ketbra{0}^{\C^2})^{\otimes k}
\right)
\right) 
\right\|_1 \\
& \leq &
|\cL|^{-(c+k)} 
\sum_{\vecx,\vecl} 
p(\vecx)
\ellone{
(\rho')_{\vecx,\vecl,\delta}^{A''_{[k]}} - 
\rho^{A_{[k]}} \otimes (\ketbra{0}^{\C^2})^{\otimes k}
} \\ 
& \leq &
2^{\frac{c+k}{2} +1} \delta,
\end{eqnarray*}
which proves Claim~\arabic{lemqtypicaldistance} of the lemma.

For each pseudosubpartition 
$(S_1, \ldots, S_l) \vdash \vdash [c] \cupdot [k]$, let
$
\Pi_{(S_1, \ldots, S_l)}^{\cX_{[c]} A_{[k]}}
$
be the optimising POVM element in the definition of
$
D_H^{\epsilon_{(S_1, \ldots, S_l)}}(
 \rho^{\cX_{[c]} A_{[k]}} \| 
 \rho^{\cX_{[c]} A_{[k]}}_{(S_1, \ldots, S_l)}
).
$
Without loss of generality, it is of the form
\[
\Pi_{\vecx, (S_1, \ldots, S_l)}^{\cX_{[c]} A_{[k]}} =
\sum_\vecx \ketbra{\vecx}^{\cX_{[c]}} \otimes
\Pi_{\vecx, (S_1, \ldots, S_l)}^{A_{[k]}},
\]
where for each $\vecx$, $\Pi_{\vecx, (S_1, \ldots, S_l)}$ is a POVM
element in $A_{[k]}$.
Define
$
\epsilon_{\vecx, (S_1, \ldots, S_l)} := 
1 -
\Tr [
\Pi_{\vecx, (S_1, \ldots, S_l)}^{A_{[k]}} \rho_\vecx^{A_{[k]}}
].
$
Then, 
$
\Pi_{\vecx, (S_1, \ldots, S_l)}^{A_{[k]}}
$
is the optimising POVM element in the definition of
\[
D_H^{\epsilon_{\vecx, (S_1, \ldots, S_l)}}(
 \rho_\vecx^{A_{[k]}} \| 
 \rho^{A_{[k]}}_{\vecx,(S_1, \ldots, S_l)}
).
\]
This implies that
\begin{equation}
\label{eq:optcqPOVM}
\begin{array}{rcl}
\lefteqn{
\sum_\vecx p(\vecx) 
\epsilon_{\vecx, (S_1, \ldots, S_l)} 
} \\
& = &
1 - \epsilon_{(S_1, \ldots, S_l)}, \\
&   & \\
\lefteqn{
\sum_\vecx q_{(S_1, \ldots, S_l)}(\vecx) 
2^{
-D_H^{\epsilon_{\vecx, (S_1, \ldots, S_l)}}(
 \rho_\vecx^{A_{[k]}} \| 
 \rho^{A_{[k]}}_{\vecx, (S_1, \ldots, S_l)}
)
} 
} \\
&   =  &
2^{
-D_H^{\epsilon_{(S_1, \ldots, S_l)}}(
 \rho^{\cX_{[c]} A_{[k]}} \| 
 \rho^{\cX_{[c]} A_{[k]}}_{(S_1, \ldots, S_l)}
)
}.
\end{array}
\end{equation}
From Claims~\ref{prop@qtypicalcompleteness}, 
\ref{prop@qtypicaldistance} of 
Proposition~\ref{prop:cqtypical},
\begin{eqnarray*}
\lefteqn{
\Tr [(\Pi')^{\cX'_{[c]} A'_{[k]}} (\rho')^{\cX'_{[c]} A'_{[k]}}]
} \\
&   =  &
|\cL|^{-(c+k)} \sum_{\vecx,\vecl} p(\vecx)
\Tr [
(\Pi')_{\vecx,\vecl,\delta}^{A''_{[k]}}
(\rho')_{\vecx,\vecl,\delta}^{A''_{[k]}} 
] \\
& \geq &
|\cL|^{-(c+k)} \sum_{\vecx,\vecl} p(\vecx)
\left(
\Tr [
(\Pi')_{\vecx,\vecl,\delta}^{A''_{[k]}}
(\rho^{A_{[k]}} \otimes (\ketbra{0}^{\C^2})^{\otimes k})
] 
\right. \\
&      &
\left.
{} -
\ellone{
(\rho')_{\vecx,\vecl,\delta}^{A''_{[k]}} - 
\rho^{A_{[k]}} \otimes (\ketbra{0}^{\C^2})^{\otimes k}
}
\right) \\
& \geq &
1 - 
\sum_\vecx p(\vecx) 
\delta^{-2k} 2^{2^{ck+4} (k+1)^k} 
\sum_{(S_1, \ldots, S_l) \vdash \vdash [c] \cupdot [k]}
\epsilon_{\vecx, (S_1, \ldots, S_l)} \\
&      &
{} -
2^{\frac{k+c}{2} + 1} \delta \\
& \geq &
1 - 
\delta^{-2k} 2^{2^{ck+4} (k+1)^k} 
\sum_{(S_1, \ldots, S_l) \vdash \vdash [c] \cupdot [k]}
\epsilon_{(S_1, \ldots, S_l)} -
2^{\frac{k+c}{2} + 1} \delta,
\end{eqnarray*}
which proves Claim~\arabic{lemqtypicalcompleteness1}
of the lemma.

Using Claims~\ref{prop@qtypicalsplitting}, 
\ref{prop@qtypicalsoundness},
\ref{prop@qtypicalellone}, \ref{prop@qtypicalellinfty}
of Proposition~\ref{prop:cqtypical} and Equation~\ref{eq:optcqPOVM}, 
we get
\begin{eqnarray*}
\lefteqn{
\Tr[(\Pi')^{\cX'_{[c]} A'_{[k]}} 
(\rho')_{(S_1, \ldots, S_{l})}^{\cX'_{[c]} A'_{[k]}}
]
} \\
&    =  &
|\cL|^{-(c+k)} \sum_{\vecx,\vecl} q_{(S_1, \ldots, S_l)}(\vecx)
\Tr [
(\Pi')_{\vecx,\vecl,\delta}^{A''_{[k]}}
(\rho')_{\vecx,\vecl,(S_1, \ldots, S_l),\delta}^{A''_{[k]}}
] \\
&   =   &
\alpha_{(S_1, \ldots, S_l), \delta} \\
&       &
|\cL|^{-(c+k)} \sum_{\vecx,\vecl} q_{(S_1, \ldots, S_l)}(\vecx) \\
&       &
\Tr [
(\Pi')_{\vecx,\vecl,\delta}^{A''_{[k]}}
(
\cT_{(S_1, \ldots S_l), \vecl, \delta}(
\rho_{\vecx,(S_1, \ldots, S_l)}^{A_{[k]}}
\otimes (\ketbra{0}^{\C^{2}})^{\otimes k} 
)
)^{A''_{[k]}}
] \\
&      &
{} +
\beta_{(S_1, \ldots, S_l), \delta} \\
&      &
~~~~~~~~~
|\cL|^{-(c+k)} \sum_{\vecx,\vecl} q_{(S_1, \ldots, S_l)}(\vecx)
\Tr [
(\Pi')_{\vecx,\vecl,\delta}^{A''_{[k]}}
N_{(S_1, \ldots S_l), \vecx, \vecl, \delta}^{A''_{[k]}}
] \\
&      &
{} +
(
1 
- \alpha_{(S_1, \ldots, S_l), \delta} 
- \beta_{(S_1, \ldots, S_l), \delta}
) \\
&      &
|\cL|^{-(c+k)} \sum_{\vecx,\vecl} q_{(S_1, \ldots, S_l)}(\vecx)
\Tr [
(\Pi')_{\vecx,\vecl,\delta}^{A''_{[k]}}
M_{(S_1, \ldots S_l), \vecx, \vecl, \delta}^{A''_{[k]}}
] \\ 
& \leq  &
\alpha_{(S_1, \ldots, S_l), \delta}
|\cL|^{-(c+k)} \sum_{\vecx,\vecl} q_{(S_1, \ldots, S_l)}(\vecx)
2^{
   -D_H^{\epsilon_{\vecx,(S_1,\ldots,S_l)}}
    (
     \rho_\vecx^{A_{[k]}} \| 
     \rho_{\vecx, (S_1, \ldots, S_l)}^{A_{[k]}}
    )
  } \\
&       &
{} + 
\beta_{(S_1, \ldots, S_l), \delta} |\cL|^{-(c+k)} \\
&       &
~~~~~~~~
\sum_{\vecx,\vecl} q_{(S_1, \ldots, S_l)}(\vecx)
\ellone{(\Pi')_{\vecx, \vecl, \delta}^{A''_{[k]}}}
\ellinfty{N_{(S_1, \ldots S_l), \vecx, \vecl, \delta}^{A''_{[k]}}}\\
&       &
{} + 
(
1 
- \alpha_{(S_1, \ldots, S_l), \delta} 
- \beta_{(S_1, \ldots, S_l), \delta}
) \\
&     &
|\cL|^{-(c+k)} \sum_{\vecx,\vecl} q_{(S_1, \ldots, S_l)}(\vecx)
\ellone{(\Pi')_{\vecx,\vecl, \delta}^{A''_{[k]}}}
\ellinfty{M_{(S_1, \ldots S_l), \vecx, \vecl, \delta}^{A''_{[k]}}} \\
& \leq  &
\alpha_{(S_1, \ldots, S_l), \delta}
2^{
   -D_H^{\epsilon_{(S_1, \ldots, S_l)}}
    (
     \rho_\vecx^{A_{[k]}} \| 
     \rho_{\vecx, (S_1, \ldots, S_l)}^{A_{[k]}}
    )
  } \\
&       &
{} + 
\beta_{(S_1, \ldots, S_l), \delta}
\frac{3 (2 |\cH|)^k}{\sqrt{|\cL|}} \\
&       &
{} +
(
1 
- \alpha_{(S_1, \ldots, S_l), \delta} 
- \beta_{(S_1, \ldots, S_l), \delta}
) 
\frac{(2 |\cH|)^k}{|\cL|} \\
&  \leq &
\max\left\{
2^{
   -D_H^{\epsilon_{(S_1, \ldots, S_l)}}
    (
     \rho^{A_{[k]}} \| 
     \rho_{(S_1, \ldots, S_l)}^{A_{[k]}}
    )
  },
\frac{3 (2 |\cH|)^k}{|\cL|}
\right\}.
\end{eqnarray*}
This proves Claim~\arabic{lemqtypicalsoundness1} of the lemma.

The proof of the one-shot classical quantum joint typicality
lemma, intersection case is finally complete.
\end{proof}

\medskip

\noindent
{\bf Remark:} 

\noindent
Often in inner bound proofs in classical network 
information theory, one has an `ideal' probability distribution
$p(\vecx_{[c]} \vecx_{[k]})$ on an alphabet $\cX^{c+k}$ which has to 
be accepted with
probability close to one, and for each subset $S \subseteq [k]$,
a `false' probability distribution 
\[
q_S(\vecx_{[c]} \vecx_S \vecx_{[k] \setminus S}) :=
p(\vecx_{[c]}) p(\vecx_S | \vecx_{[c]}) 
p(\vecx_{[k] \setminus S} | \vecx_{[c]}),
\]
which should be accepted with as low probability as possible. The
distribution $q_S(\cdot)$ is a product of the marginals on
$\cX_S$ and $\cX_{[k] \setminus S}$ conditioned on $\cX_{[c]}$,
which gives rise to the name of conditional joint typicality lemma.
For each $S \subseteq [k]$, one can take the optimal test for the 
above task, and then take the intersection of all the tests. The
preceding lemma is a generalisation of the classical 
conditional joint typicality lemma with intersection, where the 
conditioning system continues to be classical but the remaining systems
are allowed to be quantum. That is why we call it the one-shot
classical quantum joint typicality lemma, intersection case.

When there is no classical system,
i.e. $c = 0$, Lemma~\ref{lem:cqtypical} simplifies to the following
statement.
\begin{corollary}[Quant. joint typ. lem., intersec.  case]
\label{cor:qtypical}
Let $\cH$, $\cL$ be Hilbert spaces. 
Let $k$ be a positive integer.
Let  $A_1 \cdots A_k$ be a $k$-partite
system where each $A_i$ is isomorphic to $\cH$.
Let $\rho$ be a quantum state in $A_{[k]}$.
Consider the 
{\em augmented} $k$-partite system $A'_1 \cdots A'_k$ where
each $A'_i \cong A''_i \otimes \cL$, and each 
$A''_i$ is defined as
\[
A''_i := 
(\cH \otimes \C^2) \oplus 
\bigoplus_{S: i \in S \subseteq [k]}
(\cH \otimes \C^2) \otimes \cL^{\otimes |S|}.
\]

Below, $\vecl$ denotes a computational basis vector of
$\cL^{\otimes [k]}$.
Let $\frac{\one^{\cL^{\otimes k}}}{|\cL|^{k}}$ denote the completely
mixed state on $k$ tensor copies of $\cL$. 
View 
$
\rho^{A_{[k]}} 
\otimes (\ketbra{0})^{(\C^2)^{\otimes k}}
\otimes \frac{\one^{\cL^{\otimes k}}}{|\cL|^{k}} 
$
as a state in $A'_{[k]}$ under the natural embedding viz. the
embedding in the $i$th system is into the first summand of $A''_i$
defined above tensored with $\cL$. 

Let $0 \leq \epsilon, \delta \leq 1$.
Then, there is a state $\rho'$ and a POVM element $\Pi'$ in 
$A'_{[k]}$ such that:
\begin{enumerate}
\item
\setcounter{corqtypicalcq}{\value{enumi}}
The state $\rho'$ and POVM element $\Pi'$ are classical on 
$\cL^{[k]}$ and quantum on $A''_{[k]}$. 
More precisely, $\rho'$, $\Pi'$ can be expressed as
\begin{eqnarray*}
(\rho')^{A'_{[k]}} 
& = &
|\cL|^{-k} 
\sum_{\vecl}
\ketbra{\vecl}^{\cL_{[k]}} \otimes
(\rho')_{\vecl, \delta}^{A''_{[k]}}, \\
(\Pi')^{A'_{[k]}} 
& = &
\sum_{\vecl}
\ketbra{\vecl}^{\cL_{[k]}} \otimes
(\Pi')_{\vecl,\delta}^{A''_{[k]}}, \\
\end{eqnarray*}
where 
$(\rho')_{\vecl,\delta}^{A''_{[k]}}$, 
$(\Pi')_{\vecl,\delta}^{A''_{[k]}}$ are 
quantum states and POVM elements respectively
for all computational basis vectors 
$\vecl \in \cL^{\otimes [k]}$;

\item
\setcounter{corqtypicaldistance}{\value{enumi}}
\[
\ellone{
(\rho')^{A'_{[k]}} - 
\rho^{A_{[k]}} 
\otimes (\ketbra{0})^{(\C^2)^{\otimes k}}
\otimes \frac{\one^{\cL^{\otimes k}}}{|\cL|^{k}} 
} \leq
2^{\frac{k}{2} +1} \delta;
\]

\item
\setcounter{corqtypicalcompleteness1}{\value{enumi}}
\[
\Tr [(\Pi')^{A'_{[k]}} (\rho')^{A'_{[k]}}] \geq
1 - 
\delta^{-2k} 2^{17 (k+1)^k} \epsilon -
2^{\frac{k}{2}+1} \delta;
\]

\item
\setcounter{corqtypicalsoundness1}{\value{enumi}}
Let $(S_1, \ldots, S_l) \vdash [k]$, $l > 0$. 
Define $T := [k] \setminus (S_1 \cup \cdots \cup S_l)$.
Let $\sigma^{A_{T}}$ be a state in $A_T$.
Let $\{\} \neq S \subseteq [k]$.
Let $\vecl_S$ be computational basis vectors in 
$\cL^{\otimes S}$.
In the following definition, let 
$\vecl'_{\bar{S}}$ range over all
computational basis vectors of $\cL^{\otimes \bar{S}}$.
Define a state in $A''_{S}$,
\[
(\rho')_{\vecl_{S}, \delta}^{A''_{S}} := 
|\cL|^{-|\bar{S}|} 
\sum_{\vecl'_{\bar{S}}} 
\Tr_{A''_{\bar{S}}} [
(\rho')_{\vecl_{S} \vecl'_{\bar{S}}, \delta}^{A''_{[k]}}
].
\]
Analogously define 
$
\rho^{A_{S}} :=
\Tr_{A_{\bar{S}}} [
\rho^{A_{[k]}}
].
$
Define 
\begin{eqnarray*}
\lefteqn{
(\rho')_{\vecl,(S_1, \ldots, S_l),\delta}^{A''_{[k]}} 
} \\
& := &
(\rho'_{\vecl_{S_1},\delta})^{
A''_{S_1}
} 
\otimes \cdots \otimes 
(\rho'_{\vecl_{S_l},\delta})^{
A''_{S_l}
} \\
&    &
{} \otimes
(\sigma^{A_T} \otimes (\ketbra{0}^{\C^2})^{\otimes |T|}), \\
\rho_{(S_1, \ldots, S_l)}^{A_{[k]}} 
& := &
\rho^{A_{S_1}} 
\otimes \cdots \otimes 
\rho^{A_{S_l}} \otimes
\sigma^{A_T}.
\end{eqnarray*}
Define 
\begin{eqnarray*}
(\rho')_{(S_1, \ldots, S_l)}^{A'_{[k]}} 
& := &
|\cL|^{-k} 
\sum_{\vecx, \vecl}
\ketbra{\vecl}^{\cL_{[k]}} \otimes
(\rho')_{\vecl,(S_1, \ldots, S_l),\delta}^{A''_{[k]}}.
\end{eqnarray*}

Then,
\begin{eqnarray*}
\lefteqn{
\Tr [
(\Pi')^{A'_{[k]}} 
(\rho')_{(S_1, \ldots, S_l)}^{A'_{[k]}}
]
} \\
& \leq &
\max\left\{
2^{
    -D_H^\epsilon
     (
      \rho^{A_{[k]}} \| 
      \rho^{A_{[k]}}_{(S_1, \ldots, S_l)}
     )
  },
 \frac{3 (2 \cH)^k}{\sqrt{|\cL|}}
\right\}.
\end{eqnarray*}
\end{enumerate}
\end{corollary}

\medskip

\noindent
{\bf Remark:} 

\noindent
The statement of the joint typicality lemma given
at the end of Section~\ref{sec:quantunionintersection} easily follows 
from the more general statement in 
Corollary~\ref{cor:qtypical} using Proposition~\ref{prop:maxIH}.
We first observe by Proposition~\ref{prop:maxIH} that for any state
$\rho^{A_{[k]}}$ and any subset $\{\} \neq S \subset [k]$, 
\[
D_H^\epsilon 
 (
  \rho^{A_{[k]}} \| 
  \rho^{A_S} \otimes \rho^{A_{\bar{S}}}
 ) \leq 
k \log |\cH| + 3 \log \frac{1}{1-\epsilon} + 6 \log 3 - 4.
\]
Choose $\cL$ of dimension 
$\frac{3^{13} (2 |\cH|)^{4k}}{2^8 (1-\epsilon)^6}$ 
in Corollary~\ref{cor:qtypical}.
Choose the Hilbert space $\cK$ to be of dimension 
\[
2 |\cL|^{k+1} <
\frac{2^{13(k+1)} (2 |\cH|)^{4k(k+1)}}{(1 - \epsilon)^{6(k+1)}}.
\]
This makes the dimension of $\cH \otimes \cK$ large enough to contain
$A'_i$. Choose $\delta := \epsilon^{\frac{1}{4k}}$. Define 
$
\tau^{\cK^{\otimes [k]}} := 
(\ketbra{0})^{(\C^2)^{\otimes [k]}} \otimes
\frac{\one^{\cL^{\otimes [k]}}}{|\cL|^k}.
$
The statement 
at the end of Section~\ref{sec:quantunionintersection}
now follows easily from Corollary~\ref{cor:qtypical}.

In Lemma~\ref{lem:cqtypical},
in the important special case of $k = 1$, it is not necessary to 
augment the register $A_1$ with $\cL$. This is because a
pseudosubpartition of $[c] \cupdot [1]$ consists of only one subset
which must contain the register $A_1$. This simplifies the proof of 
Lemma~\ref{lem:cqtypical}. In fact,
we can prove the following joint typicality lemma in
Corollary~\ref{cor:cqtypical} which is 
quite useful
in applications to network information theory. The lemmas
allow the `negative hypotheses' to be classical quantum states where
the quantum register is either the average over some classical systems
of the positive hypothesis or a fixed quantum state. The flexibility
of allowing a  fixed quantum state is with an eye to potential applications
in Shannon theory
like generalised quantum Slepian-Wolf \cite{anshu:slepianwolf}, though
the entropic quantity in Slepian-Wolf is of the covering type where
our joint typicality lemmas do not seem to apply.
\begin{corollary}[Useful joint typ. lem., intersec. case]
\label{cor:cqtypical}
Let $\cH$, $\cL$ be Hilbert spaces and $\cX$ be a finite set. 
We will also use $\cX$ to denote the Hilbert space with computational basis
elements indexed by the set $\cX$. Let $c$ be a non-negative integer.
Let $A$ denote a quantum register with Hilbert space $\cH$.
For every $\vecx \in \cX^c$, 
let $\rho_\vecx$ be a quantum state in $A$.
Consider the extended quantum system 
\[
A' := 
(\cH \otimes \C^2) \oplus 
\bigoplus_{S: \{\} \neq S \subseteq [c]}
(\cH \otimes \C^2) \otimes \cL^{\otimes |S|}.
\]
Also define the {\em augmented} classical system 
$\cX' := \cX \otimes \cL$.

Below, $\vecx$, $\vecl$ denote computational basis vectors of
$\cX^{[c]}$, $\cL^{\otimes [c]}$.
Let $p(\cdot)$ be a probability distribution on the vectors $\vecx$.
Define the classical quantum state
\[
\rho^{\cX_{[c]} A} :=
\sum_\vecx 
p(\vecx) \ketbra{\vecx}^{\cX_{[c]}} \otimes
\rho_\vecx^{A}.
\]
Let $\frac{\one^{\cL^{\otimes c}}}{|\cL|^{c}}$ denote the completely
mixed state on $c$ tensor copies of $\cL$. 
View 
$
\rho_\vecx^{A} 
\otimes (\ketbra{0})^{\C^2}
$
as a state in $A'$ under the natural embedding viz. the
embedding is into the first summand of $A'$
defined above. Similarly, view 
$
\rho^{\cX_{[c]} A} 
\otimes (\ketbra{0})^{(\C^2)}
\otimes \frac{\one^{\cL^{\otimes c}}}{|\cL|^{c}} 
$
as a state in $\cX'_{[c]} A'$ under the natural 
embedding.

Let $0 \leq \epsilon, \delta \leq 1$. 
Choose $\cL$ to have dimension 
$|\cL| = \frac{3^{13} |\cH|^4}{2^4 (1 - \epsilon)^6}$.
Then, there is a state $\rho'$ and a POVM element $\Pi'$ in 
$\cX'_{[c]} A'$ such that:
\begin{enumerate}
\item
\setcounter{lemqtypicalcq}{\value{enumi}}
The state $\rho'$ and POVM element $\Pi'$ are classical 
on $\cX^{\otimes [c]} \otimes \cL^{[c]}$ and quantum on 
$A'$. 
More precisely, $\rho'$, $\Pi'$ can be expressed as
\begin{eqnarray*}
(\rho')^{\cX'_{[c]} A'} 
& = &
|\cL|^{-c} 
\sum_{\vecx, \vecl}
p(\vecx) 
\ketbra{\vecx}^{\cX_{[c]}} \otimes
\ketbra{\vecl}^{\cL_{[c]}} \otimes
(\rho')_{\vecx, \vecl, \delta}^{A'}, \\
(\Pi')^{\cX'_{[c]} A'} 
& = &
\sum_{\vecx, \vecl}
\ketbra{\vecx}^{\cX_{[c]}} \otimes
\ketbra{\vecl}^{\cL_{[c]}} \otimes
(\Pi')_{\vecx,\vecl,\delta}^{A'}, \\
\end{eqnarray*}
where 
$(\rho')_{\vecx,\vecl,\delta}^{A'}$, 
$(\Pi')_{\vecx,\vecl,\delta}^{A'}$ are 
quantum states and POVM elements respectively
for all computational basis vectors 
$\vecx \in \cX^{\otimes [c]}$,
$\vecl \in \cL^{\otimes [c]}$;

\item
\setcounter{lemqtypicaldistance}{\value{enumi}}
\[
\ellone{
(\rho')^{\cX'_{[c]} A'} - 
\rho^{\cX_{[c]} A} 
\otimes (\ketbra{0})^{\C^2}
\otimes \frac{\one^{\cL^{\otimes c}}}{|\cL|^c} 
} \leq
2^{\frac{c+1}{2} +1} \delta;
\]

\item
\setcounter{lemqtypicalcompleteness1}{\value{enumi}}
\[
\Tr [(\Pi')^{\cX'_{[c]} A'} (\rho')^{\cX'_{[c]} A'}] \geq
1 - 
\delta^{-2} 2^{2^{c+5}} 3^c \epsilon -
2^{\frac{c+1}{2}+1} \delta;
\]

\item
\setcounter{lemqtypicalsoundness1}{\value{enumi}}
Let $S \subseteq [c]$.
Let $\vecx_S$, $\vecl_S$ be computational basis vectors in 
$\cX^{\otimes S}$, $\cL^{\otimes S}$.
In the following definition, let $\vecx'_{\bar{S}}$,
$\vecl'_{\bar{S}}$ range over all
computational basis vectors of $\cX^{\otimes ([c] \setminus S)}$,
$\cL^{\otimes ([c] \setminus S)}$.
Define a state in $A'$,
\[
(\rho')_{\vecx_S, \vecl_S, \delta}^{A'} := 
|\cL|^{-|\bar{S}|} 
\sum_{\vecx'_{\bar{S}}, \vecl'_{\bar{S}}} 
p(\vecx'_{\bar{S}} | \vecx_S)
(\rho')_{\vecx_S \vecx'_{\bar{S}}, 
         \vecl_S \vecl'_{\bar{S}}, \delta
        }^{A'}.
\]
Analogously define 
\[
\rho_{\vecx_S}^A :=
\sum_{\vecx'_{\bar{S}}}
p(\vecx'_{\bar{S}} | \vecx_S)
\rho_{\vecx_S \vecx'_{\bar{S}}}^A.
\]
Let $(S_1, S_2, S_3) \dashv [c]$. Let $\sigma^A$ be a fixed
quantum state in $A$.
If $S_1, S_3 = \{\}$ and $S_2 = [c]$, define
\begin{eqnarray*}
\lefteqn{
(\rho')_{(\{\}, [c], \{\})}^{\cX'_{[c]} A'} 
} \\
& := &
|\cL|^{-c} 
\sum_{\vecx}
p(\vecx)
\ketbra{\vecx}^{\cX_{[c]}} \otimes
\ketbra{\vecl}^{\cL_{[c]}} \\
&   &
~~~~
{} \otimes
(\sigma^A \otimes \ketbra{0}^{\C^2}), \\
\rho_{(\{\}, [c], \{\})}^{\cX_{[c]} A} 
& := &
\sum_{\vecx}
p(\vecx)
\ketbra{\vecx}^{\cX_{[c]}} \otimes \sigma^A.
\end{eqnarray*}
Otherwise, define 
\begin{eqnarray*}
\lefteqn{
(\rho')_{(S_1, S_2, S_3)}^{\cX'_{[c]} A'} 
} \\
& := &
|\cL|^{-c} 
\sum_{\vecx_{S_1}}
p(\vecx_{S_1}) 
\ketbra{\vecx_{S_1}}^{\cX_{S_1}} \\
&   &
~~~~~~~~~~~~~~~~~
{} \otimes
\ketbra{\vecl_{S_1}}^{\cL_{S_1}} \\
&    &
~~~~~~~~~~~~~~
{} \otimes
\left(
\sum_{\vecx_{S_2}}
p(\vecx_{S_2} | \vecx_{S_1})
\ketbra{\vecx_{S_2}}^{\cX_{S_2}} 
\right. \\
&    &
~~~~~~~~~~~~~~~~~~~~~
\left.
{} \otimes
\ketbra{\vecl_{S_2}}^{\cL_{S_2}} 
\right) \\
&    &
~~~~~~~~~~~~~~
{} \otimes
\left(
\sum_{\vecx_{S_3}}
p(\vecx_{S_3} | \vecx_{S_1})
\ketbra{\vecx_{S_3}}^{\cX_{S_3}} 
\right. \\
&    &
~~~~~~~~~~~~~~~~~~~~~
{} \otimes
\ketbra{\vecl_{S_3}}^{\cL_{S_3}} \\
&   &
~~~~~~~~~~~~~~~~~~~~~
\left.
{} \otimes 
(\rho')_{\vecx_{S_1 \cup S_3}, \vecl_{S_1 \cup S_3}, \delta}^{A'}
\right), \\
\lefteqn{
\rho_{(S_1, S_2, S_3)}^{\cX_{[c]} A} 
} \\
& := &
\sum_{\vecx_{S_1}}
p(\vecx_{S_1}) 
\ketbra{\vecx_{S_1}}^{\cX_{S_1}} \\
&    &
~~~~~~~~~~~
{} \otimes
\left(
\sum_{\vecx_{S_2}}
p(\vecx_{S_2} | \vecx_{S_1})
\ketbra{\vecx_{S_2}}^{\cX_{S_2}} 
\right) \\
&    &
~~~~~~~~~~~
{} \otimes
\left(
\sum_{\vecx_{S_3}}
p(\vecx_{S_3} | \vecx_{S_1})
\ketbra{\vecx_{S_3}}^{\cX_{S_3}} 
\right. \\
&   &
~~~~~~~~~~~~~~~~
\left.
{} \otimes
\rho_{\vecx_{S_1 \cup S_3}}^{A}
\right).
\end{eqnarray*}

Then,
\[
\Tr [
(\Pi')^{\cX'_{[c]} A'} 
(\rho')_{(S_1, S_2, S_3)}^{\cX'_{[c]} A'}
]
\leq
2^{-I_H^{\epsilon}(X_{S_2} : A X_{S_3} | X_{S_1})_\rho},
\]
where 
$
I_H^{\epsilon}(X_{S_2} : A X_{S_3} | X_{S_1})_\rho :=
D_H^{\epsilon}
    (
      \rho^{\cX_{[c]} A} \| 
      \rho^{\cX_{[c]} A}_{(S_1, S_2, S_3)}
    ).
$
\end{enumerate}
\end{corollary}

We now state our most general one-shot classical quantum joint typicality
lemma.
\begin{theorem}[cq joint typ. lem., gen. case]
\label{thm:cqtypical}
Let $\cH$, $\cL$ be Hilbert spaces and $\cX$ be a finite set. We will
also use $\cX$ to denote the Hilbert space with computational basis
elements indexed by the set $\cX$. Let $c$ be a non-negative and
$k$ a positive integer.
Let  $A_1 \cdots A_k$ be a $k$-partite
system where each $A_i$ is isomorphic to $\cH$.
Let $t$ be a positive integer.
Let $\vecx^{t}$ denote a $t$-tuple of elements of $\cX^c$;
we shall denote its $i$th element by $\vecx^{t}(i)$.
Consider the 
{\em extended} $k$-partite system $\hat{A}_1 \cdots \hat{A}_k$ where
each $\hat{A}_i \cong A'_i \otimes \C^2 \otimes \C^{t+1}$, 
$A'_i \cong A''_i \otimes \cL$, and each 
$A''_i$ is defined as
\[
A''_i := 
(\cH \otimes \C^2) \oplus 
\bigoplus_{S: i \in S \subseteq [c] \cupdot [k]}
(\cH \otimes \C^2) \otimes \cL^{\otimes |S|}.
\]
Also define $\hat{\cX} := \cX \otimes \cL$.

Below, $\vecx$ denotes computational basis vectors of
$\cX^{[c]}$, and $\vecl$ denotes computational basis vectors of
$\cL^{\otimes s}$ where $s$ will be clear from the context.
Let $p(\cdot)$ denote a probability distribution on 
the vectors $\vecx$. Let $p(1; \cdot), \ldots, p(t; \cdot)$ denote
probability distributions on $\vecx^t$ such that the marginal
of $p(i; \vecx^t)$ on the $i$th element is $p(\vecx^t(i))$.
For $i \in [t]$, define the classical quantum states
\[
\rho^{(\cX_{[c]})^t A_{[k]}}(i) :=
\sum_{\vecx^t}
p(i; \vecx^t) \ketbra{\vecx^t}^{(\cX_{[c]})^t} \otimes
\rho_{\vecx^t(i)}^{A_{[k]}}.
\]
Let $\frac{\one^{\cL^{\otimes k}}}{|\cL|^{k}}$ denote the completely
mixed state on $(c+k)$ tensor copies of $\cL$. 
View 
$
\rho_\vecx^{A_{[k]}}(i) 
\otimes (\ketbra{0})^{(\C^2)^{\otimes k}}
\otimes \frac{\one^{\cL^{\otimes k}}}{|\cL|^{k}} 
\otimes (\ketbra{0})^{(\C^2)^{\otimes k}}
\otimes (\ketbra{0})^{(\C^{t+1})^{\otimes k}}
$
as a state in $\hat{A}_{[k]}$ under the natural embedding viz. the
embedding in the $j$th system is into the first summand of $A''_j$
defined above tensored with $\cL \otimes \C^2 \otimes \C^{t+1}$. 
Similarly, view 
$
(\rho(i))^{(\cX_{[c]})^t A_{[k]}} 
\otimes (\ketbra{0})^{(\C^2)^{\otimes k}}
\otimes \frac{\one^{\cL^{\otimes (ct+k)}}}{|\cL|^{ct+k}} 
\otimes (\ketbra{0})^{(\C^2)^{\otimes k}}
\otimes (\ketbra{0})^{(\C^{t+1})^{\otimes k}}
$
as a state in $(\hat{\cX}_{[c]})^t \hat{A}_{[k]}$ under the natural 
embedding.

Let $0 \leq \alpha, \epsilon, \delta \leq 1$. For each pseudosubpartition\\
$(S_1, \ldots, S_l) \vdash \vdash [c] \cupdot [k]$ and each $i \in [t]$,
let $0 \leq \epsilon_{i, (S_1, \ldots, S_l)} \leq 1$.
Then, there are states $\rho'(1), \ldots, \rho'(t)$ and a  
POVM element $\hat{\Pi}$ in 
$(\hat{\cX}_{[c]})^t \hat{A}_{[k]}$ such that:
\begin{enumerate}
\item
\setcounter{thmqtypicalcq}{\value{enumi}}
The states $\rho'(1), \ldots, \rho'(t)$ and 
POVM element $\hat{\Pi}$ are classical 
on $\cX^{\otimes [ct]} \otimes \cL^{[ct] \cupdot [k]}$ and quantum on 
$A''_{[k]}  \otimes (\C^2)^{\otimes [k]} \otimes (\C^{t+1})^{\otimes [k]}$.
More precisely, $\rho'(i)$, $i \in [t]$, 
$\hat{\Pi}$ can be expressed as
\begin{eqnarray*}
\lefteqn{
(\rho'(i))^{(\hat{\cX}_{[c]})^t \hat{A}_{[k]}} 
} \\
& = &
|\cL|^{-(ct+k)} 
\sum_{\vecx^t, \vecl}
p(i; \vecx^t) 
\ketbra{\vecx^t}^{(\cX_{[c]})^t} \otimes
\ketbra{\vecl}^{\cL_{[tc] \cupdot [k]}} \\
&   &
~~~~~~~~~~~~~~~~~~~~~~
{} \otimes
(\rho')_{\vecx^t(i), \vecl_{[ct]}(i) \vecl_{[k]}, \delta}^{A''_{[k]}} \\
&   &
~~~~~~~~~~~~~~~~~~~~~~~~~
{} \otimes (\ketbra{0})^{(\C^2)^{\otimes k}}
\otimes (\ketbra{0})^{(\C^{t+1})^{\otimes k}}, \\
\lefteqn{
(\Pi')^{(\hat{\cX}_{[c]})^t \hat{A}_{[k]}} 
} \\
& = &
\sum_{\vecx, \vecl}
\ketbra{\vecx^t}^{(\cX_{[c]})^t} \otimes
\ketbra{\vecl}^{\cL_{[ct] \cupdot [k]}} \\
&   &
~~~~~~~~
{} \otimes
(\hat{\Pi})_{\vecx^t,\vecl,\delta}^{
A''_{[k]} 
\otimes (\C^2)^{\otimes [k]}
\otimes (\C^{t+1})^{\otimes [k]}
}, \\
\end{eqnarray*}
where 
$(\rho')_{\vecx,\vecl,\delta}^{A''_{[k]}}$
are quantum states 
for all computational basis vectors 
$\vecx \in \cX^{\otimes [c]}$,
$\vecl \in \cL^{\otimes ([c] \cupdot [k])}$ and
\[
(\Pi')_{\vecx^t,\vecl,\delta}^{
A''_{[k]}
\otimes (\C^2)^{\otimes [k]}
\otimes (\C^{t+1})^{\otimes [k]}
}
\]
are POVM elements
for all computational basis vectors 
$\vecx^t \in \cX^{\otimes [ct]}$,
$\vecl \in \cL^{\otimes ([ct] \cupdot [k])}$;

\item
\setcounter{thmqtypicaldistance}{\value{enumi}}
For all $i \in [t]$,
\begin{eqnarray*}
\lefteqn{
\left\|
(\rho'(i))^{(\hat{\cX}_{[c]})^t \hat{A}_{[k]}} 
\right.
} \\
&   &
{} -
(\rho(i))^{(\cX_{[c]})^t A_{[k]}} 
\otimes (\ketbra{0})^{(\C^2)^{\otimes k}}
\otimes \frac{\one^{\cL^{\otimes (ct+k)}}}{|\cL|^{tc+k}} \\
&   &
~~~~~~
\left.
{} \otimes (\ketbra{0})^{(\C^2)^{\otimes k}}
\otimes (\ketbra{0})^{(\C^{t+1})^{\otimes k}}
\right\|_1 \\
& \leq &
2^{\frac{c+k}{2} +1} \delta;
\end{eqnarray*}

\item
\setcounter{thmqtypicalcompleteness1}{\value{enumi}}
For all $i \in [t]$,
\begin{eqnarray*}
\lefteqn{
\Tr [
(\hat{\Pi})^{(\hat{\cX}_{[c]})^t \hat{A}_{[k]}} 
(\rho'(i))^{(\hat{\cX}_{[c]})^t \hat{A}_{[k]}}
] 
} \\
& \geq &
1 - 
\delta^{-2k} 2^{2^{ck+4} (k+1)^k} 
\sum_{(S_1, \ldots, S_l) \vdash \vdash [c] \cupdot [k]}
\epsilon_{i, (S_1, \ldots, S_l)} \\
&    &
{} -
2^{\frac{c+k}{2}+1} \delta - \alpha;
\end{eqnarray*}

\item
\setcounter{thmqtypicalsoundness1}{\value{enumi}}
Let $(S_1, \ldots, S_l) \vdash \vdash [c] \cupdot [k]$, $l > 0$. 
Define $T := [k] \setminus (S_1 \cup \cdots \cup S_l)$.
Let $\sigma_\vecx^{A_{T}}$ be a state in $A_T$.
Let $S \subseteq [c] \cupdot [k]$, $S \cap [k] \neq \{\}$,
Let $\vecx_{[c] \cap S}$, $\vecl_S$ be computational basis vectors in 
$\cX^{\otimes ([c] \cap S}$, $\cL^{\otimes S}$.
Let $p_{[c] \setminus S}(\cdot)$ be a probability distribution on
$\cX^{\otimes ([c] \setminus S)}$.
In the following definition, let $\vecx'_{[c] \setminus S}$,
$\vecl'_{\bar{S}}$ range over all
computational basis vectors of $\cX^{\otimes ([c] \setminus S)}$,
$\cL^{\otimes \bar{S}}$.
Define a state in $A''_{S \cap [k]}$,
\begin{eqnarray*}
\lefteqn{
(\rho')_{\vecx_{S \cap [c]}, \vecl_{S}, \delta}^{A''_{S \cap [k]}} 
} \\
& := &
|\cL|^{-|\bar{S}|} 
\sum_{\vecx'_{[c] \setminus S}, \vecl'_{\bar{S}}} 
p_{[c] \setminus S}(\vecx'_{[c] \setminus S})
\Tr_{A''_{\bar{S} \cap [k]}} [
(\rho')_{\vecx_{S \cap [c]} \vecx'_{[c] \setminus S}, 
         \vecl_{S} \vecl'_{\bar{S}}, \delta
        }^{A''_{[k]}})
].
\end{eqnarray*}
Analogously define 
\[
\rho_{\vecx_{S \cap [c]}}^{A_{S \cap [k]}} :=
\sum_{\vecx'_{[c] \setminus S}}
p_{[c] \setminus S}(\vecx'_{[c] \setminus S})
\Tr_{A_{\bar{S} \cap [k]}} [
\rho_{\vecx_{S \cap [c]} \vecx'_{[c] \setminus S}}^{A_{[k]}}
].
\]
Define 
\begin{eqnarray*}
\lefteqn{
(\rho')_{\vecx,\vecl,(S_1, \ldots, S_l),\delta}^{A''_{[k]}} 
} \\
& := &
(\rho'_{\vecx_{S_1 \cap [c]}, \vecl_{S_1},\delta})^{
A''_{S_1 \cap [k]}
} \\
&    &
{} \otimes \cdots \otimes \\
&    &
(\rho'_{\vecx_{S_l \cap [c])}, \vecl_{S_l},\delta})^{
A''_{S_l \cap [k]}
} \\
&    &
{} \otimes
(\sigma_\vecx^{A_T} \otimes (\ketbra{0}^{\C^2})^{\otimes |T|}), \\
\rho_{\vecx,(S_1, \ldots, S_l)}^{A_{[k]}} 
& := &
\rho_{\vecx_{S_1 \cap [c]}}^{A_{S_1 \cap [k]}} 
\otimes \cdots \otimes 
\rho_{\vecx_{S_l \cap [c]}}^{A_{S_l \cap [k]}} \otimes
\sigma_\vecx^{A_T}.
\end{eqnarray*}

For $i \in [t]$, let $q_{i; (S_1, \ldots, S_l)}(\cdot)$ be 
a probability distribution over $\vecx^t$. Define 
\begin{eqnarray*}
\lefteqn{
(\rho')_{i; (S_1, \ldots, S_l)}^{(\hat{\cX}_{[c]})^t \hat{A}_{[k]}} 
} \\
& := &
|\cL|^{-(ct+k)} 
\sum_{\vecx^t, \vecl}
q_{i; (S_1, \ldots, S_l)}(\vecx^t) 
\ketbra{\vecx^t}^{(\cX_{[c]})^t} \\
&    &
~~~~
{} \otimes
\ketbra{\vecl}^{\cL_{[ct] \cupdot [k]}} \\
&   &
~~~~~
{} \otimes
(\rho')_{\vecx^t(i),\vecl_{[ct]}(i),(S_1, \ldots, S_l),\delta}^{A''_{[k]}}
\\
&   &
~~~~~
{} \otimes (\ketbra{0})^{(\C^2)^{\otimes k}}
\otimes (\ketbra{0})^{(\C^{t+1})^{\otimes k}}, \\
\lefteqn{
\rho_{i; (S_1, \ldots, S_l)}^{(\cX_{[c]})^t A_{[k]}} 
} \\
& := &
\sum_{\vecx}
q_{i; (S_1, \ldots, S_l)}(\vecx^t) 
\ketbra{\vecx^t}^{(\cX_{[c]})^t} \otimes 
\rho_{\vecx^t(i),(S_1, \ldots, S_l)}^{A_{[k]}}.
\end{eqnarray*}

Then,
\begin{eqnarray*}
\lefteqn{
\Tr [
(\hat{\Pi})^{(\hat{\cX}_{[c]})^t \hat{A}_{[k]}} 
(\rho')_{i; (S_1, \ldots, S_l)}^{(\hat{\cX}_{[c]})^t \hat{A}_{[k]}}
]
} \\
& \leq &
\frac{1-\alpha}{\alpha}
\sum_{j=1}^t \\
&    &
\max\left\{
2^{
    -D_H^{\epsilon_{j, (S_1, \ldots, S_l)}}
     (
      \rho(j)^{(\cX_{[c]})^t A_{[k]}} \| 
      \rho^{(\cX_{[c]})^t A_{[k]}}_{i; (S_1, \ldots, S_l)}
     )
  },
\right. \\
&   &
~~~~~~~~~~~~~~~
\left.
 \frac{3 (2 |\cH|)^k}{\sqrt{|\cL|}}
\right\}.
\end{eqnarray*}
\end{enumerate}
\end{theorem}
\begin{proof}
For each $\vecx \in \cX_{[c]}$, 
$\vecl \in \cL_{[c] \cupdot [k]}$, construct the state 
$(\rho')_{\vecx,\vecl,\delta}^{A''_{[k]}}$ and POVM element
$(\Pi')_{\vecx,\vecl,\delta}^{A''_{[k]}}$ 
as in Claim~\arabic{lemqtypicalcq} of Lemma~\ref{lem:cqtypical}. 
For each $\vecx^t \in (\cX_{[c]})^t$, $\vecl \in \cL_{[ct] \cupdot [k]}$,
$i \in [t]$ define the POVM element
\[
(\Pi'(i))_{\vecx^t, \vecl,\delta}^{A''_{[k]}} :=
(\Pi')_{\vecx^t(i), \vecl_{[ct]}(i), \delta}^{A''_{[k]}}.
\]  
By Fact~\ref{fact:gelfandnaimark},
there exists an orthogonal projector
$
(\hat{\Pi}(i))_{\vecx^t,\vecl,\delta}^{
A''_{[k]} \otimes (\C^2)^{\otimes [k]}
}
$
such that 
\[
\Tr [
(\hat{\Pi}(i))_{\vecx^t,\vecl,\delta}^{
A''_{[k]} \otimes (\C^2)^{\otimes [k]}
}
(
\tau^{A''_{[k]}} 
\otimes (\ketbra{0})^{(\C^2)^{\otimes [k]}}
)
] =
\Tr [
(\Pi'(i))_{\vecx^t,\vecl,\delta}^{A''_{[k]}}
\;
\tau^{A''_{[k]}} 
]
\]
for all states 
$
\tau^{A''_{[k]}}.
$
The projector 
$
(\hat{\Pi})_{\vecx^t,\vecl,\delta}^{
A''_{[k]} \otimes (\C^2)^{\otimes [k]} \otimes (\C^{t+1})^{\otimes [k]}
}
$
is now constructed from the projectors
$
(\hat{\Pi}(i))_{\vecx^t,\vecl,\delta}^{
A''_{[k]} \otimes (\C^2)^{\otimes [k]}
},
$
$i \in [t]$
using Proposition~\ref{prop:tiltedspan}.
This settles Claim~\arabic{thmqtypicalcq} of the theorem.

Claims~\arabic{thmqtypicaldistance} and \arabic{thmqtypicalcompleteness1}
of the theorem follow from the corresponding 
Claims~\arabic{lemqtypicaldistance} and \arabic{lemqtypicalcompleteness1}
of Lemma~\ref{lem:cqtypical}.
Claim~\arabic{thmqtypicalsoundness1} of the theorem follows from
Claim~\arabic{lemqtypicalsoundness1} of Lemma~\ref{lem:cqtypical} 
and Proposition~\ref{prop:tiltedspan}, combined with the observation
that averaging over $\vecx \in \cX_{[c]}$, 
$\vecl \in \cL_{[c] \cupdot [k]}$ does not affect the 
$(S_1, \ldots, S_l)$ `structure' of the state
$(\rho')_{\vecx,\vecl,(S_1, \ldots, S_l),\delta}^{A''_{[k]}}$.
\end{proof}

\medskip

\noindent
{\bf Remark:} 

\noindent
The above theorem is a generalisation of
Fact~\ref{fact:ctypical} to the classical quantum setting involving
`union of intersection of POVM elements'. However, it has the shortcoming
that it can only handle classical quantum `$q(\cdot)$' states 
corresponding to pseudosubpartitions of $[c] \cupdot [k]$, unlike 
Fact~\ref{fact:ctypical} which can handle any classical `$q(\cdot)$' 
state. Overcoming this shortcoming remains an important open 
problem.

As remarked before Corollary~\ref{cor:cqtypical},
in the important special case of $k = 1$, it is not necessary to 
augment the register $A_1$ with $\cL$. Thus,
we can prove the following joint typicality lemma which is
very useful in applications to network information theory. 
\begin{corollary}[Useful joint typicality lemma, general case]
\label{cor:gencqtypical}
Let $\cH$, $\cL$ be Hilbert spaces and $\cX$ be a finite set. We will
also use $\cX$ to denote the Hilbert space with computational basis
elements indexed by the set $\cX$. Let $c$ be a non-negative integer.
Let  $A$ denote a quantum register with Hilbert space $\cH$.
For every $\vecx \in \cX^c$, 
let $\rho_\vecx$ be a quantum state in $A$.
Let $t$ be a positive integer.
Let $\vecx^{t}$ denote a $t$-tuple of elements of $\cX^c$;
we shall denote its $i$th element by $\vecx^{t}(i)$.
Consider the 
extended quantum system $\hat{A}$ where
$\hat{A} \cong A' \otimes \C^2 \otimes \C^{t+1}$, and
$A'$ is defined as
\[
A' := 
(\cH \otimes \C^2) \oplus 
\bigoplus_{S: i \in S \subseteq [c] \cupdot [k]}
(\cH \otimes \C^2) \otimes \cL^{\otimes |S|}.
\]
Also define the {\em augmented} classical system 
$\hat{\cX} := \cX \otimes \cL$.

Below, $\vecx$, $\vecl$ denote computational basis vectors of
$\cX^{[c]}$, $\cL^{\otimes [c]}$.
Let $p(\cdot)$ denote a probability distribution on 
the vectors $\vecx$. Let $p(1; \cdot), \ldots, p(t; \cdot)$ denote
probability distributions on $\vecx^t$ such that the marginal
of $p(i; \vecx^t)$ on the $i$th element is $p(\vecx^t(i))$.
For $i \in [t]$, define the classical quantum states
\[
\rho^{(\cX_{[c]})^t A}(i) :=
\sum_{\vecx^{t}} 
p(i; \vecx^t) \ketbra{\vecx^t}^{(\cX_{[c]})^t} \otimes
\rho_{\vecx^t(i)}^{A}.
\]
Let $\frac{\one^{\cL^{\otimes c}}}{|\cL|^{c}}$ denote the completely
mixed state on $c$ tensor copies of $\cL$. 
View 
$
\rho_\vecx^{A} 
\otimes (\ketbra{0})^{\C^2}
\otimes (\ketbra{0})^{\C^2}
\otimes (\ketbra{0})^{\C^{t+1}}
$
as a state in $\hat{A}$ under the natural embedding viz. the
embedding is into the first summand of $A'$
defined above tensored with $\C^2 \otimes \C^{t+1}$. 
Similarly, view 
$
\rho^{(\cX_{[c]})^t A}(i)
\otimes (\ketbra{0})^{\C^2}
\otimes \frac{\one^{\cL^{\otimes ct}}}{|\cL|^{ct}} 
\otimes (\ketbra{0})^{\C^2}
\otimes (\ketbra{0})^{\C^{t+1}}
$
as a state in $(\hat{\cX}_{[c]})^{\otimes t} \hat{A}$ under the natural 
embedding.

Let $0 \leq \alpha, \epsilon, \delta \leq 1$. 
Choose $\cL$ to have dimension 
$|\cL| = \frac{3^{13} |\cH|^4}{2^4 (1 - \epsilon)^6}$.
Then, there are states $\rho'(1), \ldots, \rho'(t)$ and a  
POVM element $\hat{\Pi}$ in 
$(\hat{\cX}_{[c]})^{\otimes t} \hat{A}$ such that:
\begin{enumerate}
\item
\setcounter{thmqtypicalcq}{\value{enumi}}
The states $\rho'(1), \ldots, \rho'(t)$ and 
POVM element $\hat{\Pi}$ are classical 
on $\cX^{\otimes [ct]} \otimes \cL^{[ct]}$ and quantum on 
$\hat{A}$.
More precisely, $\rho'(i)$, $i \in [t]$, 
$\hat{\Pi}$ can be expressed as
\begin{eqnarray*}
\lefteqn{
(\rho'(i))^{(\hat{\cX}_{[c]})^t \hat{A}} 
} \\
& = &
|\cL|^{-ct} 
\sum_{\vecx^t, \vecl^t}
p(i; \vecx^t) 
\ketbra{\vecx^t}^{(\cX_{[c]})^{\otimes t}} \otimes
\ketbra{\vecl^t}^{(\cL_{[c]})^{\otimes t}} \\
&  &
~~~~~
{} \otimes
(\rho')_{\vecx^t(i), \vecl^t(i), \delta}^{A'}
\otimes (\ketbra{0})^{\C^2}
\otimes (\ketbra{0})^{\C^{t+1}}, \\
\lefteqn{
(\hPi)^{(\hat{\cX}_{[c]})^t \hat{A}} 
} \\
& = &
\sum_{\vecx^t, \vecl^t}
\ketbra{\vecx^t}^{(\cX_{[c]})^{\otimes t}} \otimes
\ketbra{\vecl^t}^{(\cL_{[c]})^{\otimes t}} \otimes
(\hat{\Pi})_{\vecx^t,\vecl^t,\delta}^{\hat{A}}, \\
\end{eqnarray*}
where 
$(\rho')_{\vecx,\vecl,\delta}^{A'}$ 
are quantum states for all computational basis vectors 
$\vecx \in \cX^{\otimes [c]}$, $\vecl \in \cL^{\otimes [c]}$ and
$
(\hPi)_{\vecx^t,\vecl^t,\delta}^{\hat{A}}
$ 
are POVM elements for all computational basis vectors 
$\vecx^t \in \cX^{\otimes [ct]}$,
$\vecl^t \in \cL^{\otimes [ct]}$;

\item
\setcounter{thmqtypicaldistance}{\value{enumi}}
For all $i \in [t]$,
\begin{eqnarray*}
\lefteqn{
\left\|
(\rho'(i))^{(\hat{\cX}_{[c]})^t \hat{A}}
\right.
} \\
&   &
{} - 
(\rho(i))^{(\cX_{[c]})^t A} 
\otimes (\ketbra{0})^{\C^2}
\otimes \frac{\one^{\cL^{\otimes ct}}}{|\cL|^{ct}} \\
&   &
~~~~~~~~~
\left.
{} \otimes (\ketbra{0})^{\C^2}
\otimes (\ketbra{0})^{\C^{t+1}}
\right\|_1 \\
& \leq &
2^{\frac{c+1}{2} +1} \delta;
\end{eqnarray*}

\item
\setcounter{thmqtypicalcompleteness1}{\value{enumi}}
For all $i \in [t]$,
\[
\Tr [
(\hat{\Pi})^{(\hat{\cX}_{[c]})^t \hat{A}} 
(\rho'(i))^{(\hat{\cX}_{[c]})^t \hat{A}}
] \geq
1 - 
\delta^{-2} 2^{2^{c+5}} 3^c \epsilon -
2^{\frac{c+1}{2}+1} \delta - \alpha;
\]

\item
\setcounter{thmqtypicalsoundness1}{\value{enumi}}
Let $S \subseteq [c]$.
Let $\vecx_{S}$, $\vecl_S$ be computational basis vectors in 
$\cX^{\otimes S}$, $\cL^{\otimes S}$.
In the following definition, let $\vecx'_{\bar{S}}$,
$\vecl'_{\bar{S}}$ range over all
computational basis vectors of $\cX^{\otimes ([c] \setminus S)}$,
$\cL^{\otimes ([c] \setminus S)}$.
For $S \neq \{\}$, define states in $A'$,
\[
(\rho')_{\vecx_{S}, \vecl_{S}, \delta}^{A'} := 
|\cL|^{-|\bar{S}|} 
\sum_{\vecx'_{\bar{S}}, \vecl'_{\bar{S}}} 
p(\vecx'_{\bar{S}} | \vecx_S)
(\rho')_{\vecx_{S} \vecx'_{\bar{S}}, 
         \vecl_{S} \vecl'_{\bar{S}}, \delta
        }^{A'}.
\]
Analogously define 
\[
\rho_{\vecx_S}^{A} :=
\sum_{\vecx'_{\bar{S}}}
p(\vecx'_{\bar{S}} | \vecx_S)
\rho_{\vecx_{S} \vecx'_{\bar{S}}}^{A}.
\]
Let $\sigma^A$ be a fixed quantum state in $A$. For $S = \{\}$, define
\[
(\rho')_{\{\}, \{\}, \delta}^{A'} := \sigma^A \otimes \ketbra{0}^{\C^2},
~~~~
\rho_{\{\}}^{A} := \sigma^A.
\]

For $i \in [t]$, $S \subseteq [c]$, let $q_{i; S}(\cdot)$ be a
probability distribution on $\vecx^t$. 
Define 
\begin{eqnarray*}
\lefteqn{
(\rho')_{i; S}^{(\hat{\cX}_{[c]})^t \hat{A}}
} \\
& := &
|\cL|^{-ct} 
\sum_{\vecx^t}
q_{i; S}(\vecx^t)
\ketbra{\vecx^t}^{\cX^{\otimes [ct]}} \otimes
\ketbra{\vecl^t}^{\cL^{\otimes [ct]}} \\
&   &
~~~~~~~~~~~~~~
{} \otimes
(\rho')_{\vecx^t(i)_{S}, \vecl^t(i)_{S}, \delta}^{A'}, \\
\rho_{i; S}^{(\cX_{[c]})^t A}
& := &
\sum_{\vecx^t}
q_{i; S}(\vecx^t)
\ketbra{\vecx^t}^{\cX^{\otimes [ct]}} \otimes
\rho_{\vecx^t(i)_{S}}^{A}.
\end{eqnarray*}

Then,
\[
\Tr [
(\hat{\Pi})^{(\hat{\cX}_{[c]})^t \hat{A}} 
(\rho')_{i; S}^{(\hat{\cX}_{[c]})^t \hat{A}}
]
\leq
\frac{1-\alpha}{\alpha}
\sum_{j=1}^t
2^{
    -D_H^{\epsilon}
     (
      \rho(j)^{(\cX_{[c]})^t A} \| 
      \rho_{i; S}^{(\cX_{[c]})^t A}
     )
  }.
\]
\end{enumerate}
\end{corollary}

We now state our main technical proposition.
\begin{proposition}
\label{prop:cqtypical}
Let $\cH$, $\cL$ be Hilbert spaces and $\cX$ be a finite set. We will
also use $\cX$ to denote the Hilbert space with computational basis
elements indexed by the set $\cX$. Let $c$ be a non-negative and
$k$ a positive integer.
Let  $A_1 \cdots A_k$ be a $k$-partite
system where each $A_i$ is isomorphic to $\cH$.
For every $\vecx \in \cX^c$, 
let $\rho_\vecx$ be a quantum state in $A_{[k]}$.
Consider the 
{\em augmented} $k$-partite system $A''_1 \cdots A''_k$ where
each $A''_i$ is isomorphic to 
\[
(\cH \otimes \C^2) \oplus 
\bigoplus_{S: i \in S \subseteq [c] \cupdot [k]}
(\cH \otimes \C^2) \otimes \cL^{\otimes |S|}.
\]
View 
$
\rho_\vecx^{A_{[k]}} 
\otimes (\ketbra{0})^{(\C^2)^{\otimes k}}
$
as a state in $A''_{[k]}$ under the natural embedding viz. the
embedding in the $i$th system is into the first summand of $A''_i$
defined above.

Below, $\vecx$, $\vecl$ denote computational basis vectors of
$\cX^{[c]}$, $\cL^{\otimes ([c] \cupdot [k])}$.
Let $0 \leq \delta \leq 1$. For each $\vecx \in \cX^c$ and 
each pseudosubpartition
$(S_1, \ldots, S_l) \vdash \vdash [c] \cupdot [k]$, let 
$0 \leq \epsilon_{\vecx, (S_1, \ldots, S_l)} \leq 1$. Let 
$
\epsilon_\vecx := 
\sum_{(S_1, \ldots, S_l) \vdash \vdash [c] \cupdot [k]}
\epsilon_{\vecx, (S_1, \ldots, S_l)}.
$
Then there exist:
\begin{itemize}
\item 
\setcounter{propqtypicalstate}{\value{enumi}}
States $\rho'_{\vecx, \vecl, \delta}$ in $A''_{[k]}$ for every
$\vecx$, $\vecl$;

\item 
\setcounter{propqtypicalpovm}{\value{enumi}}
POVM elements $\Pi'_{\vecx, \vecl, \delta}$ in $A''_{[k]}$ for every
$\vecx$, $\vecl$; 

\item
For every $\vecl$ and every 
$(S_1, \ldots, S_l) \vdash \vdash [c] \cupdot [k]$, 
$l > 0$,
numbers 
$
0 \leq 
\alpha_{(S_1, \ldots, S_l), \delta}, \beta_{(S_1, \ldots, S_l), \delta}
\leq 1
$ 
and isometric embeddings 
$\cT_{(S_1, \ldots, S_l), \vecl, \delta}$ of 
$(\cH \otimes \C^2)^{\otimes k}$ into 
$A''_{[k]}$;

\item
For every $\vecx$, $\vecl$ and every 
$(S_1, \ldots, S_l) \vdash \vdash [c] \cupdot [k]$, $l > 0$,
quantum states 
$M_{(S_1, \ldots, S_l), \vecx, \vecl, \delta}$ 
and unit trace Hermitian operators 
$N_{(S_1, \ldots, S_l), \vecx, \vecl, \delta}$ 
in $A''_{[k]}$;
\end{itemize}
such that:
\begin{enumerate}
\item
\setcounter{propqtypicalellone}{\value{enumi}}
$
\ellone{(\Pi')_{\vecx,\vecl, \delta}^{A''_{[k]}}} \leq (2 |\cH|)^k;
$

\item
\setcounter{propqtypicalellinfty}{\value{enumi}}
$
\ellinfty{M_{(S_1, \ldots, S_l), \vecx, \vecl, \delta}^{A''_{[k]}}} 
\leq 
\frac{1}{|\cL|},
$
$
\beta_{(S_1, \ldots, S_l), \delta}
\ellinfty{N_{(S_1, \ldots, S_l), \vecx, \vecl, \delta}^{A''_{[k]}}} 
\leq 
\frac{3}{\sqrt{|\cL|}};
$

\item
\setcounter{propqtypicaldistance}{\value{enumi}}
$
\ellone{
(\rho'_{\vecx, \vecl, \delta})^{A''_{[k]}} - 
\rho_\vecx^{A_{[k]}} \otimes (\ketbra{0}^{\C^2})^{\otimes k}
} \leq
2^{\frac{k+c}{2} + 1} \delta;
$

\item
\setcounter{propqtypicalcompleteness}{\value{enumi}}
$
\Tr [
(\Pi'_{\vecx, \vecl, \delta})^{A''_{[k]}} 
(\rho^{A_{[k]}} \otimes (\ketbra{0}^{\C^2})^{\otimes k})
] \geq 
1 - 
\delta^{-2k} 2^{2^{ck+4} (k+1)^k} \epsilon_{\vecx};
$

\item
\setcounter{propqtypicalsplitting}{\value{enumi}}
Let $(S_1, \ldots, S_l) \vdash \vdash [c] \cupdot [k]$, $l > 0$.
Define $T := [k] \setminus (S_1 \cup \cdots \cup S_l)$.
Let $\sigma_\vecx^{A_T}$ be a state in $A_T$.
Let $S \subseteq [c] \cupdot [k]$, $S \cap [k] \neq \{\}$,
Let $\vecx_{[c] \cap S}$, $\vecl_S$ be computational basis vectors in 
$\cX^{\otimes ([c] \cap S}$, $\cL^{\otimes S}$.
Let $p_{[c] \setminus S}$ be a probability distribution on
$\cX^{\otimes ([c] \setminus S)}$.
In the following definition, let $\vecx'_{[c] \setminus S}$,
$\vecl'_{\bar{S}}$ range over all
computational basis vectors of $\cX^{\otimes ([c] \setminus S)}$,
$\cL^{\otimes \bar{S}}$.
Define a state in $A''_{S \cap [k]}$,
\begin{eqnarray*}
\lefteqn{
(\rho')_{\vecx_{S \cap [c]}, \vecl_{S}, \delta}^{A''_{S \cap [k]}}
} \\
& := &
|\cL|^{-|\bar{S}|} 
\sum_{\vecx'_{[c] \setminus S}, \vecl'_{\bar{S}}} 
p_{[c] \setminus S}(\vecx'_{[c] \setminus S})
\Tr_{A''_{\bar{S} \cap [k]}} [
(\rho')_{\vecx_{S \cap [c]} \vecx'_{[c] \setminus S}, 
         \vecl_{S} \vecl'_{\bar{S}}, \delta
        }^{A''_{[k]}}
].
\end{eqnarray*}
Analogously define 
\[
\rho_{\vecx_{S \cap [c]}}^{A_{S \cap [k]}} :=
\sum_{\vecx'_{[c] \setminus S}}
p_{[c] \setminus S}(\vecx'_{[c] \setminus S})
\Tr_{A_{\bar{S} \cap [k]}} [
\rho_{\vecx_{S \cap [c]} \vecx'_{[c] \setminus S}}^{A_{[k]}}
].
\]
Define 
\begin{eqnarray*}
\lefteqn{
(\rho')_{\vecx,\vecl,(S_1, \ldots, S_l),\delta}^{A''_{[k]}} 
} \\
& := &
(\rho'_{\vecx_{S_1 \cap [c]}, \vecl_{S_1},\delta})^{
A''_{S_1 \cap [k]}
} \\
&     &
{} \otimes \cdots \otimes \\
&     &
(\rho'_{\vecx_{S_l \cap [c])}, \vecl_{S_l},\delta})^{
A''_{S_l \cap [k]}
} \\
&     &
{} \otimes
(\sigma_\vecx^{A_T} \otimes (\ketbra{0}^{\C^2})^{\otimes |T|}), \\
\rho_{\vecx,(S_1, \ldots, S_l)}^{A_{[k]}} 
& := &
\rho_{\vecx_{S_1 \cap [c]}}^{A_{S_1 \cap [k]}} 
\otimes \cdots \otimes 
\rho_{\vecx_{S_l \cap [c]}}^{A_{S_l \cap [k]}} \otimes
\sigma_\vecx^{A_T}.
\end{eqnarray*}
Then,
\begin{eqnarray*}
\lefteqn{
(\rho')_{\vecx,\vecl,(S_1, \ldots, S_l),\delta}^{A''_{[k]}}
} \\
& = &
\alpha_{(S_1, \ldots, S_l), \delta}
(
\cT_{(S_1, \ldots S_l), \vecl, \delta}
(
\rho_{\vecx,(S_1, \ldots, S_l)}^{A_{[k]}} 
\otimes (\ketbra{0}^{\C^2})^{\otimes k}
)
)^{A''_{[k]}} \\
&   &
{}
+
\beta_{(S_1, \ldots, S_l), \delta}
N_{(S_1, \ldots S_l), \vecx, \vecl, \delta}^{A''_{[k]}} \\
&   &
{}
+ 
(
1 
- \alpha_{(S_1, \ldots, S_l), \delta} 
- \beta_{(S_1, \ldots, S_l), \delta}
)
M_{(S_1, \ldots S_l), \vecx, \vecl, \delta}^{A''_{[k]}};
\end{eqnarray*}
Moreover, the support of 
$M_{(S_1, \ldots S_l), \vecx, \vecl, \delta}^{A''_{[k]}}$ is orthogonal
to the support of the sum of the first two terms in the above
equation, and $\beta_{(S_1, \ldots, S_l), \delta} = 0$ if $c = 0$ or
$[c] \subseteq S_i$ for some $i \in [l]$;

\item
\setcounter{propqtypicalsoundness}{\value{enumi}}
$
\Tr [(\Pi')_{\vecx, \vecl,\delta}^{A''_{[k]}}
(
\cT_{(S_1, \ldots S_l), \vecl, \delta}
(
\rho_{\vecx,(S_1, \ldots, S_l)}^{A_{[k]}} 
\otimes (\ketbra{0}^{\C^2})^{\otimes k}
)
)^{A''_{[k]}}
] \leq 
2^{
   -D_H^{\epsilon_{\vecx,(S_1, \ldots, S_l)}}(
    \rho_\vecx^{A_{[k]}} \|
    \rho_{\vecx,(S_1, \ldots, S_l)}^{A_{[k]}} 
    )
  }.
$
\end{enumerate}
\end{proposition}
\begin{proof}
Proved in \ref{sec:proofcqtypical}
\end{proof}

\section{One-shot inner bound for the classical quantum MAC}
\label{sec:cqMAC}
In this section, we prove our general one-shot inner bound for the
cq-MAC for any number of senders. We illustrate our method by considering
two senders only, but the method works for any number of senders. 
We will use Corollary~\ref{cor:cqtypical} as our main technical
workhorse in order to obtain the inner bound. The difference between
the inner bound in this section and the inner bound in
Section~\ref{sec:cqMACBaby} is that we use a so-called `time-sharing'
random variable $U$ now.

Define a new sample space $\cU$ and put on it a probability distribution
$p(u)$. The resulting random variable $U$ plays the role, in the one-shot
setting, of the so-called
`time sharing random variable' familiar from classical iid network
information theory.
Fix conditional probability distributions $p(x | u)$, $p(y | u)$ 
on sets $\cX$, $\cY$.
Consider the classical-quantum state
\[
\rho^{UXYZ} :=
\sum_{u,x,y} p(u) p(x|u) p(y|u) 
\ketbra{u,x,y}^{UXY} \otimes
\rho^Z_{xy}.
\]
This state `controls' the encoding and decoding performance for the
channel $\chan$.

Consider a new alphabet, as well as Hilbert space, $\cL$ and define
the augmented systems 
$\cU' := \cU \otimes \cL$,
$\cX' := \cX \otimes \cL$,
$\cY' := \cY \otimes \cL$, 
and the extended system
$\cZ'$ defined in Corollary~\ref{cor:cqtypical} with $c = 3$.
Consider a `variant' cq-MAC 
$\chan'$ with input alphabets $\cU'$, $\cX'$, $\cY'$
and output Hilbert space $\cZ'$. On input $(u,l_u)$, $(x,l_x)$, 
$(y,l_y)$ to
$\chan'$, the output of $\chan'$ is the state
$
\rho^Z_{x,y} \otimes \ketbra{0}^{\C^2}.
$
The output of
$\chan'$ is taken to embed into the first summand in the definition of 
$\cZ'$ in Corollary~\ref{cor:cqtypical}.
The classical
quantum state `controlling' the encoding and decoding for $\chan'$ is
nothing but
$
\rho^{UXYZ} \otimes 
\ketbra{0}^{\C^2} \otimes
\frac{\one^{\cL^{\otimes 3}}}{|\cL|^3}.
$
The channel $\chan'$ can be trivally obtained from channel $\chan$.
The expected average decoding error for $\chan'$ is the same as 
the expected average decoding error for $\chan$ for the same
rate pair $(R_1, R_2)$. In fact, an encoding-decoding scheme for
$\chan'$ immediately gives an encoding-decoding scheme for $\chan$
with the same rate pair $(R_1, R_2)$ and the same decoding error.

Let $0 \leq \delta \leq 1$.
Next, consider a `perturbed' cq-MAC $\chan''$ with the same input alphabets
and output Hilbert space as in $\chan'$.
However, on input $(u,l_u)$, $(x,l_x)$, $(y,l_y)$ the output is the state
$(\rho')_{(u,l_u), (x,l_x), (y, l_y)}^{Z'}$
provided by setting
$c = 3$ in Corollay~\ref{cor:cqtypical}. 
Consider the classical quantum state
\begin{eqnarray*}
\lefteqn{(\rho')^{U'X'Y'Z'}} \\
& := &
|\cL|^{-3}
\sum_{u, x, y, l_u, l_x, l_y} p(u) p(x|u) p(y|u)
\ketbra{u,l_u}^{U'} \\
&     &
~~~~~~~~~~~~~~~~~~~~~~~~~
{} \otimes 
\ketbra{x,l_x}^{X'} \otimes 
\ketbra{y,l_y}^{Y'} \\
&     &
~~~~~~~~~~~~~~~~~~~~~~~~~~~~~
{} \otimes 
(\rho')^{Z'}_{u,x,y,l_u, l_x l_y,\delta}.
\end{eqnarray*}
This state `controls' the encoding and decoding performance for the
channel $\chan''$.
By Claim~\arabic{lemqtypicaldistance} of  Corollary~\ref{cor:cqtypical},
\[
\ellone{
(\rho')^{U'X'Y'Z'} -
\rho^{XYZ} \otimes 
\ketbra{0}^{(\C^2)^{\otimes 3}} \otimes
\frac{\one^{\cL^{\otimes 3}}}{|\cL|^3}
} \leq 
2^{3} \delta.
\] 
Thus, the expected average decoding error for $\chan''$ is at most
the expected average decoding error for $\chan'$, which is also the same
as the expected average decoding error for $\chan$,  plus
$2^{3} \delta$, for the same
rate pair $(R_1, R_2)$, and the same decoding strategy.

Consider the following randomised construction of a codebook $\cC$
for Alice and Bob for communication over the channel $\chan''$. 
Fix probability distributions $p(u)$, $p(x|u)$, $p(y|u)$ on sets 
$\cU$, $\cX$, $\cY$. Choose a sample $(u,l_u)$ according to the
product of the distribution $p(u)$ on $\cU$ and the uniform distribution
on computational basis vectors of $\cL$.
For all $m_1 \in [2^{R_1}]$, choose $(x,l_x)(m_1) \in \cX \times \cL$ 
independently according to the product of the distribution $p(x|u)$ on
$\cX$ and the uniform distribution on $\cL$. Similarly
for all $m_2 \in [2^{R_2}]$, choose $(y,l_y)(m_2) \in \cY \times \cL$ 
independently according to the product of the distribution $p(y|u)$ on
$\cY$ and the uniform distribution on $\cL$.

We now describe the decoding strategy that Charlie follows in order to
try and guess the message pair $(m_1, m_2)$ that was actually sent,
given the output of $\chan''$.
Let $(\Pi')_{u,x,y, l_u, l_x, l_y, \delta}^{Z'}$
be the POVM elements provided by Corollary~\ref{cor:cqtypical}. 
Charlie uses the {\em pretty good measurement} 
constructed from the POVM elements 
$
(\Pi')_{(u,l_u), (x,l_x)(m_1), (y,l_y)(m_2), \delta}^{Z'},
$
where $(m_1,m_2) \times [2^{R_1}] \times [2^{R_2}]$.
We now analyse the expectation, under the choice of a random codebook
$\cC$, of the error probability of Charlie's decoding algorithm.
Suppose the message pair $(m_1, m_2)$ is inputted to $\chan''$.
The output of $\chan''$ is the state
$(\rho')_{(u,l_u), (x,l_x)(m_1), (y, l_y)(m_2)}^{Z'}$.
Let $\Lambda_{\hat{m}_1, \hat{m}_2}^{Z'}$ be the POVM element corresponding
to decoded output $(\hat{m}_1, \hat{m}_2)$ arising from the
pretty good measurement. By the Hayashi-Nagaoka 
inequality~\cite{HayashiNagaoka},
the decoding error for $(m_1, m_2)$ is upper bounded by
\begin{eqnarray*}
\lefteqn{
\Tr [
(\one^{Z'} - \Lambda_{m_1, m_2}^{Z'})
(\rho')_{(u,l_u), (x,l_x)(m_1), (y, l_y)(m_2)}^{Z'} 
]
} \\
& \leq &
2 \Tr [
(\one^{Z'} - (\Pi')_{(u,l_u),(x,l_x)(m_1), (y, l_y)(m_2)}^{Z'}) \\
&      &
~~~~~~~~~~~~~
(\rho')_{(u,l_u),(x,l_x)(m_1), (y, l_y)(m_2)}^{Z'} 
] \\
&      &
{} +
4 \sum_{(\hat{m}_1, \hat{m}_2) \neq (m_1, m_2)} 
\Tr [
(\Pi')_{(u,l_u),(x,l_x)(\hat{m}_1), (y, l_y)(\hat{m}_2)}^{Z'} \\
&      &
~~~~~~~~~~~~~~~~~~~~~~~~~~~~~~~~~~~~~~~
(\rho')_{(u,l_u), (x,l_x)(m_1), (y, l_y)(m_2)}^{Z'} 
] \\
&   =  &
2 \Tr [
(\one^{Z'} - (\Pi')_{(u,l_u), (x,l_x)(m_1), (y, l_y)(m_2)}^{Z'}) \\
&      &
~~~~~~~~~~~~~
(\rho')_{(u,l_u), (x,l_x)(m_1), (y, l_y)(m_2)}^{Z'} 
] \\
&      &
{} +
4 \sum_{\hat{m}_1 \neq m_1} 
\Tr [
(\Pi')_{(u,l_u), (x,l_x)(\hat{m}_1), (y, l_y)(m_2)}^{Z'} \\
&      &
~~~~~~~~~~~~~~~~~~~~~~~~~~~
(\rho')_{(u,l_u), (x,l_x)(m_1), (y, l_y)(m_2)}^{Z'} 
] \\
&      &
{} +
4 \sum_{\hat{m}_2 \neq m_2} 
\Tr [
(\Pi')_{(u,l_u), (x,l_x)(m_1), (y, l_y)(\hat{m}_2)}^{Z'} \\
&      &
~~~~~~~~~~~~~~~~~~~~~~~~~~~
(\rho')_{(u,l_u), (x,l_x)(m_1), (y, l_y)(m_2)}^{Z'} 
] \\
&      &
{} +
4 \sum_{\hat{m}_1 \neq m_1, \hat{m}_2 \neq m_2} 
\Tr [
(\Pi')_{(u,l_u), (x,l_x)(\hat{m}_1), (y, l_y)(\hat{m}_2)}^{Z'} \\
&      &
~~~~~~~~~~~~~~~~~~~~~~~~~~~~~~~~~~~~
(\rho')_{(u,l_u), (x,l_x)(m_1), (y, l_y)(m_2)}^{Z'} 
].
\end{eqnarray*}
Define the POVM element
\begin{eqnarray*}
\lefteqn{(\Pi')^{U' X' Y' Z'}} \\
& := &
\sum_{u, x, y, l_u, l_x, l_y}
\ketbra{u, l_u}^{U'} \otimes
\ketbra{x, l_x}^{X'} \otimes
\ketbra{y, l_y}^{Y'} \\
&     &
~~~~~~~~~~~~~~~~~
{} \otimes
(\Pi')_{(u,l_u), (x,l_x), (y,l_y), \delta}^{Z'}
\end{eqnarray*}
as in Corollary~\ref{cor:cqtypical}.
From Corollary~\ref{cor:cqtypical}, recall the 
hypothesis testing conditional mutual information quantitites
$I^\epsilon_H(X : Z | U Y)$,
$I^\epsilon_H(Y : Z | U X)$,
$I^\epsilon_H(XY : Z | U)$
calculated with respect to the state $\rho^{UXYZ}$ corresponding to the
original channel $\chan$.
The expectation, over the choice of the random codebook $\cC$, of the
decoding error for $(m_1, m_2)$ is upper bounded by
\begin{eqnarray*}
\lefteqn{
\E_{\cC} [
\Tr [
(\one^{Z'} - \Lambda_{m_1, m_2}^{Z'})
(\rho')_{(u,l_u), (x,l_x)(m_1), (y, l_y)(m_2)}^{Z'} 
]
]
} \\
& \leq &
2 
\E_{\cC} [
\Tr [
(\one^{Z'} - (\Pi')_{(u, l_u), (x,l_x)(m_1), (y, l_y)(m_2)}^{Z'}) \\
&     &
~~~~~~~~~~~~~~~~~~~~
(\rho')_{(u, l_u), (x,l_x)(m_1), (y, l_y)(m_2)}^{Z'} 
]
] \\
&      &
{} +
4 \sum_{\hat{m}_1 \neq m_1} 
\E_{\cC} [
\Tr [
(\Pi')_{(u,l_u), (x,l_x)(\hat{m}_1), (y, l_y)(m_2)}^{Z'} \\
&     &
~~~~~~~~~~~~~~~~~~~~~~~~~~~~~~~~
(\rho')_{(u,l_u), (x,l_x)(m_1), (y, l_y)(m_2)}^{Z'} 
]
] \\
&      &
{} +
4 \sum_{\hat{m}_2 \neq m_2} 
\E_{\cC} [
\Tr [
(\Pi')_{(u,l_u), (x,l_x)(m_1), (y, l_y)(\hat{m}_2)}^{Z'} \\
&     &
~~~~~~~~~~~~~~~~~~~~~~~~~~~~~~~~
(\rho')_{(u,l_u), (x,l_x)(m_1), (y, l_y)(m_2)}^{Z'} 
]
] \\
&      &
{} +
4 \sum_{\hat{m}_1 \neq m_1, \hat{m}_2 \neq m_2} 
\E_{\cC} [
\Tr [
(\Pi')_{(u,l_u), (x,l_x)(\hat{m}_1), (y, l_y)(\hat{m}_2)}^{Z'} \\
&     &
~~~~~~~~~~~~~~~~~~~~~~~~~~~~~~~~~~~~~~~~~~
(\rho')_{(u,l_u), (x,l_x)(m_1), (y, l_y)(m_2)}^{Z'} 
]
] \\
&   =  &
2 |\cL|^{-3} \sum_{u, x, y, l_u, l_x, l_y} p(u) p(x|u) p(y|u) \\
&      &
~~~~~~~~~~~~~~~~~~~~~~~~~~~~
 \Tr [
(\one^{Z'} - (\Pi')_{(u,l_u), (x,l_x), (y, l_y)}^{Z'}) \\
&       &
~~~~~~~~~~~~~~~~~~~~~~~~~~~~~~~~~~~~~~~~
(\rho')_{(u,l_u), (x,l_x), (y, l_y)}^{Z'} 
] \\
&      &
{} +
4 (2^{R_1} - 1)
|\cL|^{-4} \sum_{u, l_u, x, l_x, x', l'_x, y, l_y}
p(u) p(x|u) p(x'|u) p(y|u) \\
&     &
~~~~~~~~~~
\Tr [
(\Pi')_{(u,l_u), (x',l'_x), (y, l_y)}^{Z'})
(\rho')_{(u,l_u), (x,l_x), (y, l_y)}^{Z'} 
] \\
&  &
{} +
4 (2^{R_2} - 1)
|\cL|^{-4} \sum_{u, l_u, x, l_x, x', l'_x, y, l_y}
p(u) p(x|u) p(y|u) p(y'|u) \\
&     &
~~~~~~~~~~
\Tr [
(\Pi')_{(u,l_u), (x,l_x), (y', l'_y)}^{Z'})
(\rho')_{(u,l_u), (x,l_x), (y, l_y)}^{Z'} 
] \\
&      &
{} +
4 (2^{R_1} - 1) (2^{R_2} - 1)
|\cL|^{-5} \\
&     &
~~~~~~
\sum_{u, l_u, x, l_x, x', l'_x, y, l_y, y', l'_y}
p(u) p(x|u) p(x'|u) p(y|u) p(y'|u) \\
&     &
~~~~~~~~~~~~~~~~~~~~~~~
\Tr [
(\Pi')_{(u,l_u), (x',l'_x), (y', l'_y)}^{Z'})
(\rho')_{(u,l_u), (x,l_x), (y, l_y)}^{Z'} 
] \\
&   =  &
2 \Tr [
(\one^{U' X' Y' Z'} - (\Pi')^{U' X' Y' Z'})
(\rho')^{U' X' Y' Z'}
] \\
&      &
{} +
4 (2^{R_1} - 1)
\Tr [
(\Pi')^{U' X' Y' Z'}
(\rho')^{U' X' Y' Z'}_{(\{U'\}, \{X'\}, \{Y'\})}
] \\
&      &
{} +
4 (2^{R_2} - 1)
\Tr [
(\Pi')^{U' X' Y' Z'}
(\rho')^{U' X' Y' Z'}_{(\{U'\}, \{Y'\}, \{X'\})}
] \\
&      &
{} +
4 (2^{R_1} - 1) (2^{R_2} - 1)
\Tr [
(\Pi')^{U' X' Y' Z'}
(\rho')^{U' X' Y' Z'}_{(\{U'\}, \{X',Y'\}, \{\})}
] \\
& \leq &
2 (\delta^{-2} 2^{134} \epsilon + 2^{3} \delta) +
2^{R_1 + 2 - I^\epsilon_H(X : Y Z | U )_\rho} +
2^{R_2 + 2 - I^\epsilon_H(Y : X Z | U )_\rho} \\
&      &
{} +
2^{R_1 + R_2 + 2 -  I^\epsilon_H(XY : Z | U)_\rho},
\end{eqnarray*}
where we used Claims~\arabic{lemqtypicalcompleteness1}, 
\arabic{lemqtypicalsoundness1} of Corollary~\ref{cor:cqtypical} in the
last inequality above.

Taking $\delta = \epsilon^{1/3}$ and
choosing a rate pair $(R_1, R_2)$ satisfying
\begin{eqnarray*}
R_1 
& \leq & 
I^\epsilon_H(X : Y Z | U)_\rho - 2 - \log \frac{1}{\epsilon}, \\
R_2 
& \leq & 
I^\epsilon_H(Y : X Z | U)_\rho - 2 - \log \frac{1}{\epsilon}, \\
R_1 + R_2 
& \leq & 
I^\epsilon_H(XY : Z | U)_\rho - 2 - \log \frac{1}{\epsilon}, 
\end{eqnarray*} 
ensures that the expected average decoding error for channel
$\chan''$ is at most $2^{135} \epsilon^{1/3}$. This implies that
the expected average decoding error for the original channel $\chan$
is at most $2^{135} \epsilon^{1/3}$. Thus there exists a codebook
$\cC$ with average decoding error for $\chan$ at most 
$2^{135} \epsilon^{1/3}$. By a standard technique of taking maps from
classical symbols to arbitrary quantum states, we can then prove the 
following theorem.
\begin{theorem}
\label{thm:cqMAC}
Let $\chan: X' Y' \rightarrow Z$ be a quantum multiple access channel.
Let $\cU$, $\cX$, $\cY$ be three new sample spaces.
For every element $x \in \cX$, let 
$\sigma_{x}^{X'}$ be a quantum state in the input Hilbert space 
$X'$ of $\chan$.
Similarly, for every element $y \in \cY$, let 
$\sigma_{y}^{Y'}$ be a quantum state in the input Hilbert space 
$Y'$ of $\chan$.
Let $p(u) p(x|u) p(y|u)$ be a probability distribution on 
$\cU \times \cX \times \cY$.
Consider the classical quantum state
\begin{eqnarray*}
\rho^{U X Y Z}
& := &
\sum_{u, x, y}
p(u) p(x|u) p(y|u)
\ketbra{u, x, y}^{U X Y} \\
&     &
~~~~~~~~~~~
{} \otimes
(\chan(\sigma_{x}^{X'} \otimes \sigma_{y}^{Y'}))^{Z}.
\end{eqnarray*}
Let $R_1$, $R_2$, $\epsilon$, 
be such that
\begin{eqnarray*}
R_1 
& \leq & 
I^\epsilon_H(X : Y Z | U)_\rho - 2 - \log \frac{1}{\epsilon}, \\
R_2 
& \leq & 
I^\epsilon_H(Y : X Z | U)_\rho - 2 - \log \frac{1}{\epsilon}, \\
R_1 + R_2 
& \leq & 
I^\epsilon_H(XY : Z | U)_\rho - 2 - \log \frac{1}{\epsilon}. 
\end{eqnarray*} 
Then there exists an $(R_1, R_2, 2^{135} \epsilon^{1/3})$-quantum 
MAC code for sending classical information through $\chan$.
\end{theorem}

A similar bound can be proved for sending classical information through
an entanglement assisted q-MAC by using the position based coding
technique of Anshu, Jain and Warsi~\cite{anshu:broadcast}. 
Earlier, Qi, Wang and Wilde~\cite{qi:simultaneous} had constructed
a one-shot simultaneous decoder for this problem which also used position
based coding, but their inner bound is suboptimal when reduced to the
asymptotic iid setting. In contrast, our one-shot inner bound reduces 
to the best known inner bound in the asymptotic iid setting which was
proved by Hsieh, Devetak and Winter~\cite{hsieh:qmac} 
using successive cancellation arguments.
\begin{theorem}
Let $\chan: X' Y' \rightarrow Z$ be a quantum multiple access channel.
Let $\cU$, $\cX$, $\cY$ be three new Hilbert spaces.
Let $\psi^{U X Y X' Y'}$ be a classical quantum state with the
following structure:
\[
\psi^{U X Y X' Y'} =
\sum_u p(u) \ketbra{u}^U \otimes
\psi^{X X'}_u \otimes \psi^{Y Y'}_u.
\]
Consider the classical quantum state
\[
\rho^{U X Y Z} := 
\sum_u p(u) \ketbra{u}^U \otimes
(\chan^{X' Y' \rightarrow Z}(\psi^{X X'}_u \otimes \psi^{Y Y'}_u))^{XYZ}.
\]
Let $R_1$, $R_2$, $\epsilon$, 
be such that
\begin{eqnarray*}
R_1 
& \leq & 
I^\epsilon_H(X : Y Z | U)_\rho - 2 - \log \frac{1}{\epsilon}, \\
R_2 
& \leq & 
I^\epsilon_H(Y : X Z | U)_\rho - 2 - \log \frac{1}{\epsilon}, \\
R_1 + R_2 
& \leq & 
I^\epsilon_H(XY : Z | U)_\rho - 2 - \log \frac{1}{\epsilon}. 
\end{eqnarray*} 
Then there exists an entanglement assisted 
$(R_1, R_2, 2^{135} \epsilon^{1/3})$-quantum 
MAC code for sending classical information through $\chan$.
\end{theorem}

\section{Conclusion and open problems}
\label{sec:conclusions}
In this work, we have proved a one-shot classical quantum joint 
typicality lemma that not only extends the iid classical conditional 
joint typicality lemma to the one-shot and quantum settings, but also
extends intersection and union arguments that are ubiquitous in classical
network information theory to the quantum setting. We introducted
two novel tools in the process of proving our joint
typicality lemma viz. {\em tilting}, and {\em smoothing and augmentation},
which should be useful
elsewhere. Our lemma allows us to transport many packing arguments  
arising in proofs of inner bounds from the classical to the quantum 
setting. We illustrated this by constructing a simultaneous decoder
for the one-shot classical quantum MAC. More of such applications
can be found in the companion paper~\cite{sen:simultaneous}.

However, the statement of our quantum joint typicality lemma is not
as strong as the classical statement. It can only handle negative
hypothesis states that are a tensor product of marginals and at most
one arbitrary quantum state. Proving a 
quantum lemma that can handle arbitrary negative hypothesis states,
as in the classical setting, is an important open problem.

Though the main advance of our quantum joint typicality lemma is 
a robust notion of intersection of projectors, there are some technical
issues which prevent us from proving a chain rule for hypothesis
testing mutual information which an ideal notion of intersection of
projectors would have allowed one to do. The main issue here is that
the smoothing and augmentation procedure ensures that tensor products
of marginals corresponding to different partitions of a multipartite
quantum state essentially lie in `well-separated' subspaces. This feature
which was crucial for proving our joint typicality lemma turns out to
be a bottleneck for proving a chain rule. A similar bottleneck arises
if one were to attempt to prove the simultaneous smoothing conjecture
using our techniques.

Another drawback of our quantum lemma is that it cannot handle covering
arguments and their  unions and intersections that often arise in
source coding applications. This is unlike the classical case. Again
this issue seems to be related to the simultaneous smoothing conjecture
and the bottleneck mentioned above. Addressing
this deficiency is another important open problem.

\section*{Acknowledgements}
I profusely thank Mark Wilde for pointing out the need for a 
simultaneous decoder for the quantum multiple access channel
many years back while I was visitng McGill, and 
underscoring its importance for quantum network information theory. 
He also provided pointers to several important references.
Patrick Hayden, David Ding and
Hrant Gharibyan took great pains in reading a preliminary draft of this 
paper and provided useful feedback. I thank Rahul Jain for 
encouraging me to write up this result in the face of a long delay.
I thank an anonymous referee for suggesting to give a full self
contained solution
for the pathological example of failure of `union by span' in the
introduction, and more detailed intuition behind smoothing and
augmentation.
I am grateful to three anonymous referees for suggesting to 
give a self-contained
proof of the `baby case' of `intersection' of two projectors / two 
sender classical quantum multiple access channel.
I thank Naqueeb Warsi for pointing me to his paper on 
the entanglement
assisted compound quantum channel and suggesting a possible 
application of tilted span of projectors to that problem. Finally, I thank
an anonymous referee who pointed out a possible application of our
tilted span to the `quantum OR bound' of Harrow, Lin and Montanaro
and asked  whether our notion of intersection of projectors can lead
to the proof of a chain rule for one shot mutual information quantities.

\bibliography{oneshot}

\appendix

\section{Proof of Proposition~\ref{prop:cqtypical}}
\label{sec:proofcqtypical}
We prove Proposition~\ref{prop:cqtypical} in this section. We first show
how to construct the objects whose existence is promised by the
proposition. We
then proceed to prove its claims.
All the while, we use the notation of the proposition: 
$\vecx$, $\vecl$ will denote  computational basis vectors of 
$\cX^{\otimes [c]}$,
$\cL^{\otimes ([c] \cupdot [k])}$, 
$0 \leq \delta \leq 1$, and
$0 \leq \epsilon_{\vecx, (S_1, \ldots, S_l)} \leq 1$.

In \ref{subsec:rhoprime}, we show how to construct the perturbed
state $(\rho')_{\vecx,\vecl,\delta}^{A''_{[k]}}$ by tilting
the original state $\rho_{\vecx}^{A_{[k]}}$ appropriately.
In \ref{subsec:Piprime}, we construct the `intersection' projector
$(\Pi')_{\vecx,\vecl,\delta}^{A''_{[k]}}$ by using the technique of
matrix-tilted span.
In \ref{subsec:alphabetaT} we define two complex numbers and a tilting map,
and in \ref{subsec:MN} we define two operators
required for the statement of Proposition~\ref{prop:cqtypical} using
the combinatorial and algebraic framework of matrix tilting.
In \ref{subsec:Claim1}, we prove Claim~\arabic{propqtypicalellone}
of Proposition~\ref{prop:cqtypical} using simple operator inequalities.
In \ref{subsec:Claim2}, we prove Claim~\arabic{propqtypicalellinfty}
using the geometry behind the smoothing and augmentation technique 
described earlier.
In \ref{subsec:Claim3} we prove Claim~\arabic{propqtypicaldistance},
and in \ref{subsec:Claim4} we prove Claim~\arabic{propqtypicalcompleteness}
using combinatorial properties of the tilting operator.
In \ref{subsec:Claim5}, we observe that 
Claim~\arabic{propqtypicalsplitting} has already been proven in 
\ref{subsec:MN} earlier.
In \ref{subsec:Claim6}, we prove Claim~\arabic{propqtypicalsoundness}
by invoking the results of \ref{subsec:Piprime} regarding the geometry
of the `intersection' projector.

\subsection{Constructing $(\rho')_{\vecx,\vecl,\delta}^{A''_{[k]}}$}
\label{subsec:rhoprime}
Let $S \subseteq [c] \cupdot [k]$, $S \cap [k] \neq \{\}$. Define 
$\bar{S} := ([c] \cupdot [k]) \setminus S$.
Let $\vecl_S$ be a computational basis vector of $\cL^{\otimes S}$.
We define an isometric embedding $\cT_{S, \vecl_S}$ of 
$(\cH \otimes \C^2)^{\otimes (S \cap [k])}$ 
into 
$
((\cH \otimes \C^2) \otimes \cL^{\otimes |S|})^{\otimes (S \cap [k])} 
$
as follows: 
\[
(\cT_{S, \vecl_S}(\vech_{S \cap [k]}))_s :=
(\vech_s, \vecl_S)
\]
where $s \in S \cap [k]$, $\vech_{S \cap [k]}$ is a computational basis
vector of $(\cH \otimes \C^2)^{\otimes (S \cap [k])}$, $\vech_s$,
$(\cT_{S, \vecl_S}(\vech_{S \cap [k]}))_s$ are the entries 
in the $s$th coordinate of $\vech_{S \cap [k]}$,
$\cT_{S, \vecl_S}(\vech_{S \cap [k]})$. 
Observe that $\cT_{S, \vecl_S}$ maps computational
basis vectors of $(\cH \otimes \C^2)^{\otimes (S \cap [k])}$ to 
computational basis vectors of 
$
((\cH \otimes \C^2) \otimes \cL^{\otimes |S|})^{\otimes (S \cap [k])}.
$
Thus, $\cT_{S, \vecl_S}$ is an isometric embedding of 
$(\cH \otimes \C^2)^{\otimes (S \cap [k])}$ into 
$
((\cH \otimes \C^2) \otimes \cL^{\otimes |S|})^{\otimes (S \cap [k])}
$
which further embeds into $A''_{S \cap [k]}$ in the natural fashion.
Let $\one^{(\cH \otimes \C^2)^{\otimes (S \cap [k])}}$ denote the identity
embedding of $(\cH \otimes \C^2)^{\otimes (S \cap [k])}$ into 
$A''_{(S \cap [k])}$.
Observe that the range spaces of 
$
\cT_{S, \vecl_S} \otimes \one^{(\cH \otimes \C^2)^{
\otimes (\bar{S} \cap [k])
}},
$
as $S$ ranges over subsets of $[c] \cupdot [k]$ intersecting $[k]$
non-trivially,
embed orthogonally into
$A''_{[k]}$ under the natural embedding. Also, for a fixed  
$S \subseteq [c] \cupdot [k]$, $S \cap [k] \neq \{\}$ the range spaces of 
$\cT_{S, \vecl_S}$,
as $\vecl_S$ ranges over computational basis vectors of $\cL^{\otimes S}$,
embed orthogonally into $A''_{S \cap [k]}$ under the natural embedding.

We now define an isometric embedding 
$\cT_{S, \vecl_S, \delta}$ 
of $(\cH \otimes \C^2)^{\otimes (S \cap [k])}$ into 
$A''_{S \cap [k]}$ as follows:
\begin{equation}
\label{eq:isometricembedding}
\begin{array}{rcl}
\lefteqn{\cT_{S, \vecl_S, \delta}} \\
& := &
\frac{1}{\sqrt{N(S, \delta)}} 
\left(
\one^{(\cH \otimes \C^2)^{\otimes (S \cap [k])}} 
\right. \\
&    &
~~~~~~~~~~~~~
{} +
\sum_{(S_1, \ldots, S_l) \vdash \vdash S, l > 0}
\delta^{l} \, 
\cT_{S_1, \vecl_{S_1}} \otimes \cdots \otimes 
\cT_{S_l, \vecl_{S_l}} \\
&   &
~~~~~~~~~~~~~~~~~~~~~~~~~~~~~~~~~~~~~~~~
\left.
{} \otimes
\one^{(\cH \otimes \C^2)^{
\otimes ((S \cap [k]) \setminus (S_1 \cup \cdots \cup S_l))
}}
\right),
\end{array}
\end{equation}
where the normalisation factor $N(S, \delta)$ is
put in to make the embedding preserve length of vectors. Observe
that $N(S, \delta)$ is independent of the vector in 
$(\cH \otimes \C^2)^{\otimes (S \cap [k])}$ on which 
$\cT_{S, \vecl_S, \delta}$ acts since the terms in the 
above summation have orthogonal range spaces. It can be
estimated as follows: 
\begin{equation}
\label{eq:increasedlength1}
\begin{array}{rcl}
N(S, \delta) 
& := &
\displaystyle
1 + \sum_{(S_1, \ldots, S_l) \vdash \vdash S, l > 0} \delta^{2l} \\
& \leq &
\displaystyle
1 + 
\sum_{l=1}^{|S \cap [k]|} \delta^{2l} 
                 \frac{2^{l|S \cap [c]|} (l+1)^{|S \cap [k]|}}{l!} \\
& \leq &
\displaystyle
1 + \sum_{l=1}^{|S \cap [k]|} \frac{\delta^{2l} 2^{l|S|}}{l!} 
\;<\;
e^{\delta^2 2^{|S|}}.
\end{array}
\end{equation}
It is clear that
$N(S,\delta) \geq 1$ and 
$N(T_1, \delta) \cdots N(T_m, \delta) \leq N(S, \delta)$ if
$(T_1, \ldots, T_m) \vdash \vdash S$. 

The map $\cT_{S, \vecl_S, \delta}$ extends to subspaces
and density matrices in  $(\cH \otimes \C^2)^{\otimes (S \cap [k])}$ in 
the natural way. We now define the quantum state
\begin{equation}
\label{eq:rhoprimedefinition}
(\rho')_{\vecx, \vecl, \delta}^{A''_{[k]}} :=
\cT_{[c] \cupdot [k], \vecl, \delta}
\left(
\rho_\vecx^{A_{[k]}} \otimes (\ketbra{0}^{\C^2})^{\otimes k}
\right)
\end{equation}
in $A''_1 \cdots A''_k$,
where 
$
\rho_\vecx^{A_{[k]}} \otimes (\ketbra{0}^{\C^2})^{\otimes k}
$
can be thought of as a density matrix in 
$(\cH \otimes \C^2)^{\otimes [k]}$
in the natural fashion.

\subsection{Constructing $(\Pi')_{\vecx,\vecl,\delta}^{A''_{[k]}}$}
\label{subsec:Piprime}
For each $\vecx$, $(S_1, \ldots, S_l) \vdash \vdash [c] \cupdot [k]$, 
$l > 0$, 
let $\Pi''_{\vecx, (S_1, \ldots, S_l)}$ be the POVM element in 
$A_1 \cdots A_k$
with the property that
\begin{eqnarray*}
\Tr [(\Pi'')_{\vecx, (S_1, \ldots, S_l)}^{A_{[k]}} \rho_\vecx^{A_{[k]}}]
& \geq & 
1 - \epsilon_{\vecx, (S_1, \ldots, S_l)}, \\
\Tr [(\Pi'')_{\vecx, (S_1, \ldots, S_l)}^{A_{[k]}}
     \rho_{\vecx, (S_1, \ldots, S_l)}^{A_{[k]}}
    ] 
& \leq &
 2^{
    -D_H^{\epsilon_{\vecx, (S_1, \ldots, S_l)}}
     (
      \rho_\vecx^{A_{[k]}} \| 
      \rho_{\vecx, (S_1, \ldots, S_l)}^{A_{[k]}}
     )
   }. 
\end{eqnarray*}
By Fact~\ref{fact:gelfandnaimark}, there exists an orthogonal projection
$\Pi_{\vecx, (S_1, \ldots, S_l)}$ in 
$
(A_1 \cdots A_k) \otimes (\C^2)^{\otimes k} \cong 
(\cH \otimes \C^2)^{\otimes [k]}
$
such that 
\begin{equation}
\label{eq:optimalPi}
\begin{array}{rcl}
\lefteqn{
\Tr [\Pi^{(\cH \otimes \C^2)^{\otimes [k]}}_{\vecx, (S_1,\ldots, S_l)} 
     (\rho_\vecx^{A_{[k]}} \otimes (\ketbra{0}^{\C^2})^{\otimes k})
    ] 
} \\
& \geq &
1 - \epsilon_{\vecx, (S_1, \ldots, S_l)}, \\
\lefteqn{
\Tr [\Pi^{(\cH \otimes \C^2)^{\otimes [k]}}_{\vecx, (S_1,\ldots, S_l)}
     (\rho_{\vecx, (S_1, \ldots, S_l)}^{A_{[k]}}
      \otimes (\ketbra{0}^{\C^2})^{\otimes k}
     )
    ] 
} \\
& \leq &
 2^{
    -D_H^{\epsilon_{\vecx, (S_1, \ldots, S_l)}}
     (
      \rho_\vecx^{A_{[k]}} \| 
      \rho_{\vecx, (S_1, \ldots, S_l)}^{A_{[k]}}
     )
   }. 
\end{array}
\end{equation}

Let $Y_{\vecx, (S_1, \ldots, S_l)}$ denote the orthogonal complement 
of the support of $\Pi_{\vecx, (S_1, \ldots, S_l)}$
in $(\cH \otimes \C^2)^{\otimes [k]}$.
Identify $(\cH \otimes \C^2)^{\otimes [k]}$ with the Hilbert space 
$\cH$ of Proposition~\ref{prop:generalisedtiltedspan}. 
Arrange all the non-empty pseudosubpartitions 
$(T_1, \ldots, T_m) \vdash \vdash [c] \cupdot [k]$, $m > 0$
into a linear order extending the refinement partial order $\preceq$.
Define the tilting matrix $A$, whose rows and columns are indexed
by non-empty pseudosubpartitions of $[c] \cupdot [k]$ that intersect
$[k]$ non-trivially, as follows:
\begin{equation}
\label{eq:Adefinition}
\begin{array}{rcl}
\lefteqn{
A_{(S_1, \ldots, S_l), (T_1, \ldots, T_m)} 
} \\
& := &
\begin{array}{l l}
\frac{\delta^{2l}}{N(T_1, \delta) \cdots N(T_m, \delta)} &
\mbox{\ if $(S_1, \ldots, S_l) \preceq (T_1, \ldots, T_m)$}, \\
0 & 
\mbox{\ otherwise}.
\end{array}
\end{array}
\end{equation}
Observe that $A$ is upper triangular,
diagonal dominated and substochastic. The diagonal dominated property
of $A$ follows from the fact that 
$
N(T_1, \delta) \cdots N(T_m, \delta) \leq 
N(W_1, \delta) \cdots N(W_n, \delta)
$
if $(T_1, \ldots, T_m) \preceq (W_1, \ldots, W_n)$.
The reason $A$ is substochastic
in general, and not stochastic, is because the empty pseudosubpartition 
is not included amongst the rows and columns of $A$.
More precisely,
\[
\sum_{
(S_1, \ldots, S_l): (S_1, \ldots, S_l) \preceq (T_1, \ldots, T_m), l > 0
} \delta^{2l}
=
N(T_1, \delta) \cdots N(T_m, \delta) - 1.
\]

For $(S_1, \ldots, S_l) \vdash \vdash [c] \cupdot [k]$, $l > 0$ define 
an isometric 
embedding $\cT_{(S_1, \ldots, S_l), \vecl, \delta}$
of $(\cH \otimes \C^2)^{\otimes [k]}$ into $A''_{[k]}$ as follows:
\begin{equation}
\label{eq:isometricembedding1}
\cT_{(S_1, \ldots, S_l), \vecl, \delta} :=
\cT_{S_1, \vecl_{S_1}, \delta} \otimes \cdots \otimes
\cT_{S_l, \vecl_{S_l}, \delta} \otimes
\one^{(\cH \otimes \C^2)^{\otimes [k] \setminus (S_1 \cup \cdots S_l)}}.
\end{equation}
Observe that the $A$-tilt of 
$Y_{\vecx, (S_1, \ldots, S_l)}$ along the 
$(S_1, \ldots, S_l)$th direction is nothing but the action of
the isometric embedding $\cT_{(S_1, \ldots, S_l), \vecl, \delta}$:
\begin{equation}
\label{eq:Yprimedefinition}
\begin{array}{rcl}
\lefteqn{
Y'_{\vecx, (S_1, \ldots, S_l), \vecl, \delta} 
} \\
& := &
(\cT_\vecl)_{(S_1, \ldots, S_l), A}(Y_{\vecx, (S_1, \ldots, S_l)}) =
\cT_{(S_1, \ldots, S_l), \vecl, \delta}(Y_{\vecx, (S_1, \ldots, S_l)}).
\end{array}
\end{equation}
Define the $A$-tilted span
\begin{equation}
\label{eq:Ydefinition}
Y'_{\vecx, \vecl, \delta} :=
(Y'_{\vecx, \vecl})_A =
\bigplus_{(S_1, \ldots, S_l) \vdash \vdash [c] \cupdot [k], l > 0}
Y'_{(S_1, \ldots, S_l), \vecl, \delta}.
\end{equation}
View $Y'_{\vecx, \vecl, \delta}$ as a subspace of $A''_1 \ldots A''_k$.
Let $(\Pi')^{A''_{[k]}}_{Y'_{\vecx, \vecl, \delta}}$ denote the orthogonal 
projection in 
$A''_1 \ldots A''_k$ onto $Y'_{\vecx, \vecl, \delta}$.

Let $(\Pi')^{A''_{[k]}}_{(\cH \otimes \C^2)^{\otimes [k]}}$ denote the
orthgonal projection in $A''_1 \ldots A''_k$ onto 
$(\cH \otimes \C^2)^{\otimes [k]}$. 
We finally define the POVM element 
$(\Pi')^{A''_{[k]}}_{\vecx,\vecl,\delta}$ in 
$A''_1 \ldots A''_k$ to be
\begin{equation}
\label{eq:Piprimedefinition}
(\Pi')^{A''_{[k]}}_{\vecx,\vecl,\delta} := 
(
\one^{A''_{[k]}} -
(\Pi')^{A''_{[k]}}_{Y'_{\vecx, \vecl, \delta}}
)
((\Pi')^{A''_{[k]}}_{(\cH \otimes \C^2)^{\otimes [k]}})
(
\one^{A''_{[k]}} -
(\Pi')^{A''_{[k]}}_{Y'_{\vecx, \vecl, \delta}}
)
\end{equation}

\subsection{Defining 
$\alpha_{(S_1, \ldots, S_l), \delta}$, 
$\beta_{(S_1, \ldots, S_l), \delta}$, 
$\cT_{(S_1, \ldots, S_l), \vecl, \delta}$
}
\label{subsec:alphabetaT}
Recall that the isometric embedding
$
\cT_{(S_1, \ldots, S_l), \vecl, \delta}
$
of $(\cH \otimes \C^2)^{\otimes [k]}$ into $A''_{[k]}$ has already been
defined in Equation~\ref{eq:isometricembedding1} above.
Also define
\begin{eqnarray*}
\lefteqn{
\alpha_{(S_1, \ldots, S_l), \delta} 
} \\
& := &
\frac{N(S_1, \delta) N(\bar{S_1} \cup [c], \delta)}
     {N([c] \cupdot [k], \delta)} \cdots
\frac{N(S_l, \delta) N(\bar{S_l} \cup [c], \delta)}
     {N([c] \cupdot [k], \delta)}, \\
\lefteqn{
\alpha_{(S_1, \ldots, S_l), \delta} +
\beta_{(S_1, \ldots, S_l), \delta} 
} \\
& := &
\frac{N(S_1 \cup [c], \delta) N(\bar{S_1} \cup [c], \delta)}
     {N([c] \cupdot [k], \delta)} \cdots {} \\
&    &
~~~~~~~
\frac{N(S_l \cup [c], \delta) N(\bar{S_l} \cup [c], \delta)}
     {N([c] \cupdot [k], \delta)}.
\end{eqnarray*}
Since 
\[
(S, \bar{S} \cup [c]) \preceq (S \cup [c], \bar{S} \cup [c]) 
\vdash \vdash [c] \cupdot [k]
\]
for any $S \subseteq [c] \cupdot [k]$,
$\{\} \neq S \cap [k] \subset [k]$, 
\[
N(S, \delta) N(\bar{S} \cup [c], \delta) \leq 
N(S \cup [c], \delta) N(\bar{S} \cup [c], \delta) \leq 
N([c] \cupdot [k], \delta).
\]
It follows that
$
0 \leq 
\alpha_{(S_1, \ldots, S_l), \delta}, \beta_{(S_1, \ldots, S_l), \delta} 
\leq 1.
$

\subsection{Constructing 
$M_{(S_1, \ldots, S_l), \vecl, \delta}^{A''_{[k]}}$,
$N_{(S_1, \ldots, S_l), \vecl, \delta}^{A''_{[k]}}$
}
\label{subsec:MN}
Let $S \subseteq [c] \cupdot [k]$, $\{\} \neq S \cap [k] \subset [k]$.
Define $\bar{S} := ([c] \cupdot [k]) \setminus S$.
For a subset $T \subseteq [c] \cupdot [k]$, $T \cap [k] \neq \{\}$ 
we say that $T$ {\em crosses $S$} if
$T \cap S \cap [k] \neq \{\}$ and $T \cap \bar{S} \cap [k] \neq \{\}$. 
We define a pseudosubpartition 
$(T_1, \ldots, T_l) \vdash \vdash [c] \cupdot [k]$ to {\em cross $S$}
if there exists an $i \in [l]$ such that $T_i$ crosses $S$.
We use the
notation $(T_1, \ldots, T_l) \models_{\times} S$ to denote that
pseudosubpartition $(T_1, \ldots, T_l)$ crosses $S$.
We define the {\em $S$-signature} of a pseudosubpartition 
$(T_1, \ldots, T_l)$ to be the pseudosubpartition 
$(T_1, \ldots, T_{l'})$ where $l' \leq l$ and $T_1, \ldots, T_{l'}$ 
are the subsets that actually cross $S$. We shall denote
pseudosubpartitions $(T_1, \ldots, T_{l}) \vdash \vdash [c] \cupdot [k]$ 
where for all $i \in [l]$, $T_i$ crosses $S$, by the notation
$(T_1, \ldots, T_{l}) \models_{\times \times} S$.
If pseudosubpartition $(T_1, \ldots, T_l)$ has $S$-signature
$(T_1, \ldots, T_{l'})$, we shall denote it by 
$(T_1, \ldots, T_l) \rightthreetimes_S (T_1, \ldots, T_{l'})$.
Observe that $(T_1, \ldots, T_{l'}) \models_{\times \times} S$.

From Equation \ref{eq:isometricembedding}, we observe that
\begin{equation}
\label{eq:tensordecomposition}
\begin{array}{rcl}
\lefteqn{
\cT_{[c] \cupdot [k], \vecl, \delta} 
} \\
& = & 
\displaystyle
\sqrt{
\frac{N(S \cup [c], \delta) N (\bar{S} \cup [c], \delta)}
     {N([c] \cupdot [k], \delta)}
}
\cT_{S \cup [c], \vecl_{S \cup [c]}, \delta} \otimes 
\cT_{\bar{S} \cup [c], \vecl_{\bar{S} \cup [c]}, \delta} \\
&   &
\displaystyle
{} +
\frac{1}{\sqrt{N([c] \cupdot [k], \delta)}} 
\sum_{
(T_1, \ldots, T_l) \vdash \vdash [c] \cupdot [k],
(T_1, \ldots, T_l) \models_{\times} S
} \\
&  &
~~~~~~~~
\delta^{l} \, 
(
\cT_{T_1, \vecl_{T_1}} \otimes \cdots \otimes  
\cT_{T_l, \vecl_{T_l}} \otimes
\one^{
(\cH \otimes \C^2)^{\otimes ([k] \setminus (T_1 \cup \cdots \cup T_l))}
}
).
\end{array}
\end{equation}
We will denote the second term of the summation above by
$\cT_{\models_{\times} S, \vecl, \delta}$. Observe that 
$\cT_{\models_{\times} S, \vecl, \delta}$ is a scaled 
isometric embedding of
$(\cH \otimes \C^2)^{\otimes [k]}$ into $A''_{[k]}$.

Let $\ket{h}$ be a unit length vector in 
$(\cH \otimes \C^2)^{\otimes [k]}$. Let the Schmidt
decomposition of $\ket{h}$ with respect to 
$(S \cap [k], \bar{S} \cap [k])$ be
\[
\ket{h}^{(\cH \otimes \C^2)^{\otimes [k]}} =
\sum_i \sqrt{p_i} 
\ket{\alpha_i}^{(\cH \otimes \C^2)^{\otimes (S \cap [k])}} \otimes
\ket{\beta_i}^{(\cH \otimes \C^2)^{\otimes (\bar{S} \cap [k])}},
\]
where $\sum_i p_i = 1$.
Let the Schmidt 
decomposition of $\cT_{\models_{\times} S, \vecl, \delta}(\ket{h})$ 
with respect to $(S \cap [k], \bar{S} \cap [k])$ be
\begin{equation}
\label{eq:Mdecomposition1}
\begin{array}{rcl}
\lefteqn{
(\cT_{\models_{\times} S, \vecl, \delta}(\ket{h}))^{A''_{[k]}}
} \\
& = &
\sum_{
(T_1, \ldots, T_{l'}) \models_{\times \times} S
}
\sum_{j_{(T_1, \ldots, T_{l'})}} \sqrt{q_{j_{(T_1, \ldots, T_{l'})}}} \\
&   &
~~~~~~~
\ket{\gamma_{j_{(T_1, \ldots, T_{l'})}}}^{A''_{S \cap [k]}} \otimes
\ket{\delta_{j_{(T_1, \ldots, T_{l'})}}}^{A''_{\bar{S} \cap [k]}},
\end{array}
\end{equation}
where 
$
\sum_{
(T_1, \ldots, T_{l'}) \models_{\times \times} S
}
\sum_{j_{(T_1, \ldots, T_{l'})}} q_{j_{(T_1, \ldots, T_{l'})}}
$ 
may be less than one reflecting the fact that the
length of $\cT_{\models_{\times} S, \vecl, \delta}(\ket{h})$ 
may be less than one.
For a pseudosubpartition $(T_1, \ldots, T_{l'}) \models_{\times \times} S$,
the expression 
$
\sum_{j_{(T_1, \ldots, T_{l'})}} \sqrt{q_{j_{(T_1, \ldots, T_{l'})}}} 
\ket{\gamma_{j_{(T_1, \ldots, T_{l'})}}}^{A''_{S \cap [k]}} \otimes
\ket{\delta_{j_{(T_1, \ldots, T_{l'})}}}^{A''_{\bar{S} \cap [k]}}
$
is the Schmidt decomposition  of
\begin{eqnarray*}
\lefteqn{
\frac{1}{\sqrt{N([c] \cupdot [k], \delta)}} 
\sum_{
(T_1, \ldots, T_l) \models_{\times} S,
(T_1, \ldots, T_{l}) \rightthreetimes_S (T_1, \ldots, T_{l'})
}
} \\
&  &
\delta^{l} \, 
(
(
\cT_{T_1, \vecl_{T_1}} \otimes \cdots \otimes 
\cT_{T_l, \vecl_{T_l}} \otimes
\one^{
(\cH \otimes \C^2)^{\otimes ([k] \setminus (T_1 \cup \cdots \cup T_l))}
}
)
(\ket{h})
)^{A''_{[k]}}
\end{eqnarray*}
with respect to $(S \cap [k], \bar{S} \cap [k])$.
We observe that the span of the vectors
$
\{\ket{\gamma_{j_{(T_1, \ldots, T_{l'})}}}^{A''_{S \cap [k]}}\}_{
j_{(T_1, \ldots, T_{l'})}
}
$
is orthogonal to the span of the vectors 
$
\{\ket{\gamma_{j_{(S_1, \ldots, S_{m'})}}}^{A''_{S \cap [k]}}\}_{
j_{(S_1, \ldots, S_{m'})}
}
$
for different pseudosubpartitions 
$
(T_1, \ldots, T_{l'}),
(S_1, \ldots, S_{m'}) \models_{\times \times} S.
$
Similarly, the span of
$
\{\ket{\delta_{j_{(T_1, \ldots, T_{l'})}}}^{A''_{\bar{S} \cap [k]}}\}_{
j_{(T_1, \ldots, T_{l'})}
}
$
is orthogonal to the span of
$
\{\ket{\delta_{j_{(S_1, \ldots, S_{m'})}}}^{A''_{\bar{S} \cap [k]}}\}_{
j_{(S_1, \ldots, S_{m'})}
}.
$
Thus, Equation~\ref{eq:Mdecomposition1} is indeed a valid Schmidt
decomposition for 
$
\cT_{\models_{\times} S, \vecl, \delta}(\ket{h}).
$

Now,
\begin{eqnarray*}
\lefteqn{
(\cT_{[c] \cupdot [k], \vecl, \delta}(\ket{h}))^{A''_{[k]}} 
} \\
& = &
\sqrt{
\frac{N(S \cup [c], \delta) N (\bar{S} \cup [c], \delta)}
     {N([c] \cupdot [k], \delta)}
}
\sum_i \sqrt{p_i} \\
&   &
~~~~~~~~
(
\cT_{S \cup [c], \vecl_{S \cup [c]}, \delta}
(\ket{\alpha_i})
)^{A''_{S \cap [k]}} \otimes
(
\cT_{\bar{S} \cup [c], \vecl_{\bar{S} \cup [c]}, \delta}
(\ket{\beta_i})
)^{A''_{\bar{S} \cap [k]}} \\
&   &
{} +
\sum_{
(T_1, \ldots, T_{l'}) \vdash_{\times \times} S
}
\sum_{j_{(T_1, \ldots, T_{l'})}} \sqrt{q_{j_{(T_1, \ldots, T_{l'})}}} \\
&    &
~~~~~~~~~~~~~~
\ket{\gamma_{j_{(T_1, \ldots, T_{l'})}}}^{A''_{S \cap [k]}} \otimes
\ket{\delta_{j_{(T_1, \ldots, T_{l'})}}}^{A''_{\bar{S} \cap [k]}}.
\end{eqnarray*}
Observe that this is a Schmidt decomposition of 
$
(\cT_{[c] \cupdot [k], \vecl, \delta}(\ket{h}))^{A''_{[k]}},
$
that is, for all $i$, 
$(T_1, \ldots, T_{l'}) \models_{\times \times} S$, 
$j_{(T_1, \ldots, T_{l'})}$,
\[
\cT_{S \cup [c], \vecl_{S \cup [c]}, \delta}(\ket{\alpha_i}) 
\perp
\ket{\gamma_{j_{(T_1, \ldots, T_{l'})}}}
\]
and
$
\cT_{\bar{S} \cup [c], \vecl_{\bar{S} \cup [c]}, \delta}(\ket{\beta_i}) 
\perp
\ket{\delta_{j_{(T_1, \ldots, T_{l'})}}}.
$

Thus,
\begin{eqnarray*}
\lefteqn{
\Tr_{A''_{\bar{S} \cap [k]}} [
(\cT_{[c] \cupdot [k], \vecl, \delta}(\ketbra{h}))^{A''_{[k]}} 
] 
} \\
& = &
\frac{N(S \cup [c], \delta) N (\bar{S} \cup [c], \delta)}
     {N([c] \cupdot [k], \delta)} \\
&   &
\sum_i p_i
(
\cT_{S \cup [c], \vecl_{S \cup [c]}, \delta}
(\ketbra{\alpha_i})
)^{A''_{S \cap [k]}} \\
&   &
{} +
\sum_{
(T_1, \ldots, T_{l'}) \models_{\times \times} S
}
\sum_{j_{(T_1, \ldots, T_{l'})}} q_{j_{(T_1, \ldots, T_{l'})}}
\ketbra{\gamma_{j_{(T_1, \ldots, T_{l'})}}}^{A''_{S \cap [k]}} \\
& = &
\frac{N(S \cup [c], \delta) N (\bar{S} \cup [c], \delta)}
     {N([c] \cupdot [k], \delta)} \\
&   &
\left(
\cT_{S \cup [c], \vecl_{S \cup [c]}, \delta}
\left(
\Tr_{(\cH \otimes \C^2)^{\otimes {\bar{S} \cap [k]}}} [\ketbra{h}]
\right)
\right)^{A''_{S \cap [k]}} \\
&   &
{} +
\sum_{
(T_1, \ldots, T_{l'}) \models_{\times \times} S
}
\sum_{j_{(T_1, \ldots, T_{l'})}} q_{j_{(T_1, \ldots, T_{l'})}}
\ketbra{\gamma_{j_{(T_1, \ldots, T_{l'})}}}^{A''_{S \cap [k]}}.
\end{eqnarray*}

Express 
$\rho_\vecx^{A_{[k]}} \otimes (\ketbra{0}^{\C^2})^{\otimes k}$
in terms of its eigenbasis
\begin{equation}
\label{eq:rhovecx}
\rho_\vecx^{A_{[k]}} \otimes (\ketbra{0}^{\C^2})^{\otimes k} =
\sum_i s_i \ketbra{(h(i))}^{(\cH \otimes \C^2)^{\otimes k}},
\end{equation}
where $\sum_i s_i = 1$.
Let $\ket{\gamma_{j_{(T_1, \ldots, T_{l'})}}(i)}$, 
$q_{j_{(T_1, \ldots, T_{l'})}}(i)$ be the appropriate vectors and
coefficients for $\ket{(h(i))}$ as defined in 
Equation~\ref{eq:Mdecomposition1}.
Then using Equation~\ref{eq:rhoprimedefinition}, we get
\begin{eqnarray*}
\lefteqn{
\Tr_{A''_{\bar{S} \cap [k]}} [
(\rho')_{\vecx,\vecl,\delta}^{A''_{[k]}}
]
} \\
& = &
\frac{N(S \cup [c], \delta) N (\bar{S} \cup [c], \delta)}
     {N([c] \cupdot [k], \delta)} \\
&    &
\left(
\cT_{S \cup [c], \vecl_{S \cup [c]}, \delta}
\left(
\Tr_{A_{\bar{S} \cap [k]}} [
\rho_\vecx^{A_{[k]}} \otimes (\ketbra{0}^{\C^2})^{\otimes k}
]
\right)
\right)^{A''_{S \cap [k]}} \\
&   &
{} +
\sum_i s_i
\sum_{
(T_1, \ldots, T_{l'}) \models_{\times \times} S
}
\sum_{j_{(T_1, \ldots, T_{l'})}} q_{j_{(T_1, \ldots, T_{l'})}}(i) \\
&   &
~~~~~~~~~~~
\ketbra{\gamma_{j_{(T_1, \ldots, T_{l'})}}(i)}^{A''_{S \cap [k]}} \\
& = &
\frac{N(S \cup [c], \delta) N (\bar{S} \cup [c], \delta)}
     {N([c] \cupdot [k], \delta)} \\
&   &
\left(
\cT_{S \cup [c], \vecl_{S \cup [c]}, \delta}
\left(
\rho_\vecx^{A_{S \cap [k]}} \otimes 
(\ketbra{0}^{\C^2})^{\otimes |S \cap [k]|}
\right)
\right)^{A''_{S \cap [k]}} \\
&   &
{} +
(M'')_{S, \vecx, \vecl, \delta}^{A''_{S \cap [k]}},
\end{eqnarray*}
where $(M'')_{S, \vecx, \vecl, \delta}^{A''_{S \cap [k]}}$ is defined
to be the second term in the summation in the first equality above. 
Note that
$(M'')_{S, \vecx, \vecl, \delta}^{A''_{S \cap [k]}}$ is 
a positive semidefinite matrix with support orthogonal to
the support of the first matrix in the summation above.

We say that a pseudosubpartition 
$(T_1, \ldots, T_l) \vdash \vdash S \cup [c]$
{\em leaks out of $S$} if there exists an $i \in [l]$ such that
$T_i \cap \bar{S} \cap [c] \neq \{\}$. We use the notation
$(T_1, \ldots, T_l) \models_{\rightsquigarrow} S$ to denote that 
pseudosubpartition
$(T_1, \ldots, T_l)$ leaks out of $S$.
Observe that 
\begin{eqnarray*}
\lefteqn{
\cT_{S \cup [c], \vecl_{S \cup [c]}, \delta}
} \\
& = &
\sqrt{\frac{N(S, \delta)}{N(S \cup [c], \delta)}}
\cT_{S, \vecl_S, \delta} \\
&   &
\displaystyle
{} +
\frac{1}{\sqrt{N(S \cup [c], \delta)}}
\sum_{
(T_1, \ldots, T_l) \vdash \vdash S \cup [c], 
(T_1, \ldots, T_l) \models_{\rightsquigarrow} S
} \\
&   &
~~~~~~~~~~
\delta^l \,
(
\cT_{T_1, \vecl_{T_1}} \otimes \cdots \otimes  
\cT_{T_l, \vecl_{T_l}} \otimes
\one^{
(\cH \otimes \C^2)^{
\otimes ((S \cap [k]) \setminus (T_1 \cup \cdots \cup T_l))
}
}
).
\end{eqnarray*}
We use $\cT_{\models_{\rightsquigarrow} S, \vecl_{S \cup [c]}, \delta}$ 
to denote the second term in the sum above. The map 
$\cT_{\models_{\rightsquigarrow} S, \vecl_{S \cup [c]}, \delta}$ is 
a scaled
isometric embedding of $(\cH \otimes \C^2)^{\otimes (S \cap [k])}$ into
$A''_{S \cap [k]}$.
Let $\ket{h} \in (\cH \otimes \C^2)^{\otimes (S \cap [k])}$ be
a unit length vector.
Define the Hermitian matrix
\begin{eqnarray*}
\lefteqn{
(N''_{\vecl_{S \cup [c]}, \delta}(\ketbra{h}))^{A''_{S \cap [k]}}
} \\
& := &
(\cT_{S \cup [c], \vecl_{S \cup [c]}, \delta}(\ketbra{h}))^{
A''_{S \cap [k]}
} \\
&   &
{} -
\frac{N(S, \delta)}{N(S \cup [c], \delta)}
(\cT_{S, \vecl_S, \delta}(\ketbra{h}))^{A''_{S \cap [k]}}.
\end{eqnarray*}
Then,
\begin{equation}
\label{eq:leak}
\begin{array}{rcl}
\lefteqn{
(N''_{\vecl_{S \cup [c]}, \delta}(\ketbra{h}))^{A''_{S \cap [k]}} 
} \\
& = &
\displaystyle
\sqrt{\frac{N(S, \delta)}{N(S \cup [c], \delta)}}
(
(
\cT_{S, \vecl_S, \delta}(\ket{h})
\cT_{\models_{\rightsquigarrow} S, \vecl_{S \cup [c]}, \delta}(\bra{h})
)^{A''_{S \cap [k]}} \\
&  &
{} +
(
\cT_{\models_{\rightsquigarrow} S, \vecl_{S \cup [c]}, \delta}(\ket{h})
\cT_{S, \vecl_S, \delta}(\bra{h})
)^{A''_{S \cap [k]}} 
) \\
&   &
{} +
(
\cT_{\models_{\rightsquigarrow} S, \vecl_{S \cup [c]}, \delta}
(\ketbra{h})
)^{A''_{S \cap [k]}}.
\end{array}
\end{equation}

Looking at Equation~\ref{eq:rhovecx}, we define the Hermitian matrix
\begin{equation}
\label{eq:Nprimeprime}
(N''_{S, \vecx, \vecl, \delta})^{A''_{S \cap [k]}} :=
\sum_i s_i
(N''_{\vecl_{S \cup [c]}, \delta}(\ketbra{h(i)}))^{A''_{S \cap [k]}}.
\end{equation}
Note that 
$
(N''_{S, \vecx, \vecl, \delta})^{A''_{S \cap [k]}} = 0
$
if $c = 0$.
Thus,
\begin{eqnarray*}
\lefteqn{
\left(
\cT_{S \cup [c], \vecl_{S \cup [c]}, \delta}
\left(
\rho_\vecx^{A_{S \cap [k]}} \otimes 
(\ketbra{0}^{\C^2})^{\otimes |S \cap [k]|}
\right)
\right)^{A''_{S \cap [k]}} 
} \\
& = &
\frac{N(S, \delta)}{N(S \cup [c], \delta)}
\left(
\cT_{S, \vecl_S, \delta}
\left(
\rho_\vecx^{A_{S \cap [k]}} \otimes 
(\ketbra{0}^{\C^2})^{\otimes |S \cap [k]|}
\right)
\right)^{A''_{S \cap [k]}} \\
&   &
{} + 
(N''_{S, \vecx, \vecl, \delta})^{A''_{S \cap [k]}},
\end{eqnarray*}
and
\begin{eqnarray*}
\lefteqn{
\Tr_{A''_{\bar{S} \cap [k]}} [
(\rho')_{\vecx,\vecl,\delta}^{A''_{[k]}}
]
} \\
& = &
\frac{N(S, \delta) N(\bar{S} \cup [c], \delta)}
     {N([c] \cupdot [k], \delta)} \\
&   &
\left(
\cT_{S, \vecl_S, \delta}
\left(
\rho_\vecx^{A_{S \cap [k]}} \otimes 
(\ketbra{0}^{\C^2})^{\otimes |S \cap [k]|}
\right)
\right)^{A''_{S \cap [k]}} \\
&   &
{} + 
\frac{N(S \cup [c], \delta) N(\bar{S} \cup [c], \delta)}
     {N([c] \cupdot [k], \delta)}
(N''_{S, \vecx, \vecl, \delta})^{A''_{S \cap [k]}} \\
&   &
{} +
(M''_{S, \vecx, \vecl, \delta})^{A''_{S \cap [k]}}.
\end{eqnarray*}
Observe that 
$
(M''_{S, \vecx, \vecl, \delta})^{A''_{S \cap [k]}}
$
has support orthogonal to the sum of the first two terms in the
above equation and
\[
\Tr [(M''_{S, \vecx, \vecl, \delta})^{A''_{S \cap [k]}}] =
1 - 
\frac{N(S \cup [c], \delta) N(\bar{S} \cup [c], \delta)}
     {N([c] \cupdot [k], \delta)}.
\]
Also note that the sum of the first two terms is a positive semidefinite
matrix.

Now define
\begin{equation}
\label{eq:Mprimedefinition}
\begin{array}{rcl}
\lefteqn{
(M')_{S, \vecx_{S \cap [c]}, \vecl_S, \delta}^{A''_{S \cap [k]}} 
} \\
& := &
|\cL|^{-|\bar{S}|} 
\sum_{\vecx'_{[c] \setminus S}, \vecl'_{\bar{S}}} 
p_{[c] \setminus S}(\vecx'_{[c] \setminus S})
(M'')_{S, \vecx_{S \cap [c]} \vecx'_{[c] \setminus S}, 
       \vecl_{S} \vecl'_{\bar{S}}, \delta}^{A''_{S \cap [k]}}, \\
\lefteqn{
(N')_{S, \vecx_{S \cap [c]}, \vecl_S, \delta}^{A''_{S \cap [k]}} 
} \\
& := &
|\cL|^{-|\bar{S}|} 
\frac{N(S \cup [c], \delta) N(\bar{S} \cup [c], \delta)}
     {N([c] \cupdot [k], \delta)}
\sum_{\vecx'_{[c] \setminus S}, \vecl'_{\bar{S}}} 
p_{[c] \setminus S}(\vecx'_{[c] \setminus S}) \\
&   &
~~~~~~~~~~~~~~~~~~~~~~~~~~~~~~~~~~~~~~~~~~
(N'')_{S, \vecx_{S \cap [c]} \vecx'_{[c] \setminus S}, 
       \vecl_{S} \vecl'_{\bar{S}}, \delta}^{A''_{S \cap [k]}}.
\end{array}
\end{equation}
Obviously, 
$
(M')_{S, \vecx_{S \cap [c]}, \vecl_S, \delta}^{A''_{S \cap [k]}}
$
is a positive semidefinite matrix and
$
(N')_{S, \vecx_{S \cap [c]}, \vecl_S, \delta}^{A''_{S \cap [k]}}
$
is a Hermitian matrix. Moreover,
$
(N')_{S, \vecx_{S \cap [c]}, \vecl_S, \delta}^{A''_{S \cap [k]}} = 0
$
if $c = 0$.
Recalling the definition of
$
(\rho')_{\vecx_{S \cap [c]}, \vecl_{S}, \delta}^{A''_{S \cap [k]}} 
$
in Claim~\ref{prop@qtypicalsplitting} of 
Proposition~\ref{prop:cqtypical},
we see that
\begin{eqnarray*}
\lefteqn{
(\rho')_{\vecx_{S \cap [c]}, \vecl_{S}, \delta}^{A''_{S \cap [k]}}
} \\
& = &
\frac{N(S, \delta) N (\bar{S} \cup [c], \delta)}
     {N([c] \cupdot [k], \delta)} \\
&   &
(
\cT_{S, \vecl_S, \delta}(
\rho_{\vecx_{S \cap [c]}}^{A_{S \cap [k]}} \otimes 
(\ketbra{0}^{\C^2})^{\otimes |S \cap [k]|}
)
)^{A''_{S \cap [k]}} \\
&   &
{} + 
(N')_{S, \vecx_{S \cap [c]}, \vecl_S, \delta}^{A''_{S \cap [k]}} +
(M')_{S, \vecx_{S \cap [c]}, \vecl_S, \delta}^{A''_{S \cap [k]}}.
\end{eqnarray*}
Observe that
$
(M')_{S, \vecx_{S \cap [c]}, \vecl_S, \delta}^{A''_{S \cap [k]}}
$
has support orthogonal to that of
the sum of the first two terms in the above equation and
$
\Tr [
(M')_{S, \vecx_{S \cap [c]}, \vecl_S, \delta}^{A''_{S \cap [k]}}
] =
1 - 
\frac{N(S \cup [c], \delta) N(\bar{S} \cup [c], \delta)}
     {N([c] \cupdot [k], \delta)}.
$
Also note that the sum of the first two terms is a positive semidefinite
matrix.

Now recalling the definition of 
$
(\rho')_{\vecx, \vecl,(S_1, \ldots, S_l),\delta}^{A''_{S \cap [k]}} 
$
from Claim~\ref{prop@qtypicalsplitting} of 
Proposition~\ref{prop:cqtypical}, we see that
\begin{equation}
\label{eq:Mprimebreakup}
\begin{array}{rcl}
\lefteqn{
(\rho')_{\vecx, \vecl,(S_1, \ldots, S_l),\delta}^{A''_{[k]}} 
} \\
& = &
\left(
\frac{N(S_1, \delta) N (\overline{S_1} \cup [c], \delta)}
     {N([c] \cupdot [k], \delta)} \cdots
\frac{N(S_l, \delta) N (\overline{S_l} \cup [c], \delta)}
     {N([c] \cupdot [k], \delta)} 
\right. \\
&   &
~~~~~~~~
\cT_{S_1, \vecl_{S_1}, \delta}(
\rho_{\vecx_{S_1 \cap [c]}}^{A_{S_1 \cap [k]}} \otimes 
(\ketbra{0}^{\C^2})^{\otimes |S_1 \cap [k]|}
)
)^{A''_{S_1 \cap [k]}} 
\otimes \cdots \otimes \\
&    &
~~~~~~~~
\cT_{S_l, \vecl_{S_l}, \delta}(
\rho_{\vecx_{S_l \cap [c]}}^{A_{S_l \cap [k]}} \otimes 
(\ketbra{0}^{\C^2})^{\otimes |S_l \cap [k]|}
)
)^{A''_{S_l \cap [k]}} \otimes \\
&    &
~~~~~~~~
\left.
(
\sigma_\vecx^{A_{[k] \setminus (S_1 \cup \cdots \cup S_l)}} \otimes
(\ketbra{0})^{(\C^2)^{\otimes |[k] \setminus (S_1 \cup \cdots \cup S_l)|}}
)
\right) \\
&   &
{} + \mbox{Other Terms I} + \mbox{Other Terms II} \\
& = &
\alpha_{(S_1, \ldots, S_l), \delta}
(
\cT_{(S_1, \ldots, S_l), \vecl, \delta}(
\rho_{\vecx, (S_1, \ldots, S_l)}^{A_{[k]}} \otimes 
(\ketbra{0}^{\C^2})^{\otimes k}
)
)^{A''_{[k]}} \\
&   &
{} +
\mbox{Other Terms I} + \mbox{Other Terms II}. \\
\end{array}
\end{equation}
Above, the notation ``Other Terms II'' denotes a 
$(2^l - 1)$-fold sum of tensor products of
$(l+1)$ matrices where one multiplicand is 
$
\sigma_\vecx^{A_{[k] \setminus (S_1 \cup \cdots \cup S_l)}} \otimes
(\ketbra{0})^{(\C^2)^{\otimes |[k] \setminus (S_1 \cup \cdots \cup S_l)|}}
$,
at least one multiplicand is of the form
$
(M')^{A''_{S_i \cap [k]}}_{
S_i, \vecx_{S_i \cap [c]}, \vecl_{S_i}, \delta
},
$
and the remaining multiplicands are of the form
\begin{eqnarray*}
\lefteqn{
\frac{N(S_j, \delta) N (\overline{S_j} \cup [c], \delta)}
     {N([c] \cupdot [k], \delta)} 
} \\
&  &
(
\cT_{S_j, \vecl_{S_j}, \delta}(
\rho_{\vecx_{S_j \cap [c]}}^{A_{S_j \cap [k]}} \otimes 
(\ketbra{0}^{\C^2})^{\otimes |S_j \cap [k]|}
)
)^{A''_{S_j \cap [k]}} \\
&  &
~~~~~~
{} +
(N')_{S_j, \vecx_{S_j \cap [c]}, \vecl_{S_j}, \delta}^{A''_{S_j \cap [k]}}
\end{eqnarray*}
having $\ell_1$-norm at most one.
We use $(M')^{A''_{[k]}}_{(S_1, \ldots, S_l),\vecx,\vecl,\delta}$ to
denote the ``Other Terms II''. It is clear that 
$
(M')^{A''_{[k]}}_{(S_1, \ldots, S_l),\vecx,\vecl,\delta}
$
is a positive semidefinite matrix with trace 
$
1 
- \alpha_{(S_1, \ldots, S_l), \delta} 
- \beta_{(S_1, \ldots, S_l), \delta}.
$
Define
\begin{equation}
\label{eq:Mdefinition}
M^{A''_{[k]}}_{(S_1, \ldots, S_l),\vecx,\vecl,\delta} :=
\frac{(M')^{A''_{[k]}}_{(S_1, \ldots, S_l),\vecx,\vecl,\delta}}
     {1 
      - \alpha_{(S_1, \ldots, S_l), \delta}
      - \beta_{(S_1, \ldots, S_l), \delta}
     }.
\end{equation}
It is now clear that
$
M^{A''_{[k]}}_{(S_1, \ldots, S_l),\vecx,\vecl,\delta} 
$
is a positive semidefinite matrix with unit trace with support orthogonal
to that of the sum of the first two terms in 
Equation~\ref{eq:Mprimebreakup}.
The notation ``Other Terms I'' denotes a 
$(2^l - 1)$-fold sum of tensor products of
$(l+1)$ matrices where one multiplicand is 
$
\sigma_\vecx^{A_{[k] \setminus (S_1 \cup \cdots \cup S_l)}} \otimes
(\ketbra{0})^{(\C^2)^{\otimes |[k] \setminus (S_1 \cup \cdots \cup S_l)|}}
$,
at least one multiplicand is of the form
$
(N')^{A''_{S_i \cap [k]}}_{
S_i, \vecx_{S_i \cap [c]}, \vecl_{S_i}, \delta
},
$
and the remaining multiplicands are of the form
\[
\frac{N(S_j, \delta) N (\overline{S_j} \cup [c], \delta)}
     {N([c] \cupdot [k], \delta)} 
(
\cT_{S_j, \vecl_{S_j}, \delta}(
\rho_{\vecx_{S_j \cap [c]}}^{A_{S_j \cap [k]}} \otimes 
(\ketbra{0}^{\C^2})^{\otimes |S_j \cap [k]|}
)
)^{A''_{S_j \cap [k]}}
\]
with $\ell_1$-norm at most one.
We use $(N')^{A''_{[k]}}_{(S_1, \ldots, S_l),\vecx,\vecl,\delta}$ to
denote the ``Other Terms I''. It is clear that 
$
(N')^{A''_{[k]}}_{(S_1, \ldots, S_l),\vecx,\vecl,\delta}
$
is a Hermitian matrix with trace 
$
\beta_{(S_1, \ldots, S_l), \delta}.
$
Define
\begin{equation}
\label{eq:Ndefinition}
N^{A''_{[k]}}_{(S_1, \ldots, S_l),\vecx,\vecl,\delta} :=
\begin{array}{l l}
\frac{(N')^{A''_{[k]}}_{(S_1, \ldots, S_l),\vecx,\vecl,\delta}}
     {\beta_{(S_1, \ldots, S_l), \delta}} 
&
\mbox{if $\beta_{(S_1, \ldots, S_l), \delta} = 0$} \\
\frac{\one^{A''_{[k]}}}{|A''_{[k]}|}
&
\mbox{otherwise}
\end{array}.
\end{equation}
It is now clear that
$
N^{A''_{[k]}}_{(S_1, \ldots, S_l),\vecx,\vecl,\delta} 
$
is a Hermitian matrix with unit trace. Also,
$\beta_{(S_1, \ldots, S_l), \delta} = 0$ if $c = 0$.

Thus,
\begin{eqnarray*}
\lefteqn{
(\rho')_{\vecx,\vecl,(S_1, \ldots, S_l),\delta}^{A''_{[k]}} 
} \\
& = &
\alpha_{(S_1, \ldots, S_l), \delta}
(
\cT_{(S_1, \ldots, S_l), \vecl, \delta}(
\rho_{\vecx,(S_1, \ldots, S_l)}^{A_{[k]}} \otimes 
(\ketbra{0}^{\C^2})^{\otimes k}
)
)^{A''_{[k]}} \\
&   &
{} +
\beta_{(S_1, \ldots, S_l), \delta} 
N^{A''_{[k]}}_{(S_1, \ldots, S_l),\vecx,\vecl,\delta} \\
&  &
{} +
(
1 
- \alpha_{(S_1, \ldots, S_l), \delta}
- \beta_{(S_1, \ldots, S_l), \delta}
)
M^{A''_{[k]}}_{(S_1, \ldots, S_l),\vecx,\vecl,\delta}.
\end{eqnarray*}

\subsection{Proving Claim~\arabic{propqtypicalellone}}
\label{subsec:Claim1}
Using Equation~\ref{eq:Piprimedefinition},
we have
\begin{eqnarray*}
\lefteqn{
\ellone{(\Pi')_{\vecx,\vecl,\delta}^{A''_{[k]}}} 
} \\
& \leq &
\ellinfty{
(
\one^{A''_{[k]}} -
(\Pi')^{A''_{[k]}}_{Y'_{\vecx,\vecl, \delta}}
)
}
\ellone{
((\Pi')^{A''_{[k]}}_{(\cH \otimes \C^2)^{\otimes [k]}})
} \\
&    &
\ellinfty{
(
\one^{A''_{[k]}} -
(\Pi')^{A''_{[k]}}_{Y'_{\vecx,\vecl, \delta}}
)
} \\
& \leq &
\ellone{(\Pi')_{(\cH \otimes \C^2)^{\otimes [k]}}^{A''_{[k]}}} 
\; =  \;
|(\cH \otimes \C^2)^{\otimes k}| \\
&  =   &
(2 |\cH|)^k,
\end{eqnarray*}
where we use the fact that
$
\ellinfty{
(
\one^{A''_{[k]}} -
(\Pi')^{A''_{[k]}}_{Y'_{\vecx,\vecl, \delta}}
)
}
\leq 1
$ 
as
$
(
\one^{A''_{[k]}} -
(\Pi')^{A''_{[k]}}_{Y'_{\vecx,\vecl, \delta}}
)
$
is the orthogonal projection onto the complement of the subspace
$
Y'_{\vecx,\vecl, \delta}.
$

\subsection{Proving Claim~\arabic{propqtypicalellinfty}}
\label{subsec:Claim2}
Let $S \subseteq [c] \cupdot [k]$, $\{\} \neq S \cap [k] \neq [k]$.
Let $\ket{h_\vecx(i)} \in (\cH \otimes \C^2)^{\otimes k}$ be the $i$th 
eigenvector of 
$\rho_\vecx^{A_{[k]}} \otimes (\ketbra{0}^{\C^2})^{\otimes k}$.
Fix a subpartition $(T_1, \ldots, T_{l'}) \models_{\times \times} S$
and an index $j_{(T_1, \ldots, T_{l'})}$ 
in the Schmidt decomposition of
$\cT_{\models_\times S, \vecl, \delta}(\ket{h_\vecx(i)})$ with 
respect to $(S \cap [k], \bar{S} \cap [k])$ given in 
Equation~\ref{eq:Mdecomposition1}. 
Define 
\begin{eqnarray*}
\lefteqn{
(M')^{A''_{S \cap [k]}}_{
S,\vecx, \vecl_S,\delta,
(T_1,\ldots,T_{l'}),j_{(T_1, \ldots, T_{l'})}
}(i)
} \\
& := &
|\cL|^{-|\bar{S}|} \sum_{\vecl'_{\bar{S}}}
\ketbra{
\gamma_{j_{(T_1, \ldots, T_{l'})}, \vecx,
\vecl_S \vecl'_{\bar{S}}}(i)
}^{A''_{S \cap [k]}},
\end{eqnarray*}
where we write 
$
\ket{\gamma_{j_{(T_1, \ldots, T_{l'})}, \vecx,
\vecl_S \vecl'_{\bar{S}}}(i)
}
$
in order to emphasise the dependence on $\vecx$ and 
$\vecl := \vecl_S \vecl'_{\bar{S}}$.

From the definition of
$
(M')_{S, \vecx_{S \cap [c]}, \vecl_S, \delta}^{A''_{S \cap [k]}}
$
in Equation~\ref{eq:Mprimedefinition},
it is easy to see that proving 
\[
\ellinfty{
(M')^{A''_{S \cap [k]}}_{
S,\vecx, \vecl_S,\delta,(T_1,\ldots,T_{l'}),j_{(T_1, \ldots, T_{l'})}
}(i)
} \leq \frac{1}{|\cL|}
\]
for all subpartitions $(T_1,\ldots,T_{l'}) \vdash_{\times \times} S$,
indices $j_{(T_1, \ldots, T_{l'})}$ and $i$ suffices to show that
\[
\ellinfty{
(M')_{S, \vecx_{S \cap [c]}, \vecl_S, \delta}^{A''_{S \cap [k]}} 
} \leq 
\frac{1}{|\cL|} 
\Tr [
(M')_{S, \vecx_{S \cap [c]}, \vecl_S, \delta}^{A''_{S \cap [k]}} 
].
\]
Now it is
easy to see using Equations~\ref{eq:Mprimebreakup}, \ref{eq:Mdefinition} 
that this implies that
\begin{eqnarray*}
\ellinfty{(M')_{(S_1, \ldots, S_l), \vecx, \vecl, \delta}^{A''_{[k]}}} 
& \leq &
\frac{1}{|\cL|} \Tr [
(M')_{(S_1, \ldots, S_l), \vecx, \vecl, \delta}^{A''_{[k]}} 
] \\
{} \implies
\ellinfty{M_{(S_1, \ldots, S_l), \vecx, \vecl, \delta}^{A''_{[k]}}} 
& \leq &
\frac{1}{|\cL|}.
\end{eqnarray*}

It only remains to show
$
\ellinfty{
(M')^{A''_{S \cap [k]}}_{
S,\vecx, \vecl_S,\delta,(T_1,\ldots,T_{l'}),j_{(T_1, \ldots, T_{l'})}
}(i)
} \leq \frac{1}{|\cL|}
$
for any subpartition $(T_1,\ldots,T_{l'}) \vdash_{\times \times} S$,
index $j_{(T_1, \ldots, T_{l'})}$ and $i$. In fact, we will prove the
stronger statement that
\[
\ellinfty{
(M')^{A''_{S \cap [k]}}_{
S,\vecx, \vecl_S,\delta,(T_1,\ldots,T_{l'}),j_{(T_1, \ldots, T_{l'})}
}(i)
} \leq 
\frac{1}{|\cL|^{|(T_1 \cup \cdots \cup T_{l'}) \cap \bar{S}|}}.
\]
Since $(T_1 \cup \cdots \cup T_{l'}) \cap \bar{S} \neq \{\}$, this
would complete the proof of the first part of 
Claim~\arabic{propqtypicalellinfty} of Proposition~\ref{prop:cqtypical}.

By triangle inequality, we have
\begin{eqnarray*}
\lefteqn{
\ellinfty{
(M')^{A''_{S \cap [k]}}_{
S,\vecx,\vecl_S,\delta,(T_1,\ldots,T_{l'}),j_{(T_1, \ldots, T_{l'})}
}(i)
} 
} \\
& \leq &
|\cL|^{-|\bar{S} \cap (T_1 \cup \cdots T_{l'})|} 
|\cL|^{-|\bar{S} \setminus (T_1 \cup \cdots T_{l'})|} 
\sum_{\vecl'_{\bar{S} \setminus (T_1 \cup \cdots T_{l'})}} \\
&      &
~~~~
\left\|
\sum_{\vecl'_{\bar{S} \cap (T_1 \cup \cdots T_{l'})}} 
\ket{
\gamma_{
j_{T_1, \ldots, T_{l'})}, 
\vecx,
\vecl_S 
\vecl'_{\bar{S} \setminus (T_1 \cup \cdots T_{l'})}
\vecl'_{\bar{S} \cap (T_1 \cup \cdots T_{l'})}
}(i)
} 
\right. \\
&   &
~~~~~~~~~~~~~~~~~~~~~~~~
\left.
\bra{
\gamma_{
j_{T_1, \ldots, T_{l'})}, 
\vecx,
\vecl_S 
\vecl'_{\bar{S} \setminus (T_1 \cup \cdots T_{l'})}
\vecl'_{\bar{S} \cap (T_1 \cup \cdots T_{l'})}
}(i)
}^{A''_{S \cap [k]}}
\right\|_\infty \\
&   =  &
|\cL|^{-|\bar{S} \cap (T_1 \cup \cdots T_{l'})|} 
|\cL|^{-|\bar{S} \setminus (T_1 \cup \cdots T_{l'})|} 
\sum_{\vecl'_{\bar{S} \setminus (T_1 \cup \cdots T_{l'})}} 
1 \\
&   =  &
|\cL|^{-|\bar{S} \cap (T_1 \cup \cdots T_{l'})|},
\end{eqnarray*}
where we use the fact that 
\begin{eqnarray*}
\lefteqn{
\ket{
\gamma_{
j_{(T_1, \ldots, T_l)},
\vecx,
\vecl_S 
\vecl'_{\bar{S} \setminus (S_1 \cup \cdots \cup S_{l'})}
\vecl'_{\bar{S} \cap (T_1 \cup \cdots T_{l'})}
}(i)
}
} \\
& \perp &
\ket{
\gamma_{
j_{(T_1, \ldots, T_l)},
\vecx,
\vecl_S 
\vecl'_{\bar{S} \setminus (S_1 \cup \cdots \cup S_{l'})}
\vecl''_{\bar{S} \cap (T_1 \cup \cdots T_{l'})}
}(i)
}
\end{eqnarray*}
for 
$
\vecl'_{\bar{S} \cap (T_1 \cup \cdots T_{l'})} \neq
\vecl''_{\bar{S} \cap (T_1 \cup \cdots T_{l'})}
$
in the equality above. This follows from the observation that
for distinct computational basis vectors 
$\vecl'_{\bar{S} \cap (T_1 \cup \cdots T_{l'})}$,
$\vecl''_{\bar{S} \cap (T_1 \cup \cdots T_{l'})}$,
there exists an $i \in [l']$, a coordinate $a \in \bar{S} \cap T_i$ such
that $\vecl'_a \neq \vecl''_a$
which implies that
$
\ket{
\gamma_{
j_{(T_1, \ldots, T_l)},
\vecx,
\vecl_S 
\vecl'_{\bar{S} \setminus (S_1 \cup \cdots \cup S_{l'})}
\vecl'_{\bar{S} \cap (T_1 \cup \cdots T_{l'})}
}(i)
},
$ \\
$
\ket{
\gamma_{
j_{(T_1, \ldots, T_l)},
\vecx,
\vecl_S 
\vecl'_{\bar{S} \setminus (S_1 \cup \cdots \cup S_{l'})}
\vecl''_{\bar{S} \cap (T_1 \cup \cdots T_{l'})}
}(i)
}
$
lie in the orthogonal subspaces 
$
(
(\cH \otimes \C^2) \otimes 
\ket{\vecl_{T_i \cap S} \vecl'_{T_i \cap \bar{S}}}
)_b \otimes 
A''_{(S \cap [k]) \setminus \{b\}},
$ 
$
(
(\cH \otimes \C^2) \otimes 
\ket{\vecl_{T_i \cap S} \vecl''_{T_i \cap \bar{S}}}
)_b \otimes 
A''_{(S \cap [k]) \setminus \{b\}},
$
$b \in S \cap [k] \cap T_i$,
where the first multiplicands in the two tensor products are embedded
into $A''_b$.

This completes the proof of the first part of 
Claim~\arabic{propqtypicalellinfty} of
Proposition~\ref{prop:cqtypical}.

Again, let 
$S \subseteq [c] \cupdot [k]$, $\{\} \neq S \cap [k] \neq [k]$.
If $[c] \subset S$, then 
$
N''_{\vecl_{S \cup [c]}, \delta}(\ketbra{h}) = 0.
$
Suppose $S \cap [c] \neq [c]$.
Suppose one were to show that
\[
\ellinfty{
\sum_{\vecl'_{[c] \setminus S}}
(
N''_{\vecl_{S} \vecl'_{[c] \setminus S}, \delta}(\ketbra{h})
)^{A''_{S \cap [k]}}
} \leq 3 |\cL|^{|[c] \setminus S| - 1/2}
\]
for all unit length vectors 
$\ket{h} \in (\cH \otimes \C^2)^{\otimes (S \cap [k])}$. 
From Equations~\ref{eq:Nprimeprime}, \ref{eq:Mprimedefinition} and the
triangle inequality, in
either case, it will follow that
\[
\ellinfty{
(N')_{S, \vecx_{S \cap [c]}, \vecl_S, \delta}^{A''_{S \cap [k]}} 
} \leq 
\frac{N(S \cup [c], \delta) N (\overline{S} \cup [c], \delta)}
     {N([c] \cupdot [k], \delta)}
\cdot
\frac{3}{\sqrt{|\cL|}}.
\]
Now it is
easy to see using Equations~\ref{eq:Mprimebreakup}, \ref{eq:Ndefinition} 
that this implies that
\[
\beta_{(S_1, \ldots, S_l), \delta}
\ellinfty{(N')_{(S_1, \ldots, S_l), \vecx, \vecl, \delta}^{A''_{[k]}}} 
\leq
\frac{3}{\sqrt{|\cL|}}.
\]

It only remains to show for $S \cap [c] \neq [c]$ that
\[
\ellinfty{
\sum_{\vecl'_{[c] \setminus S}}
(
N''_{\vecl_{S} \vecl'_{[c] \setminus S}, \delta}(\ketbra{h})
)^{A''_{S \cap [k]}}
} \leq 3 |\cL|^{|[c] \setminus S| - 1/2}
\]
for any unit length vector 
$\ket{h} \in (\cH \otimes \C^2)^{\otimes (S \cap [k])}$. By 
Equation~\ref{eq:leak}, it suffices to show that
\[
\elltwo{
\sum_{\vecl'_{[c] \setminus S}}
\cT_{\models_\rightsquigarrow S, \vecl_S \vecl'_{[c] \setminus S}, \delta}
(\ket{h})
} \leq |\cL|^{|[c] \setminus S| - 1/2},
\]
\[
\ellinfty{
\sum_{\vecl'_{[c] \setminus S}}
\cT_{\models_\rightsquigarrow S, \vecl_S \vecl'_{[c] \setminus S}, \delta}
(\ketbra{h})
} \leq |\cL|^{|[c] \setminus S| - 1}.
\]
Since the range spaces of the summands in the definition of
$\cT_{\models_\rightsquigarrow S, \vecl_{S \cup [c]} , \delta}$
are orthogonal, it suffices to show, for any
$(T_1, \ldots, T_l) \vdash \vdash S \cup [c]$, 
$(T_1, \ldots, T_l) \models_\rightsquigarrow S$,
$\vecl_S$, 
$\vecl'_{[c] \setminus (S \cup T_1 \cup \cdots \cup T_l)}$
that 
\begin{eqnarray*}
\lefteqn{
\left\|
\sum_{\vecl''_{([c] \setminus S) \cap (T_1 \cup \cdots \cup T_l)}}
(
\cT_{T_1, \vecl_{T_1}} 
\otimes \cdots \otimes
\cT_{T_l, \vecl_{T_l}} 
\right.
} \\
&    &
~~~
\left.
{} \otimes
\one^{(\cH \otimes \C^2)^{
\otimes ((S \cap [k]) \setminus (T_1 \cup \cdots \cup T_l))
}
}
)
(\ket{h})
\right\|_2 \\
& = &
\sqrt{|\cL|^{|([c] \setminus S) \cap (T_1 \cup \cdots \cup T_l)|}}
\end{eqnarray*}
and
\begin{eqnarray*}
\lefteqn{
\left\|
\sum_{\vecl''_{([c] \setminus S) \cap (T_1 \cup \cdots \cup T_l)}}
(
\cT_{T_1, \vecl_{T_1}} 
\otimes \cdots \otimes
\cT_{T_l, \vecl_{T_l}} 
\right. 
} \\
&   &
\left.
~~~
{} \otimes
\one^{(\cH \otimes \C^2)^{
\otimes ((S \cap [k]) \setminus (T_1 \cup \cdots \cup T_l))
}
}
)
(\ketbra{h})
\right\|_\infty \\
&  = &
1.
\end{eqnarray*}
The last two equalities arise from the fact that for any two 
distinct computational basis vectors
\[
\vecl''_{([c] \setminus S) \cap (T_1 \cup \cdots \cup T_l)} \neq
\vecl'''_{([c] \setminus S) \cap (T_1 \cup \cdots \cup T_l)} 
\]
the range spaces of
\[
\cT_{T_1, \vecl_{T_1}} 
\otimes \cdots \otimes
\cT_{T_l, \vecl_{T_l}} \otimes
\one^{(\cH \otimes \C^2)^{
\otimes ((S \cap [k]) \setminus (T_1 \cup \cdots \cup T_l))
}
}
\]
are orthogonal.
This follows from the observation that
there exists an $i \in [l]$, a coordinate 
$a \in [c] \cap \bar{S} \cap T_i$ such
that $\vecl''_a \neq \vecl'''_a$
which implies that the two range spaces embed into the orthogonal spaces
$
(
(\cH \otimes \C^2) \otimes 
\ket{
\vecl_{T_i \cap S} 
\vecl''_{T_i \cap \bar{S} \cap [c]}
}
)_b \otimes 
A''_{(S \cap [k]) \setminus \{b\}},
$
$
(
(\cH \otimes \C^2) \otimes 
\ket{
\vecl_{T_i \cap S} 
\vecl'''_{T_i \cap \bar{S} \cap [c]}
}
)_b \otimes 
A''_{(S \cap [k]) \setminus \{b\}},
$
$b \in S \cap [k] \cap T_i$,
where the first multiplicands in the two tensor products are embedded
into $A''_b$.

This completes the proof of the second part of 
Claim~\arabic{propqtypicalellinfty} of
Proposition~\ref{prop:cqtypical}.

\subsection{Proving Claim~\arabic{propqtypicaldistance}}
\label{subsec:Claim3}
Let $\ket{h}$ be a unit length vector in $(\cH \otimes \C^2)^{\otimes k}$.
From Equation~\ref{eq:isometricembedding} and 
Inequality~\ref{eq:increasedlength1}, we get
\[
\braket{\cT_{[c] \cupdot [k], \vecl, \delta}(\ket{h})}{h} = 
\frac{1}{\sqrt{N([c] \cupdot [k], \delta)}} >
e^{-\delta^2 2^{c+k-1}}.
\]
Thus,
\begin{eqnarray*}
\lefteqn{
\ellone{\cT_{[c] \cupdot [k], \vecl, \delta}(\ketbra{h}) - \ketbra{h}} 
} \\
& \leq &
2 \elltwo{\cT_{[c] \cupdot [k], \vecl, \delta}(\ket{h}) - \ket{h}} \\
& \leq &
2 \sqrt{2 - 2 e^{-\delta^2 2^{c+k-1}}}
\; < \;
2^{\frac{c+k}{2} + 1} \delta,
\end{eqnarray*}
where we used the fact that 
$e^{-x} \geq 1 - x$ in the last inequality above.
Recalling Equation~\ref{eq:rhoprimedefinition} and
applying the above inequality to the eigenvectors of 
$\rho^{A_{[k]}} \otimes (\ketbra{0}^{\C^2})^{\otimes k}$ allows us
to prove the desired Claim~\arabic{propqtypicaldistance} of
Proposition~\ref{prop:cqtypical}.

\subsection{Proving Claim~\arabic{propqtypicalcompleteness}}
\label{subsec:Claim4}
Let $\ket{h}$ be an eigenvector of 
$\rho^{A_{[k]}} \otimes (\ketbra{0}^{\C^2})^{\otimes k}$.
For a pseudosubpartition 
$(T_1, \ldots, T_m) \vdash \vdash [c] \cupdot [k]$, define 
\[
\epsilon_{\vecx, (T_1, \ldots, T_m)}(\ket{h}) :=
\elltwo{
\Pi^{(\cH \otimes \C^2)^{\otimes [k]}}_{Y_{\vecx, (T_1, \ldots, T_m)}} \ket{h}
}^2 
\]
where the subspace 
$Y_{\vecx, (T_1, \ldots, T_m)} \leq (\cH \otimes \C^2)^{\otimes [k]}$ 
is defined just after Equation~\ref{eq:optimalPi} above.
Recall that the subspace $Y'_{\vecx, \vecl,\delta}$ defined in
Equation~\ref{eq:Ydefinition} is the $A$-tilted span of 
$
\{Y_{\vecx, (T_1, \ldots, T_m)}:
(T_1, \ldots, T_m) \vdash \vdash [c] \cupdot [k], m > 0
\}
$
where the tilting matrix $A$ is defined in 
Equation~\ref{eq:Adefinition} above.
Recall that $A$ is upper triangular, substochastic and 
diagonal dominated, and
the diagonal entries of $A$ satisfy
\begin{eqnarray*}
\lefteqn{
A_{(T_1, \ldots, T_m), (T_1, \ldots, T_m)}
} \\
& = &
\frac{\delta^{2m}}{N(T_1, \delta) \cdots N(T_m, \delta)} \\
& \geq &
\frac{\delta^{2k}}{N([c] \cupdot [k], \delta)} 
\;\geq\;
e^{-\delta^2 2^{c+k}} \delta^{2k},
\end{eqnarray*}
for all $(T_1, \ldots, T_m) \vdash \vdash [c] \cupdot [k]$.
By Proposition~\ref{prop:generalisedtiltedspan},
\begin{eqnarray*}
\lefteqn{
\elltwo{
(\Pi')_{Y'_{\vecx, \vecl,\delta}}^{A''_{[k]}} \ket{h}
}^2
} \\
& \leq &
e^{\delta^2 2^{c+k}} \delta^{-2k} 2^{2^{ck+1} (k+1)^k + 1}  
\sum_{
(T_1, \ldots, T_m) \vdash \vdash [c] \cupdot [k], m > 0
}
\epsilon_{\vecx, (T_1, \ldots, T_m)}(\ket{h}).
\end{eqnarray*}
Using Equation~\ref{eq:optimalPi} and
applying this inequality to the eigenvectors of
$\rho^{A_{[k]}} \otimes (\ketbra{0}^{\C^2})^{\otimes k}$,
we get
\begin{eqnarray*}
\lefteqn{
\Tr [
(\Pi')_{Y'_{\vecx,\vecl,\delta}}^{A''_{[k]}} 
(
\rho^{A_{[k]}} \otimes (\ketbra{0}^{\C^2})^{\otimes k}
)
]
} \\
& \leq &
e^{\delta^2 2^{c+k}} \delta^{-2k} 2^{2^{ck+1} (k+1)^k + 1}  
\sum_{
(T_1, \ldots, T_m) \vdash \vdash [c] \cupdot [k], m > 0
}
\epsilon_{\vecx, (T_1, \ldots, T_m)}  \\
& <  &
\delta^{-2k} 2^{2^{ck+3} (k+1)^k}
\epsilon_{\vecx}.
\end{eqnarray*}
Finally, using Fact~\ref{fact:noncommutativeunionbound} and
Equation~\ref{eq:Piprimedefinition}, we get
\begin{eqnarray*}
\lefteqn{
\Tr [
(\Pi')_{\vecx, \vecl,\delta}^{A''_{[k]}}
(\rho^{A_{[k]}} \otimes (\ketbra{0}^{\C^2})^{\otimes k})
]
} \\
&   =  &
\Tr [
(\Pi')^{A''_{[k]}}_{(\cH \otimes \C^2)^{\otimes [k]}}
(
\one^{A''_{[k]}} -
(\Pi')^{A''_{[k]}}_{Y'_{\vecx, \vecl, \delta}}
) \\
&      &
~~~~~~~~~~
(\rho^{A_{[k]}} \otimes (\ketbra{0}^{\C^2})^{\otimes k})
(
\one^{A''_{[k]}} -
(\Pi')^{A''_{[k]}}_{Y'_{\vecx,\vecl, \delta}}
) \\
&      &
~~~~~~~~~~~~~~~~
(\Pi')^{A''_{[k]}}_{(\cH \otimes \C^2)^{\otimes [k]}}
] \\ 
& \geq &
1 - 
4 (
\Tr [
(\Pi')^{A''_{[k]}}_{Y'_{\vecx,\vecl, \delta}}
(\rho^{A_{[k]}} \otimes (\ketbra{0}^{\C^2})^{\otimes k})
] \\
&    &
{} +
1 -
\Tr [
(\Pi')^{A''_{[k]}}_{(\cH \otimes \C^2)^{\otimes [k]}}
(\rho^{A_{[k]}} \otimes (\ketbra{0}^{\C^2})^{\otimes k})
]
) \\
&  >   &
1 - 
4 (
\delta^{-2k} 2^{2^{ck+3} (k+1)^k} 
\epsilon_{\vecx}
) 
\; \geq \;
1 - 
\delta^{-2k} 2^{2^{ck+4} (k+1)^k} 
\epsilon_{\vecx}.
\end{eqnarray*}

\subsection{Proving Claim~\arabic{propqtypicalsplitting}}
\label{subsec:Claim5}
This claim and the accompanying claims of orthogonality and 
vanishing of $\beta_{(S_1, \ldots, S_l),\delta}$ are proved at
the end of \ref{subsec:MN} above.

\subsection{Proving Claim~\arabic{propqtypicalsoundness}}
\label{subsec:Claim6}
Using Equations~\ref{eq:Piprimedefinition}, \ref{eq:Ydefinition},
\ref{eq:Yprimedefinition}, \ref{eq:optimalPi}, we get
\begin{eqnarray*}
\lefteqn{
\Tr [
(\Pi')_{\vecx,\vecl,\delta}^{A''_{[k]}}
(
\cT_{(S_1, \ldots S_l), \vecl, \delta}
(
\rho_{\vecx,(S_1, \ldots, S_l)}^{A_{[k]}}
\otimes (\ketbra{0}^{\C^2})^{\otimes k}
)
)^{A''_{[k]}}
] 
} \\
&   =  &
\Tr [
(\Pi')^{A''_{[k]}}_{(\cH \otimes \C^2)^{\otimes k}}
(
\one^{A''_{[k]}} -
(\Pi')^{A''_{[k]}}_{Y'_{\vecx,\vecl, \delta}}
) \\
&     &
~~~~~~~~~~
(
\cT_{(S_1, \ldots S_l), \vecl, \delta}
(
\rho_{\vecx,(S_1, \ldots, S_l)}^{A_{[k]}}
\otimes (\ketbra{0}^{\C^2})^{\otimes k}
)
)^{A''_{[k]}} \\
&       &
~~~~~~~~~~~~~~~~~~~~~~~
(
\one^{A''_{[k]}} -
(\Pi')^{A''_{[k]}}_{Y'_{\vecx,\vecl, \delta}}
)
(\Pi')^{A''_{[k]}}_{(\cH \otimes \C^2)^{\otimes k}}
] \\
& \leq &
\Tr [
(
\one^{A''_{[k]}} -
(\Pi')^{A''_{[k]}}_{Y'_{\vecx,\vecl, \delta}}
) \\
&      &
~~~~~~~~~~
(
\cT_{(S_1, \ldots S_l), \vecl, \delta}
(
\rho_{\vecx,(S_1, \ldots, S_l)}^{A_{[k]}}
\otimes (\ketbra{0}^{\C^2})^{\otimes k}
)
)^{A''_{[k]}} \\
&      &
~~~~~~~~~~~~~~
(
\one^{A''_{[k]}} -
(\Pi')^{A''_{[k]}}_{Y'_{\vecx,\vecl, \delta}}
)
] \\
& \leq &
\Tr [
(
\one^{A''_{[k]}} -
\Pi^{A''_{[k]}}_{
\cT_{(S_1,\ldots,S_l),\vecl,\delta}(Y_{\vecx,(S_1, \ldots, \cup S_l)}) 
}
) \\
&       &
~~~~~~~~~~
(
\cT_{(S_1, \ldots S_l), \vecl, \delta}
(
\rho_{\vecx,(S_1, \ldots, S_l)}^{A_{[k]}}
\otimes (\ketbra{0}^{\C^2})^{\otimes k}
)
)^{A''_{[k]}} \\
&      &
~~~~~~~~~~~~~~
(
\one^{A''_{[k]}} -
\Pi^{A''_{[k]}}_{
\cT_{(S_1,\ldots,S_l),\vecl,\delta}(Y_{\vecx,(S_1, \ldots, \cup S_l)}) 
}
)
] \\
&   =  &
\Tr [
(
\one^{A''_{[k]}} -
\Pi^{A''_{[k]}}_{
\cT_{(S_1,\ldots,S_l),\vecl,\delta}(Y_{\vecx,(S_1, \ldots, \cup S_l)}) 
}
)
\Pi^{A''_{[k]}}_{
\cT_{(S_1,\ldots,S_l),\vecl,\delta}((\cH \otimes \C^2)^{\otimes k})
} \\
&      &
~~~~~~
(
\cT_{(S_1, \ldots S_l), \vecl, \delta}
(
\rho_{\vecx,(S_1, \ldots, S_l)}^{A_{[k]}}
\otimes (\ketbra{0}^{\C^2})^{\otimes k}
)
)^{A''_{[k]}} \\
&      &
~~~~~~
\Pi^{A''_{[k]}}_{
\cT_{(S_1,\ldots,S_l),\vecl,\delta}((\cH \otimes \C^2)^{\otimes k})
}
(
\one^{A''_{[k]}} -
\Pi^{A''_{[k]}}_{
\cT_{(S_1,\ldots,S_l),\vecl,\delta}(Y_{\vecx,(S_1, \ldots, \cup S_l)}) 
}
)
] \\
&   =  &
\Tr [
(
\Pi^{A''_{[k]}}_{
\cT_{(S_1,\ldots,S_l),\vecl,\delta}((\cH \otimes \C^2)^{\otimes k})
} - 
\Pi^{A''_{[k]}}_{
\cT_{(S_1,\ldots,S_l),\vecl,\delta}(Y_{\vecx,(S_1, \ldots, \cup S_l)}) 
}
) \\
&      &
~~~~~~~~~~
(
\cT_{(S_1, \ldots S_l), \vecl, \delta}
(
\rho_{\vecx,(S_1, \ldots, S_l)}^{A_{[k]}}
\otimes (\ketbra{0}^{\C^2})^{\otimes k}
)
) \\
&      &
~~~~~~~~~~
(
\Pi^{A''_{[k]}}_{
\cT_{(S_1,\ldots,S_l),\vecl,\delta}((\cH \otimes \C^2)^{\otimes k})
} - 
\Pi^{A''_{[k]}}_{
\cT_{(S_1,\ldots,S_l),\vecl,\delta}(Y_{\vecx,(S_1, \ldots, S_l)}) 
}
)
] \\
&   =  &
\Tr [
\cT_{(S_1, \ldots S_l), \vecl, \delta}
(
(
\one^{(\cH \otimes \C^2)^{\otimes k}} -
\Pi^{(\cH \otimes \C^2)^{\otimes k}}_{
Y_{\vecx,(S_1, \ldots, \cup S_l)} 
}
) \\
&      &
~~~~~~~~~~~~~~~~~~~~~~~~~~~~~~
(
\rho_{\vecx,(S_1, \ldots, S_l)}^{A_{[k]}}
\otimes (\ketbra{0}^{\C^2})^{\otimes k}
) \\
&       &
~~~~~~~~~~~~~~~~~~~~~~~~~~~~~~~~~~
(
\one^{(\cH \otimes \C^2)^{\otimes k}} -
\Pi^{(\cH \otimes \C^2)^{\otimes k}}_{
Y_{\vecx,(S_1, \ldots, \cup S_l)} 
}
) 
)
] \\
&   =  &
\Tr [
(
\one^{(\cH \otimes \C^2)^{\otimes k}} -
\Pi^{(\cH \otimes \C^2)^{\otimes k}}_{
Y_{\vecx,(S_1, \ldots, \cup S_l)} 
}
) \\
&     &
~~~~~~~~~~~
(
\rho_{\vecx,(S_1, \ldots, S_l)}^{A_{[k]}}
\otimes (\ketbra{0}^{\C^2})^{\otimes k}
) \\
&     &
~~~~~~~~~~~~~~~
(
\one^{(\cH \otimes \C^2)^{\otimes k}} -
\Pi^{(\cH \otimes \C^2)^{\otimes k}}_{
Y_{\vecx,(S_1, \ldots, \cup S_l)} 
}
) 
] \\
&   =  &
\Tr [
\Pi^{(\cH \otimes \C^2)^{\otimes k}}_{\vecx, (S_1,\ldots,S_l)}
(
\rho_{\vecx,(S_1, \ldots, S_l)}^{A_{[k]}}
\otimes (\ketbra{0}^{\C^2})^{\otimes k}
) 
\Pi^{(\cH \otimes \C^2)^{\otimes k}}_{\vecx,(S_1,\ldots,S_l)}
] \\
& \leq &
2^{-D_H^{\epsilon_{\vecx,(S_1, \ldots, S_l)}}
    (
     \rho_\vecx^{A_{[k]}} \| 
     \rho_{\vecx, (S_1, \ldots, S_l)}^{A_{[k]}}
    )
}.
\end{eqnarray*}
In the second inequality above, we used the fact that
\[
\cT_{(S_1,\ldots,S_l),\vecl,\delta}
(Y_{\vecx, (S_1, \ldots, \cup S_l)}) \leq
Y'_{\vecx, \vecl,\delta}.
\]
In the second equality above, we used the property that
the support of
$
\cT_{S_1, \ldots S_l), \vecl, \delta}
(
\rho_{\vecx, (S_1, \ldots, S_l)}^{A_{[k]}}
\otimes (\ketbra{0}^{\C^2})^{\otimes k}
)
$
lies in the vector space
$
\cT_{(S_1,\ldots,S_l),\vecl,\delta}((\cH \otimes \C^2)^{\otimes k}).
$
We used the property that the map
$
\cT_{(S_1, \ldots S_l), \vecl, \delta}
$
is an isometric embedding of $(\cH \otimes \C^2)^{\otimes k})$ in
the fourth and fifth equalities. Finally, we used the definition
of 
$
Y_{\vecx,(S_1, \ldots, \cup S_l)}
$
given just below Equation~\ref{eq:optimalPi}
for obtaining the last equality.

This completes the proof of Claim~\arabic{propqtypicalsoundness}
of Proposition~\ref{prop:cqtypical} and thus finishes the proof of
Proposition~\ref{prop:cqtypical}.

\section{Proof of Equation~\ref{eq:cqMAC}}
\label{sec:proofcqMAC}
The probability of incorrectly decoding the sent message pair 
$(m_1,m_2)$, for a given codebook
$\cC$, is upper bounded by
\begin{eqnarray*}
\lefteqn{
p_e(\cC; m_1, m_2) 
} \\
& \leq &
\sum_z p(z | x(m_1), y(m_2)) \\
&      &
~~~~
\sum_{
(\hat{m}_1, \hat{m}_2):
(\hat{m}_1, \hat{m}_2) \prec (m_1, m_2)
} 
f(x(\hat{m}_1), y(\hat{m}_2), z) \\
&      &
{} +
\sum_z p(z | x(m_1), y(m_2))
(1 - f(x(m_1), y(m_2), z)) \\
& \leq &
\sum_z p(z | x(m_1), y(m_2)) \\
&      &
~~~~
\sum_{
(\hat{m}_1, \hat{m}_2):
(\hat{m}_1, \hat{m}_2) \neq (m_1, m_2)
} 
f(x(\hat{m}_1), y(\hat{m}_2), z) \\
&      &
{} +
\sum_z p(z | x(m_1), y(m_2))
(1 - f(x(m_1), y(m_2), z)) \\
&   =  &
\sum_z p(z | x(m_1), y(m_2))
\sum_{
(\hat{m}_1, \hat{m}_2):
\hat{m}_1 \neq m_1, \hat{m}_2 \neq m_2
} 
f(x(\hat{m}_1), y(\hat{m}_2), z) \\
&      &
{} +
\sum_z p(z | x(m_1), y(m_2))
\sum_{
\hat{m}_2:
\hat{m}_2 \neq  m_2
} 
f(x(m_1), y(\hat{m}_2), z) \\
&      &
{} +
\sum_z p(z | x(m_1), y(m_2))
\sum_{
\hat{m}_1:
\hat{m}_1 \neq  m_1
} 
f(x(\hat{m}_1), y(m_2), z) \\
&      &
{} +
\sum_z p(z | x(m_1), y(m_2))
(1 - f(x(m_1), y(m_2), z)).
\end{eqnarray*}
The expectation, over the choice of the random codebook $\cC$, of 
the decoding error is then upper bounded by  
\begin{eqnarray*}
\lefteqn{
\E_{\cC}[
p_e(\cC; m_1, m_2) 
]
} \\
& \leq &
(2^{R_1} - 1) (2^{R_2} - 1) \\
&      &
\sum_{x,y,z} p(x) p(y) p(z | x, y)
\sum_{x', y'} p(x') p(y') f(x', y', z) \\
&      &
{} +
(2^{R_2} - 1)
\sum_{x,y,z} p(x) p(y) p(z | x, y)
\sum_{y'} p(y') f(x, y', z) \\
&      &
{} +
(2^{R_1} - 1)
\sum_{x,y,z} p(x) p(y) p(z | x, y)
\sum_{x'} p(x') f(x', y, z) \\
&      &
{} +
\sum_{x,y,z} p(x) p(y) p(z | x, y)
(1 - f(x, y, z)) \\
& \overset{a}{\leq} &
(2^{R_1} - 1) (2^{R_2} - 1) \\
&      &
\sum_{x,y,z} p(x) p(y) p(z | x, y)
\sum_{x', y'} p(x') p(y') f^{X,Y}(x', y', z) \\
&      &
{} +
(2^{R_2} - 1)
\sum_{x,y,z} p(x) p(y) p(z | x, y)
\sum_{y'} p(y') f^Y(x, y', z) \\
&      &
{} +
(2^{R_1} - 1)
\sum_{x,y,z} p(x) p(y) p(z | x, y)
\sum_{x'} p(x') f^X(x', y, z) \\
&      &
{} +
\sum_{x,y,z} p(x) p(y) p(z | x, y)
(
(1 - f^X(x, y, z)) \\
&       &
~~~~~~~~~~~~~~~~~~
{} +
(1 - f^Y(x, y, z)) +
(1 - f^{X,Y}(x, y, z))
) \\
&   =  &
(2^{R_1} - 1) (2^{R_2} - 1)
\sum_{x', y', z} p(x') p(y') p(z) f^{X,Y}(x', y', z) \\
&      &
{} +
(2^{R_2} - 1)
\sum_{x, y', z} p(y') p(x) p(z | x) f^Y(x, y', z) \\
&      &
{} +
(2^{R_1} - 1)
\sum_{x', y, z} p(x') p(y) p(z | y) f^X(x', y, z) \\
&      &
{} +
\sum_{x,y,z} p(x) p(y) p(z | x, y)
(
(1 - f^X(x, y, z)) \\
&      &
~~~~~~~~~~~~~~~~~~
{} +
(1 - f^Y(x, y, z)) +
(1 - f^{X,Y}(x, y, z))
) \\
& \overset{b}{\leq} &
2^{R_1 + R_2} 2^{-I^\epsilon_H(X Y : Z)} +
2^{R_2} 2^{-I^\epsilon_H(Y : X Z)} +
2^{R_1} 2^{-I^\epsilon_H(X : Y Z)} +
3 \epsilon.
\end{eqnarray*}

\balance 

Above, we use the property that
\begin{eqnarray*}
f(x,y,z) 
& \leq &
\{f^X(x,y,z), f^Y(x,y,z), f^{X,Y}(x,y,z)\}, \\
1 - f(x,y,z) 
& \leq &
(1 - f^X(x,y,z)) +
(1 - f^Y(x,y,z)) \\
&     &
{} +
(1 - f^{X,Y}(x,y,z))
\end{eqnarray*}
for all triples $(x,y,z)$ in Step~(a), and Equation~\ref{eq:cmacDeps}
in Step~(b).

\end{document}